\documentclass{article}
\usepackage{ijcai24}

\usepackage{times}
\usepackage{soul}
\usepackage{url}
\usepackage[hidelinks]{hyperref}
\usepackage[utf8]{inputenc}
\usepackage[small]{caption}
\usepackage{graphicx}
\usepackage{amsmath}
\usepackage{amsthm}
\usepackage{booktabs}
\usepackage{algorithm}
\usepackage{algorithmic}
\usepackage[switch]{lineno}

\usepackage[dvipsnames]{xcolor}
\usepackage[backgroundcolor=white, textsize=small]{todonotes}
\usepackage{amssymb, bm, tikz}
\usepackage{mathtools,amsfonts,nicefrac,subcaption}

\usepackage{defs}

\title{Cooperation and Control in Delegation Games}

\author{
    Paper ID: 2127
    \affiliations
    \emails
}

\author{
Oliver Sourbut\thanks{Equal contribution. Correspondence to \href{mailto:oly@robots.ox.ac.uk}{oly@robots.ox.ac.uk} and \href{mailto:lewis.hammond@cs.ox.ac.uk}{lewis.hammond@cs.ox.ac.uk}. Published at IJCAI 2024, DOI: \url{https://doi.org/10.24963/ijcai.2024/26}.}
\and
Lewis Hammond\footnotemark[1]\And
Harriet Wood
\affiliations
University of Oxford
}

\begin{document}

\maketitle

\begin{abstract}
    Many settings of interest involving humans and machines -- from virtual personal assistants to autonomous vehicles -- can naturally be modelled as principals (humans) delegating to agents (machines), which then interact with each other on their principals' behalf.
    We refer to these multi-principal, multi-agent scenarios as \emph{delegation games}. 
    In such games, there are two important failure modes: problems of control (where an agent fails to act in line their principal's preferences) and problems of cooperation (where the agents fail to work well together).
    In this paper we formalise and analyse these problems, further breaking them down into issues of \emph{alignment} (do the players have similar preferences?) and \emph{capabilities} (how competent are the players at satisfying those preferences?).
    We show -- theoretically and empirically -- how these measures determine the principals' welfare, how they can be estimated using limited observations, and thus how they might be used to help us design more aligned and cooperative AI systems.
\end{abstract}

\section{Introduction}
\label{sec:intro}

With the continuing development of powerful and increasing general AI systems, we are likely to see many more tasks delegated to autonomous machines, from writing emails to driving us from place to place. Moreover, these machines are increasingly likely to come into contact with each other when acting on behalf of their human principals, whether they are virtual personal assistants attempting to schedule a meeting or autonomous vehicles (AVs) using the same road network. We refer to these multi-principal, multi-agent scenarios as \emph{delegation games}. The following example is shown in Figure \ref{fig:driving}.

\begin{example}
    \label{ex:driving}
    Two AVs can choose between routes on behalf of their passengers: $A$(utobahn) or $B$(eachfront). Their objectives are determined by the journey's speed and comfort (the passengers' objectives may differ, as shown). Each AV receives utility $6$ or $2$ for routes $A$ or $B$, respectively, with an additional penalty of $-3$ or $-2$ if both AVs choose the same route (due to congestion).
\end{example}

In delegation games there are two primary ways in which things can go wrong. First, a (machine) agent might not act according to the (human) principal's objective, such as when an AV takes an undesirable route -- a \emph{control} problem \cite{Russell2019}. Second, agents may fail to reach a cooperative solution, even if they are acting in line with their principals' objectives, 
such as when multiple AVs take the same route and end up causing congestion -- a \emph{cooperation} problem \cite{Doran1997,Dafoe2020}.

\begin{figure}[h]
    \centering
    \begin{subfigure}[b]{0.36\columnwidth}
        \centering
        \renewcommand\arraystretch{1.7}
        \begin{tabular}{c|c|c|c}
            \multicolumn{1}{c}{} & \multicolumn{1}{c}{$A$} & \multicolumn{1}{c}{$B$} &\\\cline{2-3}
            $A$ & \textcolor{orange}{3}, \textcolor{blue}{3} & \textcolor{orange}{6}, \textcolor{blue}{2} &\phantom{$A$}\\\cline{2-3}
            $B$ & \textcolor{orange}{2}, \textcolor{blue}{6} & \textcolor{orange}{0}, \textcolor{blue}{0} &\phantom{$B$}\\\cline{2-3}
        \end{tabular}
        \caption{}
    \label{fig:driving:a}
    \end{subfigure}    
    \begin{subfigure}[b]{0.36\columnwidth}
        \centering
        \renewcommand\arraystretch{1.7}
        \begin{tabular}{c|c|c|c}
            \multicolumn{1}{c}{} & \multicolumn{1}{c}{$A$} & \multicolumn{1}{c}{$B$} &\\\cline{2-3}
            $A$ & \textcolor{red}{2}, \textcolor{teal}{3} & \textcolor{red}{3}, \textcolor{teal}{3} &\phantom{$A$}\\\cline{2-3}
            $B$ & \textcolor{red}{4}, \textcolor{teal}{6} & \textcolor{red}{3}, \textcolor{teal}{0} &\phantom{$B$}\\\cline{2-3}
        \end{tabular}
        \caption{}
    \label{fig:driving:b}
    \end{subfigure}
    \begin{subfigure}[b]{0.25\columnwidth}
        \centering
        \includegraphics[width=0.8\textwidth]{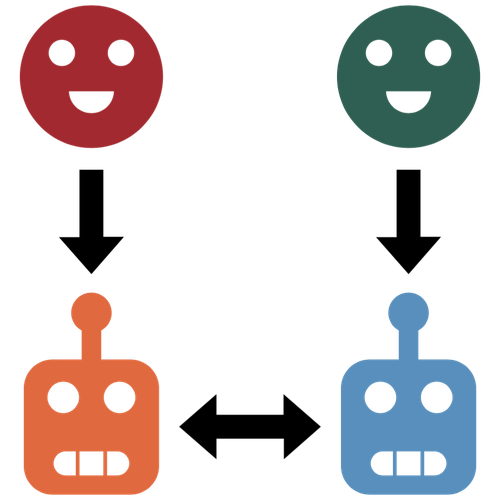}
        \caption{}
    \label{fig:driving:c}
    \end{subfigure}
    \caption{(a) The payoffs of the agents in Example \ref{ex:driving}; (b) the payoffs of the principals; and (c) a graphical representation, with vertical and horizontal arrows indicating control and cooperation, respectively.}
    \label{fig:driving}
\end{figure}

Control and cooperation can in turn be broken down into problems of \emph{alignment} and of \emph{capabilities} \cite{Hubinger2020a,Christiano2018c,Bostrom2014}.
For example, in the control failure above, the first AV might drive undesirably by taking route $A$ even though their passenger prefers the scenic beachfront (an alignment problem), or the second AV might undesirably take route $B$ because it is incapable of calculating the best route accurately (a capabilities problem).
Similarly, in the cooperation failure, the AVs might cause congestion because they cannot plan and communicate effectively enough (a capabilities problem), or because their objectives are fundamentally at odds with one another, e.g. they cannot both drive alone on route $A$ (an alignment problem).

As one might expect, ensuring good outcomes for the principals requires overcoming \emph{all} of these problems. This is made more challenging because most research considers each in isolation -- such as cooperation between agents with the same objective \cite{Torreno2017,Rizk2019,Du2023}, or alignment between a single principal and agent \cite{Kenton2021,Taylor2020,Russell2019} --
despite this being an increasingly unrealistic assumption for AI deployment. To ensure positive outcomes, we cannot rely on solutions to only \emph{some} of these problems.

\subsection{Contributions}

In this work we provide the first systematic investigation of these four failure modes and their interplay. More concretely, we make the following three core contributions: \emph{measures} for assessing each failure mode that satisfy a number of key desiderata (in Section \ref{sec:defs}); \emph{theoretical results} describing the relationships between these measures and principal welfare (in Section \ref{sec:theory}); and \emph{experiments} that validate these results and explore how the measures can be inferred from limited observations (in Section \ref{sec:exps}).
In doing so, we formalise and substantiate the intuition that solving all four of these problems is, in general, both necessary \emph{and} sufficient for good outcomes in multi-agent settings, which in turn has important implications for the project of building safe and beneficial AI systems.

\subsection{Related Work}

Given the foundational nature of the problems we study in this work, there is a vast amount of relevant prior research; due to space constraints we mention only a few exemplars on each topic.
The \emph{principal-agent literature} typically considers settings with a single principal and agent \cite{Laffont2002}. While there exist multi-agent variants such as competing mechanism games \cite{Yamashita2010,Peters2014}, to the best of our knowledge no previous work investigates the general requirements for high principal welfare.
Our setting is also similar to that of \emph{strategic delegation} \cite{Vickers1985,Sengul2011}, though we do not focus on principals' responses to each other's choice of agents.

The degree of alignment between two or more agents can be viewed as a measure of \emph{similarity between preferences}. Such measures have been introduced in areas such as mathematical economics \cite{Back1986}, computational social choice \cite{AlcaldeUnzu2015}, and reinforcement learning \cite{Gleave2021,Skalse2023}, though these works focus on either cooperation \emph{or} alignment.
In game theory, there are several classical values that measure the \emph{degree of (and costs from) competition}, such as the price of anarchy \cite{Koutsoupias1999} or coco value \cite{Kalai2013}.
Other works consider the robustness of these values under approximate equilibria \cite{Awasthi2010,Roughgarden2015}. We take inspiration from these ideas, extending them to settings in which the game we study is a \emph{proxy} for the game whose value we truly care about.

There have been several proposals for how to formally measure the capabilities of an agent. These include formal definitions of, e.g. \emph{general} intelligence \cite{Legg2007}, \emph{social} intelligence \cite{InsaCabrera2012}, and \emph{collective} intelligence \cite{Woolley2010}. In this work, we focus on how capabilities at the individual and collective level can be distinguished and how they combine with alignment to impact welfare.
Similarly, definitions of \emph{cooperation} also abound (see \cite{Tuomela2000} for an overview) though these are often informal and/or inconsistent \cite{West2007}, whereas we require a mathematical formalisation appropriate for applications to AI systems.

Finally, when it comes to estimating such properties from data, there is a large literature on the problem of \emph{preference elicitation/learning} \cite{Fuernkranz2011,Fischhoff2000}, including in general-sum games \cite{Conen2001,Gal2004,Yu2019}. 
The latter setting is also studied in \emph{empirical game theory} \cite{Wellman2006,Walsh2002,Waugh2011,Kuleshov2015}, including in recent work on inferring properties such as the price of anarchy \cite{Cousins2023}.
Other works attempt to quantify the capabilities of (reinforcement learning) agents by assessing their generalisation to different environments \cite{Cobbe2019} or co-players \cite{Leibo2021}.
While such ideas are not our focus, we make use of them in our final set of experiments.

\section{Background}
\label{sec:background}

In general, we use uppercase letters to denote sets, and lowercase letters to denote elements of sets or functions. We use boldface to denote tuples or sets thereof, typically associating one element to each of an ordered collection of players.
Unless otherwise indicated, we use superscripts to indicate an agent $1 \leq i \leq n$ and subscripts $j$ to index the elements of a set; for example, player $i$'s $j^{\th}$ (pure) strategy is denoted by $s^i_j$.
We also use the $\hat{}$ symbol to represent principals; for example, the $i^{\th}$ principal's utility function is written $\uHat^i$.
A notation table can be found for reference in Appendix \ref{app:notation}.

\begin{definition}
    A (strategic-form) \textbf{game} between $n$ players is a tuple $G = (\bm{S}, \bm{u})$ where $\bm{S} = \bigtimes_i S^i$ is the product space of (pure) strategy sets $S^i$ and $\bm{u}$ contains a utility function $u^i : \bm{S} \to \mathbb{R}$, for each agent $1 \leq i \leq n$.
    We write $s^i \in S^i$ and $\bm{s} \in \bm{S}$ to denote \textbf{pure strategies} and \textbf{pure strategy profiles}, respectively.
    A \textbf{mixed strategy} for player $i$ is a distribution $\sigma^i \in \Sigma^i$ over $S^i$, and a \textbf{mixed strategy profile} is a tuple $\bm{\sigma} = (\sigma^1, \ldots, \sigma^n) \in \bm{\Sigma} \coloneqq \bigtimes_i \Sigma^i$. 
    We will sometimes refer to pure strategy profiles in strategic-form games as \textbf{outcomes}.
    We write $\bm{s}^{-i} \in \bm{S}^{-i} \coloneqq \bigtimes_{j \neq i} S^j$ and therefore $\bm{s} = (\bm{s}^{-i}, s^i)$, with analogous notation for mixed strategies. We also abuse notation by sometimes writing $u^i(\bm{\sigma}) \coloneqq \expect_{\bm{\sigma}} \big[ u^i(\bm{s}) \big]$ and $\bm{u}(\bm{\sigma}) \coloneqq \big( u^1(\bm{\sigma}), \ldots, u^n(\bm{\sigma}) \big) \in \mathbb{R}^n$.
\end{definition}

Formally, a \emph{solution concept} maps from games $G$ to subsets of the mixed strategy profiles $\bm{\Sigma}$ in $G$. These concepts pick out certain strategy profiles based on assumptions about the (bounded) rationality of the individual players. The canonical solution concept is the Nash equilibrium.

\begin{definition}
    Given some $\bm{\sigma}^{-i}$ in a game $G$, $\sigma^i$ is a \textbf{best response} ($\BR$) for player $i$ if $u^i(\bm{\sigma}) \geq \max_{\tilde{\sigma}^i} u^i(\bm{\sigma}^{-i}, \tilde{\sigma}^i)$.
    We write the set of best responses for player $i$ as $\BR(\bm{\sigma}^{-i}; G)$. 
    A \textbf{Nash equilibrium} ($\NE$) in a game $G$ is a mixed strategy profile $\bm{\sigma}$ such that $\sigma^i \in \BR(\bm{\sigma}^{-i}; G)$ for every player $i$.
    We denote the set of $\NE$s in $G$ by $\NE(G)$.
\end{definition}

A \emph{social welfare function} $w : \mathbb{R}^n \to \mathbb{R}$ in an $n$-player game maps from payoff profiles $\bm{u}(\bm{s})$ to a single real number, aggregating players' payoffs into a measure of collective utility. We again abuse notation by writing $w(\bm{s}) \coloneqq w \big( \bm{u}(\bm{s}) \big)$ and $w(\bm{\sigma}) \coloneqq \expect_{\bm{\sigma}} \big[ w(\bm{s}) \big]$. In the remainder of this paper, we assume use of the following social welfare function, though the concepts we introduce do not heavily depend on this choice.

\begin{definition}
    Given a strategy profile $\bm{s}$, the \textbf{average utilitarian social welfare} is given by $w(\bm{s}) = \frac{1}{n} \sum_i u^i(\bm{s})$.
\end{definition}

\section{Delegation Games}
\label{sec:delegation-games}

In this work, we make the simplifying assumption that there is a one-to-one correspondence between principals and agents, and that each principal delegates fully to their corresponding agent (i.e. only agents can take actions). Our basic setting of interest can thus be characterised as follows.

\begin{definition}
    A (strategic-form) \textbf{delegation game} with $n$ principals and $n$ agents is a tuple $D = (\bm{S}, \bm{u}, \bm{\uHat})$, where $G \coloneqq (\bm{S}, \bm{u})$ is the game played by the agents, and $\GHat \coloneqq (\bm{S}, \bm{\uHat})$ is the game representing the principals' payoffs as a function of the agents' pure strategies.
\end{definition}

We refer to $G$ as the \emph{agent game} and $\GHat$ as the \emph{principal game}. For instance, $G$ and $\GHat$ from Example \ref{ex:driving} are shown in Figures \ref{fig:driving:a} and \ref{fig:driving:b} respectively. When not referring to principals or agents specifically, we refer to the \emph{players} of a delegation game. We denote the welfare in a delegation game for the principals and agents as $\wHat(\bm{s})$ and $w(\bm{s})$, respectively.

\begin{definition}
    Given a game $G$ and a social welfare function $w$, we define the \textbf{maximal (expected) welfare} achievable under $w$ as $w_\star(G) \coloneqq \max_{\bm{\sigma}} w(\bm{\sigma})$. We define the \textbf{ideal welfare} under $w$ as $w_+(G) \coloneqq w(\bm{u}_+)$ where $\bm{u}_+[i] = \max_{\bm{s}} u^i(\bm{s})$. 
    Similarly, we denote $w_\bullet(G) \coloneqq \min_{\bm{\sigma}} w(\bm{\sigma})$ and $w_-(G) \coloneqq w(\bm{u}_-)$ where $\bm{u}_-[i] = \min_{\bm{s}} u^i(\bm{s}).$
    We extend these definitions to \emph{delegation} games $D$ by defining $w_\dagger(D) \coloneqq w_\dagger(G)$ and $\wHat_\dagger(D) \coloneqq w_\dagger(\GHat)$, for $\dagger \in \{\star, +, \bullet, -\}$. When unambiguous, we omit the reference to $D$.
\end{definition}

Note that the maximal and ideal welfare may not be equivalent. For instance, in Example \ref{ex:driving} we have $w_\star = 4$ but $w_+ = 6$. The former is the maximum achievable among the available outcomes, while the latter is what would be achievable if all principals were somehow able to receive their maximal payoff simultaneously. We return to this distinction, represented graphically in Figure \ref{fig:welfare-line}, later. In general, our dependent variables of interest will be the principals' welfare regret $\wHat_\star - \wHat(\bm{\sigma})$ and the difference $\wHat_+ - \wHat_\star$.

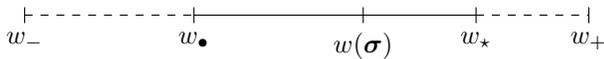
\begin{figure}[h]
    \centering
    \def\len{0.75}
    \begin{tikzpicture}
        \draw[dashed] (0,0) -- (3*\len,0);
        \draw (3*\len,0) -- (8*\len,0);
        \draw[dashed] (8*\len,0) -- (10*\len,0);
        \draw (0,3pt) -- (0,-3pt) node[below] {$w_-$};
        \draw (3*\len,3pt) -- (3*\len,-3pt) node[below] {$w_\bullet$};
        \draw (6*\len,3pt) -- (6*\len,-3pt) node[below] {$w(\bm{\sigma})$};
        \draw (8*\len,3pt) -- (8*\len,-3pt) node[below] {$w_\star$};
        \draw (10*\len,3pt) -- (10*\len,-3pt) node[below] {$w_+$};
    \end{tikzpicture}
    \caption{The range of social welfares in a game $G$.
    }
    \label{fig:welfare-line}
\end{figure}

\section{Control and Cooperation}
\label{sec:defs}

In Section \ref{sec:intro}, we distinguished between \emph{alignment} and \emph{capabilities} as contributors to the level of both control and cooperation. Our goal is to investigate how variations in the alignment and capabilities of agents impact the welfare of the principals. We therefore require ways to measure these concepts. Before doing so, we put forth a set of natural desiderata that we argue \emph{any} any such measures should satisfy.\footnote{Formalisations of these desiderata are provided in Section \ref{sec:desiderata}.}

\begin{enumerate}
    \item[(D1)] Alignment and capabilities -- both individual and collective -- are all `orthogonal' to one another in the sense that they can be instantiated in arbitrary combinations.\label{D1}
    
    \item[(D2)] Two players are perfectly individually aligned (misaligned) if and only if they have identical (opposite) preferences. Two or more players are perfectly \emph{collectively} aligned  if and only if they have identical preferences.\label{D2}

    \item[(D3)] If a set of agents are maximally capable, they achieve maximal \emph{agent} welfare. If they are also maximally individually aligned, then maximal \emph{principal} welfare is also achieved.
    \label{D3}

    \item[(D4)] If a set of players are perfectly collectively aligned, then their maximal welfare is their ideal welfare.\label{D4}

    \item[(D5)] Individual measures are independent of any transformations to the game that preserve \emph{individual} preferences, and collective measures are independent of any transformations that preserve \emph{collective} preferences (as captured by some measure of social welfare).\label{D5} 
    
\end{enumerate}

\subsection{Control}

We begin by considering the control of a single agent by a single principal. In essence, we wish to capture the degree to which an agent is acting in line with its principal's preferences. As we -- and others \cite{Bostrom2014,Christiano2018c,Hubinger2020a,Armstrong2018} -- have noted, this can be decomposed into a question of: a) how similar the agent's preferences are to the principal's; and b) how capable the agent is of pursuing its preferences.

\subsubsection{Alignment}
\label{sec:individual-alignment}

How can we tell if principal $i$'s and agent $i$'s preferences are similar? First, we must be more precise about what we mean by preferences.
Following our assumption that agents may play stochastically, we view preferences as orderings over \emph{distributions} of outcomes $\bm{\sigma} \preceq \bm{\sigma}' \Leftrightarrow u(\bm{\sigma}) \leq u(\bm{\sigma}')$.\footnote{Note that in strategic-form games, (mixed) strategy profiles are equivalent to distributions over the domain of players' utility functions, but in general this need not be the case.} 
To compare the preferences of principal $i$ ($\hat{\preceq}^i$) and agent $i$ ($\preceq^i$) we can therefore compare $\uHat^i$ and $u^i$. It is well-known, however, that the same preferences can be represented by different utility functions. In particular, $u$ and $u'$ represent the same preferences if and only if one is a positive affine transformation of the other \cite{MasColell1995}.

In order to meaningfully measure the difference between two utility functions, therefore, we must map each to a canonical element of the equivalence classes induced by the preferences they represent. We define such a map using a \emph{normalisation function} $\nu : U \to U$, where $U \coloneqq \R^{\abs{\bm{S}}}$ represents the space of utility functions in a game $G$. There are many possible choices of normalisation function, but in essence they must consist of a constant shift $c$ and a multiplicative factor $m$ \cite{Tewolde2021}, which together define the affine relationship. To satisfy our desiderata, we place additional requirements on $m$ and $c$, as follows.

\begin{definition}
    \label{def:normalisation}
    For each player $i$ in $G$, we define their \textbf{normalised utility function} $\nu(u^i) = u^i_\nu$ as:
    $$u^i_\nu \coloneqq \begin{cases}
        0 & \text{ if } m \big( u^i - c(u^i) \bm{1} \big) = 0\\
        \frac{u^i - c(u^i) \bm{1}}{m ( u^i - c(u^i) \bm{1} )} & \text{ otherwise,} 
    \end{cases}$$
    where $c : U \to \R$ is an affine-equivariant function and $m : U \to \R$ is any strictly convex norm.
\end{definition}

For notational convenience, we sometimes write $c^i \coloneqq c(u^i) \bm{1}$ and $m^i \coloneqq m(u^i - c^i)$, with equivalent notation $\hat{c}^i$ and $\hat{m}^i$ when applied to $\uHat^i$, resulting in $\uHat^i_\nu$. Then $u^i = m^i u_\nu^i + c^i$.

\begin{lemma}
    \label{lem:nu-equivalence}
    For any $u,u' \in U$, $u_\nu = u'_\nu$ if and only if $\preceq = \preceq'$.
\end{lemma}

Given $\uHat^i_\nu$ and $u^i_\nu$, a natural way to measure the degree of alignment between the $i^{\th}$ principal and agent is to compute some kind \emph{distance} between $\uHat^i_\nu$ and $u^i_\nu$. To do so, we use a norm of the difference between $\uHat^i_\nu$ and $u^i_\nu$, which gives rise to the following pseudometric over $U$. In order to contrast this principal-agent alignment measure with our later alignment measure over $n$ players, we sometimes refer to this as \emph{individual} (as opposed to \emph{collective}) alignment.

\begin{definition}
    \label{def:IA}
    Given a delegation game $D$, the \textbf{(individual) alignment} between agent $i$ and their principal is given by $\IA^i(D) = 1 - \frac{1}{2} m(\uHat^i_\nu - u^i_\nu)$, where $m$ is the same (strictly convex) norm used for normalisation.
\end{definition}

The choice of $m$ and $c$ determine which differences between payoffs are emphasised by the measure. One way to make this choice is by writing $m = \norm{\cdot}_{d}$, where $d$ is a distribution over $\bm{S}$. But how should one choose $d$ and $\norm{\cdot}$?

Beginning with $d$, one's first intuition might be to consider a distribution over the \emph{equilibria} of the game. Assuming agents act self-interestedly, there are certain outcomes that are game-theoretically untenable; divergences between preferences over these outcomes could reasonably be ignored. This intuition, however, conflicts with one of our primary desiderata (D1), which is to tease apart the difference between alignment and capabilities -- in the next subsection, we show that agents' individual capabilities determine the equilibria of the game.
Instead, we argue that the outcome of a game does \emph{not} change the extent to which preferences (dis)agree, and so in general assume that $d$ has full support.

Our primary requirements on $\norm{\cdot}$ are that $m$ is \emph{strictly convex} and is the \emph{same} in Definitions \ref{def:normalisation} and \ref{def:IA}. These restrictions are required in order to satisfy all of our desiderata, but relaxations are possible if fewer requirements are needed (see Appendix \ref{appendix:relaxations} for more discussion).

\subsubsection{Capabilities}

One obvious way of creating a formal measure of an agent's capabilities is to consider the number of (distinct) strategies available to them. In the cases of boundedly rational agents or multi-agent settings, however, it can beneficial to restrict one's strategy space, either for computational reasons \cite{Wellman2006}, or by pre-committing to avoid temptation \cite{Gul2001}, or to force one's opponent to back down \cite{Rapoport1966}.
Alternatively, one could invoke a complexity-theoretic measure of capabilities by considering the time and memory available to each agent, though in this work we aim to be agnostic to such constraints.

Instead, inspired by the seminal work of \cite{Legg2007}, we view an individual agent's capabilities as the degree to which it is able to achieve its objectives in a range of situations, regardless of what those objectives are. In game-theoretic parlance, we consider the \emph{rationality} of the agent. We can naturally formalise this idea by defining an agent's capabilities as the degree of optimality of their responses to a given partial strategy profile $\bm{\sigma}^{-i}$.

\begin{definition}
    Given some $\bm{\sigma}^{-i}$ in a game $G$, a mixed strategy $\sigma^i$ is an $\epsilon^i$\textbf{-best response} ($\epsBR$) for player $i$ if:
    \begin{align*}
        u^i(\bm{\sigma}) 
        &\geq \min_{\tilde{\sigma}^i} u^i(\bm{\sigma}^{-i}, \tilde{\sigma}^i)\\
        &+ (1 - \epsilon^i) \big( \max_{\tilde{\sigma}^i} u^i(\bm{\sigma}^{-i}, \tilde{\sigma}^i) - \min_{\tilde{\sigma}^i} u^i(\bm{\sigma}^{-i}, \tilde{\sigma}^i) \big).
    \end{align*}
    We write the set of such best responses for player $i$ as $\epsBR(\bm{\sigma}^{-i}; G)$. An $\eps$\textbf{-Nash equilibrium} ($\epsNE$) in a game $G$ is a mixed strategy profile $\bm{\sigma}$ such that $\sigma^i \in \epsBR(\bm{\sigma}^{-i}; G)$ for every every player $i$, where $\eps = (\epsilon^1,\ldots,\epsilon^n)$.
    We denote the set of $\epsNE$s in $G$ by $\epsNE(G)$.
\end{definition}

In essence, $\epsilon^i$ captures the fraction of their attainable utility that player $i$ manages to achieve. Note that if $\epsilon^i = 0$ then $\epsBR(\bm{\sigma}^{-i}; G) = \BR(\bm{\sigma}^{-i}; G)$ and if $\epsilon^i = 1$ then $\epsBR(\bm{\sigma}^{-i}; G) = \Sigma^i$. Similarly, when $\eps = \bm{0}$ then $\epsNE(G) = \NE(G)$, and when $\eps = \bm{1}$ then $\epsNE(G) = \bm{\Sigma}$.

\begin{definition}
    \label{def:IC}
    Given a delegation game $D$, the \textbf{individual capability} of agent $i$ is $\IC^i(D) \coloneqq 1 - \epsilon^i \in [0,1]$ where $\epsilon^i$ is the smallest value such that agent $i$ plays an $\epsBR$ in $G$.
\end{definition}

Unlike other formulations of bounded rationality, such as a softmax strategy or randomisation with some fixed probability, Definition \ref{def:IC} -- which is analogous to \emph{satisficing} \cite{Simon1956,Taylor2016} -- is agnostic as to the precise \emph{mechanism} via which players are irrational, and thus serves as a general-purpose descriptor of a player's (ir)rationality level.\footnote{Our choice of $\epsilon^i$-best responses could also be weakened to, e.g. $\epsilon^i$-rationalisability \cite{Bernheim1984,Pearce1984}.\label{fn:rationalisability}}

\subsection{Cooperation}

In order to achieve good outcomes for the principals, it is not sufficient for the agents to coordinate with their principals individually, they must also coordinate with one another. Indeed, it is easy to show that the principals' welfare regret can be arbitrarily high in the only NE of a game, despite perfect control of each agent by its principal. 

\begin{lemma}
    \label{lem:PD}
    There exists a (two-player, two-action) delegation game $D$ such that for any $x > 0$, however small, even if $\IA^i(D) = 1$ and $\IC^i(D) = 1$ for each agent, we have only one NE $\bm{\sigma}$, and $\frac{w_\star - w(\bm{\sigma})}{w_+- w_-} = 1-x$.
\end{lemma}

To achieve low principal welfare regret, we need to have a sufficiently high degree of cooperation, both in terms of: a) collective \emph{alignment} (the extent to which agents have similar preferences); and b) collective \emph{capabilities} (the extent to which agents can work together to overcome their differences in preferences).

\subsubsection{Alignment}
\label{sec:collective-alignment}

Intuitively, it should be easier to achieve high welfare in a game where the players have similar preferences than one in which the players have very different preferences. At the extremes of this spectrum we have zero-sum games and common-interest games, respectively. This intuition can be formalised by generalising Definition \ref{def:IA} to measure the degree of alignment between $n$ utility functions, rather than two.

\begin{definition}
    \label{def:CA}
    Given a delegation game $D$, the \textbf{collective alignment} between the agents is given by:
    $$\CA(D) = 1 - \sum_i \frac{m^i}{\sum_j m^j} \cdot m(\mu^w - u^i_\nu),$$
    where $\mu^w \coloneqq \frac{\sum_i u^i - c^i}{\sum_i m^i}$ is a proxy for the agents' (normalised) welfare, recalling that $m^i \coloneqq m(u^i - c^i)$.
\end{definition}

Intuitively, we consider the misalignment of each agent from a hypothetical agent whose objective is precisely to promote overall social welfare. This misalignment is weighted by $m^i$, the idea being that the `stronger' agent $i$'s preferences (and hence the larger $m^i$ is), the more their misalignment with the overall welfare of the collective matters. While it may not be immediately obvious why we use $\mu^w$ instead of $\nu(w)$, the former will allow us to derive tighter bounds and can easily be shown to induce the same ordering over mixed strategy profiles as $w$ (and hence also $\nu(w)$).

\begin{lemma}
    \label{lem:mu-welfare}
    For any $\bm{\sigma}, \bm{\sigma}' \in \bm{\Sigma}$, $\mu^w(\bm{\sigma}) \leq \mu^w(\bm{\sigma}')$ if and only if $w(\bm{\sigma}) \leq w(\bm{\sigma}')$.
\end{lemma}

Unfortunately, collective alignment (even when paired with perfect control) is insufficient for high principal welfare. The most trivial examples of this are equilibrium selection problems, but we can easily construct a game with a \emph{unique} NE and arbitrarily high welfare regret, even when we have perfect control and arbitrarily high collective alignment. These issues motivate our fourth and final measure.

\begin{lemma}
    \label{lem:TD}
    There exists a family of (two-player, $k$-action) delegation games $D$ such that even if $\IA^i(D) = 1$, $\IC^i(D) = 1$ for each agent, and there is only one NE $\bm{\sigma}$, we have $\lim_{k \to \infty} \CA(D) = 1$ but $\lim_{k \to \infty} \frac{\wHat_\star - \wHat(\bm{\sigma})}{\wHat_+ - \wHat_-} = 1$.
\end{lemma}

\subsubsection{Capabilities}
\label{sec:collective-capabilities}

One way to model collective capabilities is as `internal' to the game $G$. Under this conceptualisation, we assume that the ability of the agents to cooperate is captured entirely by their actions and payoffs in $G$. For example, if the agents were able to coordinate in $G$ using a commonly observed signal $\gamma$, then this would be modelled as them playing a \emph{different} game $G'$, in which each agent's action consists of a choice of action in $G$ given their observation of $\gamma$.\footnote{Thus, in a \emph{true} prisoner's dilemma, the `only thing to do' (and therefore trivially the cooperative action) is to defect. The idea here is that if the agents possessed better cooperative capabilities, they would not be faced with an actual prisoner's dilemma to begin with.}

This approach, however, conflates the agents' collective alignment with their collective capabilities. In order to tease these concepts apart, we require a way to measure the extent to which the agents can avoid welfare loss due to their selfish incentives. Perhaps the best known formalisation of this loss is the \emph{price of anarchy} \cite{Koutsoupias1999}, which captures the difference in welfare between the best possible outcome and the worst possible NE.\footnote{Considering the \emph{worst} case captures the welfare loss from equilibrium selection problems even in common-interest games.} Inspired by this idea, we measure collective capabilities as follows.

\begin{definition}
    \label{def:CC}
    Let $w_{\eps} \coloneqq \min_{\bm{\sigma} \in \epsNE(G)} w(\bm{\sigma})$. 
    Given a delegation game $D$ where the agents have individual capabilities $\eps$, the \textbf{collective capabilities} of the agents are $\CC(D) \coloneqq \delta \in [0,1]$ if and only if the agents achieve welfare at least $w_{\eps} + \delta \cdot (w_\star - w_{\bm{0}})$, where recall that $\bm{0}$-$\NE(G) = \NE(G)$.\footnote{As remarked in Footnote \ref{fn:rationalisability}, the use of $\epsNE$s is not intrinsic to our definition of (collective) capabilities and could be weakened to, e.g. $\eps$-rationalisable outcomes.}
\end{definition}

Note that if $\epsilon^i \geq \tilde{\epsilon}^i$ for every $1 \leq i \leq n$, then we must have $w_{\eps} \leq w_{\tilde{\eps}}$; a special case is $w_{\eps} \leq w_{\bm{0}}$. Thus, the individual irrationality of the agents can only \emph{lower} the (worst-case) welfare loss. On the other hand, we can see that greater collective capabilities can potentially compensate for this loss.

As in the case of individual capabilities, we provide a measure that is agnostic to the precise mechanism via which the agents cooperate, be it through commitments, communications, norms, institutions, or more exotic schemes. Rather, we take as input the fact that agents are able to obtain a certain amount of welfare, and use this quantify how well they are cooperating. At one extreme, they do no better than $w_{\eps}$, at the other they get as close to the maximal welfare $w_\star$ as their individual capabilities will allow. 

\section{Theoretical Results}
\label{sec:theory}

We begin our theoretical results by proving that the measures defined in the preceding section satisfy our desiderata, before using them to bound the principals' welfare regret.

\subsection{Desiderata}
\label{sec:desiderata}

The fact that we define alignment as a feature of the underlying \emph{game}, and capabilities are a feature of how the game is \emph{played} means that capabilities and alignment are naturally orthogonal. The potentially arbitrary difference between the principals' and agents' utility functions is the key to the other parts of the following result.

\begin{proposition}[D1]
    \label{prop:D1}
    Consider a delegation game $D$ with measures $\IA(D)$, $\IC(D)$, $\CA(D)$, and $\CC(D)$. Holding fixed any three of the measures, then for any value $v \in [0,1]$ (or $\bm{v} \in [0,1]^n$ for $\IA$ or $\IC$), there is a game $D'$ such that the fourth measure takes value $v$ ($\bm{v}$) in $D'$ and the other measures retain their previous values.
\end{proposition}

D2 follows chiefly from classic results linking preferences over mixed strategy profiles to sets of utility functions that are positive affine transformations of one another. 

\begin{proposition}[D2]
    \label{prop:D2}
    For any $1 \leq i \leq n$, $\IA^i = 1$ ($\IA^i = 0$) if and only if $\preceq^i = \hat{\preceq}^i$ ($\preceq^i = \hat{\succeq}^i$). Similarly $\CA = 1$ if and only if $\preceq^i = \preceq^j$ for every $1 \leq i, j \leq n$.
\end{proposition}

The first half of D3 is straightforward. The subtlety in the second half is that -- perhaps counterintuitively, at first -- perfect capabilities (both individual and collective) and perfect alignment between the principals and their agents is \emph{not} sufficient for the principals to achieve maximal welfare (unlike the agents). Rather, the resulting solution will be one (merely) on the Pareto frontier for the principals. 

In essence, this is because individual alignment does not preserve \emph{welfare orderings} $\preceq^w$, only individual preference orderings $\preceq^i$. Recalling that $\hat{m}^i$ and $m^i$ quantify the magnitudes of $\uHat^i$ and $u^i$ respectively, we can see that, in general, the aggregation over agents' utilities (used to measure their success at cooperating) may not give the same weight to each party as the aggregation over principals' utilities (used to measure the value we care about).
Alternatively, the variation in magnitudes $m^i$ can be viewed as capturing a notion of fairness (used to select a point on the Pareto frontier), which may not be the same as in $\GHat$ unless $\hat{m}^i = r \cdot m^i$ for some $r > 0$. Further discussion on this point can be found in Appendix \ref{appendix:pareto}.

\begin{proposition}[D3]
    \label{prop:D3}
    If $\IC = \bm{1}$ and $\CC = 1$ then any strategy $\bm{\sigma}$ the agents play is such that $w(\bm{\sigma}) = w_\star$. 
    If $\IA = \bm{1}$ then $\bm{\sigma}$ is Pareto-optimal for the principals.
    If, furthermore, $\hat{m}^i = r \cdot m^i$ for some $r > 0$ for all $i$, then $\hat{w}(\bm{\sigma}) = \hat{w}_\star$.
\end{proposition}

The proof of D4 follows naturally from Definition \ref{def:CA}. The final desideratum (D5) is simple in the case of individual preferences (due to the form of our normalisation function). In the case of collective preferences (i.e. the ordering $\preceq^w$ over mixed strategies induced by $w$), we make use of the fact that the relative magnitude of the agents' utility functions in both games must be the same (which is closely related to the `fairness' condition $\hat{m}^i = r \cdot m^i$ in Proposition \ref{prop:D3}).

\begin{proposition}[D4]
    \label{prop:D4}
    If $\CA = 1$ then $w_\star = w_+$.
\end{proposition}

\begin{proposition}[D5]
    \label{prop:D5}
    Given a delegation game $D_1$, let $D_2$ be such that $\preceq^i_1 = \preceq^i_2$ and $\hat{\preceq}^i_1 = \hat{\preceq}^i_2$ for each $1 \leq i \leq n$. Then $\IA(D_1) = \IA(D_2)$ and $\IC(D_1) = \IC(D_2)$. Moreover, if $D_2$ is such that $\preceq^w_1 = \preceq^w_2$ and the $u^i$ are affine-independent, then $\CA(D_1) = \CA(D_2)$ and $\CC(D_1) = \CC(D_2)$ as well.
\end{proposition}

\subsection{Bounding Welfare Regret}
\label{sec:bounds}

The primary question we investigate in this work is how the principals fare given the different levels of control and cooperation in the game played by the agents. We begin by characterising the \emph{principals'} welfare in terms of the \emph{agents'} utilities, and the alignment of each agent with its principal.

\begin{proposition}
    \label{prop:IA-bound}
    Given a delegation game $D$, we have:
    $$\wHat_\star - \wHat(\bm{\sigma}) \leq \frac{1}{n} \sum_i r^i \big( u^i(\hat{\bm{s}}_\star) - u^i(\bm{\sigma}) \big) + \frac{4K}{n} \hat{\bm{m}}^\top(\bm{1} - \IA),$$
    where $r^i \coloneqq \frac{\hat{m}^i}{m^i}$, $\hat{\bm{m}}[i] = \hat{m}^i$, $K$ satisfies $\norm{u_\nu - u'_\nu}_\infty \leq K \cdot m(u_\nu - u'_\nu)$ for any $u, u' \in U$, and $\wHat(\hat{\bm{s}}_\star) = \wHat_\star$.
\end{proposition}

Using this result, we can bound the principals' welfare regret in terms of both principal-agent alignment \emph{and} the agents' welfare regret, which is in turn a function of the agents' capabilities. As we saw in Propositions \ref{prop:D3} and \ref{prop:D5}, the `calibration' between the principals and agents -- as captured by the potentially differing ratios $r^i$ -- remains critical for ensuring we reach better outcomes in terms of principal welfare.
\begin{theorem}
    \label{thm:capabilities-bound}
    Given a delegation game $D$, we have that:
    \begin{align*}
        \wHat_\star - \wHat(\bm{\sigma})
        &\leq \frac{4K}{n} \hat{\bm{m}}^\top(\bm{1} - \IA) + r^* \left( (w_{\bm{0}} - w_{\eps}) \right.\\
        &+ \left. (1 - \CC) (w_\star - w_{\bm{0}}) \right) + R(\bm{\sigma}),
    \end{align*}
    where $\IC = \bm{1} - \eps$, 
    $r^* \in [\min_i r^i, \max_i r^i]$,
    $R(\bm{\sigma}) \coloneqq \frac {1}{n} \sum_i (r^i - r^*)m^i \left( u_\nu^i(\bm{\sHat}_\star) - u_\nu^i(\bm{\sigma}) \right)$ is a remainder accounting for collective alignment and calibration,
    and $K$ and $r^i$ are defined as in Proposition \ref{prop:IA-bound}.
    Note that when all $r^i$ are equal \emph{or} $\CA = 1$ then there is $r^*$ with $R(\bm{\sigma}) = 0$.
\end{theorem}
Before continuing, we note that unlike in the single-agent case, even small irrationalities can compound to dramatically lower individual payoffs (and thus welfare) in multi-agent settings, as formalised by the following lemma.
\begin{lemma}
    \label{lem:fragile}
    For any $\eps \succ \bm{0}$, there exists a game $G$ such that $w_{\bm{0}} = w_+$ but for any $x > 0$, however small, $w_{\eps} - w_- < x$.
\end{lemma}
In many games, however, the players' welfare will be much more robust to small mistakes. For example, suppose that $G$ is $(\eps, \Delta)$-robust, in the sense that all $\epsNE$s are contained within a ball of radius $\Delta$ around a (true) NE \cite{Awasthi2010}. Then it is relatively straightforward to show that: 
$$w_{\bm{0}} - w_{\eps} \leq \frac{2\Delta}{n} \sum_i \max_{\bm{s},\bm{s'}} \abs{u^i(\bm{s}) - u^i(\bm{s}')}.$$
Indeed, in many settings the price of anarchy can be bounded under play that is not perfectly rational \cite{Roughgarden2015}. While our bound above is highly general, assuming further structure in the game may allow us to tighten it further.

Theorem \ref{thm:capabilities-bound} characterises the principals' welfare regret in terms of three of our four measures. Our next result characterises the gap between the ideal and maximal welfare in terms of our fourth measure: collective alignment.
\begin{proposition}
    \label{prop:ideal-welfare-bound}
    Given a game $G$, then $w_+ - w_\star \leq \frac{K \sum_i m^i}{n} (1 - \CA)$, where $K$ is defined as in Proposition \ref{prop:IA-bound}.
\end{proposition}
This bound can be applied to either the agent or principal game; we denote the collective alignment in the latter as $\hat{\CA}$. While it is possible characterise the difference between $\hat{\CA}$ and $\CA$ in terms of $\IA$, a tighter bound on $\wHat_+ - \wHat(\bm{\sigma})$ can be obtained by considering the collective alignment between the principals and simply summing the right-hand terms of the inequalities in Theorem \ref{thm:capabilities-bound} and Proposition \ref{prop:ideal-welfare-bound}.

\section{Experiments}
\label{sec:exps}

While our primary contributions are theoretical, we support these results by: i) empirically validating the bounds above; and ii) showing how the various measures we introduce can be inferred from data. In our experiments we define $\nu$ using $c(u) = \expect_{\bm{s}}[u(\bm{s})]$ and $m(u) = \norm{u}_2$,
and limit our attention to pure strategies due to the absence of scalable methods for exhaustively finding mixed $\epsNE$s.\footnote{Our code is available at \url{https://github.com/lrhammond/delegation-games}.}

\subsection{Empirical Validation}
\label{sec:validation}

In order to visualise the results in the preceding section, we conduct a series of experiments in which we monitor how the principals' welfare changes based on the degree of control and cooperation in the delegation game. An example is shown in Figure \ref{fig:welfare-bounds}, in which we change one measure, setting all others to 0.9. At each step we generate 25 random delegation games (with approximately ten outcomes), compute the set of strategies $\bm{s}$ such that $w(\bm{s}) \in \big[w_{\eps} + \CC \cdot (w_\star - w_{\bm{0}}),  w_{\eps} + (w_\star - w_{\bm{0}})\big]$ where $\eps = \bm{1} - \IC$, and record the mean principal welfare over these strategies.

\begin{figure}
    \centering
    \begin{subfigure}{0.49\columnwidth}
        \centering
        \includegraphics[width=0.9\textwidth]{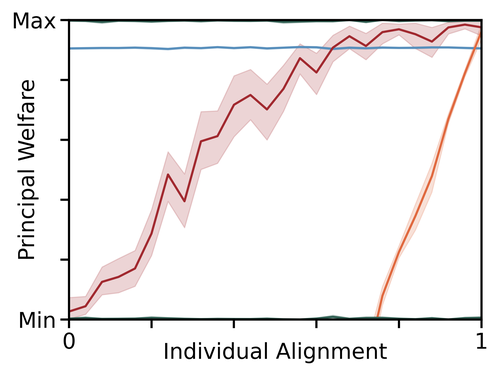}
    \end{subfigure}
    \begin{subfigure}{0.49\columnwidth}
        \centering
        \includegraphics[width=0.9\textwidth]{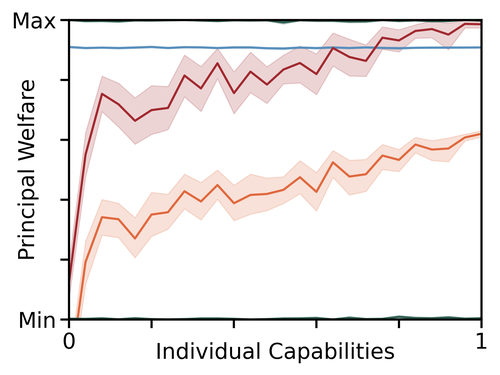}
    \end{subfigure}

    \begin{subfigure}{0.49\columnwidth}
        \centering
        \includegraphics[width=0.9\textwidth]{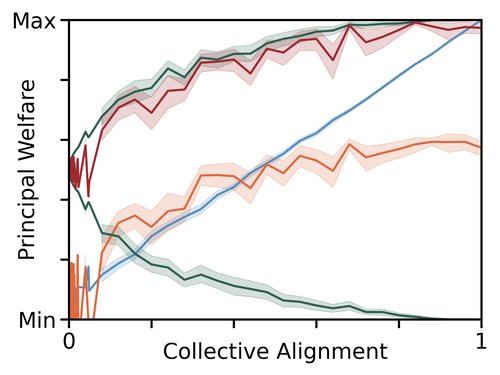}
    \end{subfigure}
    \begin{subfigure}{0.49\columnwidth}
        \centering
        \includegraphics[width=0.9\textwidth]{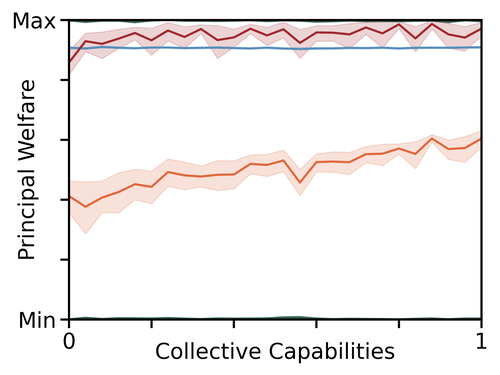}
    \end{subfigure}
    \caption{We report mean principal welfare (in red) normalised to $[\wHat_-,\wHat_+]$, with $\wHat_\bullet$ and $\wHat_\star$ in green. The lower bounds on welfare, given by Theorem \ref{thm:capabilities-bound}, and on $\wHat_\star$ (compared to $\wHat_+$), given by Proposition \ref{prop:ideal-welfare-bound}, are in orange and blue, respectively. Shaded areas show 90\% confidence intervals.}
    \label{fig:welfare-bounds}
\end{figure}
As expected, we see a positive relationship for each measure.
Given the ease of coordinating in relatively small games, alignment is more important in this example than capabilities, as can be seen from the gentler slope of the welfare curve as $\IA$ increases, and in how $\hat{\CA}$ places an upper limit on principal welfare under otherwise ideal conditions.

In Appendix \ref{app:more_games} we include further details and plots for games ranging between $10^1$ and $10^3$ outcomes, with values of the fixed measures ranging between $0$ to $1$. We find that the overall dependence of principal welfare upon each measure is robust, though the tightness of the bounds is reduced in larger games and for lower values of the fixed measures.

\subsection{Inference of Measures}
\label{sec:inference}

Previously, we implicitly assumed full knowledge of the delegation game in defining our measures. In the real world, this assumption will rarely be valid, motivating the question of when and how we might \emph{infer} their values given limited observations. Concretely, we assume access to a dataset $\D$ of (pure) strategy profiles and payoff vectors $(\bm{s}, \bm{u}, \bm{\uHat}) \sim_{\text{i.i.d.}} d$.

Inferring alignment is relatively straightforward, as we can simply approximate each $u^i$ and $\uHat^i$ by using only the strategies $\D(\bm{S}) \subseteq \bm{S}$ contained in $\D$, for which we know their values.
Inferring capabilities is much more challenging, as they determine how agents play the game and therefore limit our observations. Fundamentally, measuring capabilities requires comparing the achieved outcomes against better/worse alternatives, but if agents' capabilities are fixed we might never observe these other outcomes.
Moreover, only observing the agents acting together (alone) leaves us unable to asses how well they would perform alone (together), due to the orthogonality of individual 
and collective capabilities.

While there are many relaxations that might overcome these issues, perhaps the simplest and weakest is to assume that: a) all utilities are \emph{non-negative};
and b) we receive observations of the agents acting both alone \emph{and} together. We can then estimate (upper bounds of) $\IC$ and $\CC$ as follows.

\begin{proposition}
    \label{prop:inference}
    Given a game $G$, if $u^i(\bm{s})$ for every $i$ and $\bm{s} \in S$, and $d$ has maximal support over the outcomes generated when agents act together/alone (respectively), then: 
    \begin{align*}
        \CC &\leq \lim_{\abs{\D} \to \infty} \frac{\min_{\bm{s} \in \D(\bm{S})} w(\bm{s})}{\max_{\bm{s} \in \D(\bm{S})} w(\bm{s})}, \text{ and }\\
        \IC^i &\leq \lim_{\abs{\D} \to \infty} \min_{\bm{s} \in \D(\bm{S})} \frac{u^i(\bm{s})}{\max_{\tilde{s}^i \in \D(S^i)} u^i(\bm{s}^{-i},\tilde{s}^i)}.
    \end{align*}
\end{proposition}

In Figure \ref{fig:inference} we evaluate the accuracy of these estimates using samples generated from 100 randomly generated delegation games of various sizes. Consistent with our previous arguments, it is far easier to infer alignment than capabilities in the setup we consider, though we leave open the question of whether stronger assumptions and/or different setups might allow us to gain improved estimates of the latter.

\begin{figure}
    \centering
    \begin{subfigure}{0.49\columnwidth}
        \centering
        \includegraphics[width=0.95\textwidth]{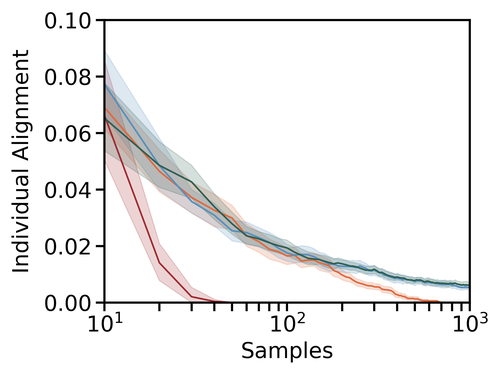}
    \end{subfigure}
    \begin{subfigure}{0.49\columnwidth}
        \centering
        \includegraphics[width=0.95\textwidth]{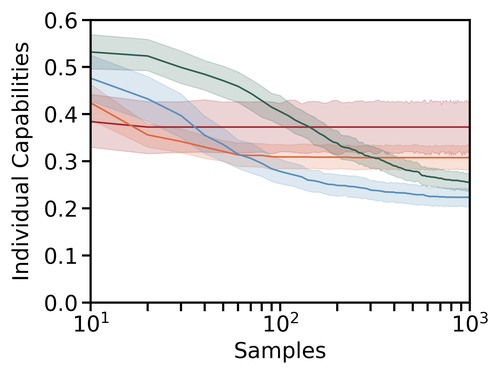}
    \end{subfigure}

    \begin{subfigure}{0.49\columnwidth}
        \centering
        \includegraphics[width=0.95\textwidth]{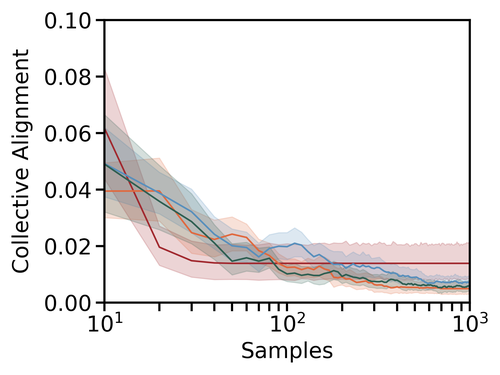}
    \end{subfigure}
    \begin{subfigure}{0.49\columnwidth}
        \centering
        \includegraphics[width=0.95\textwidth]{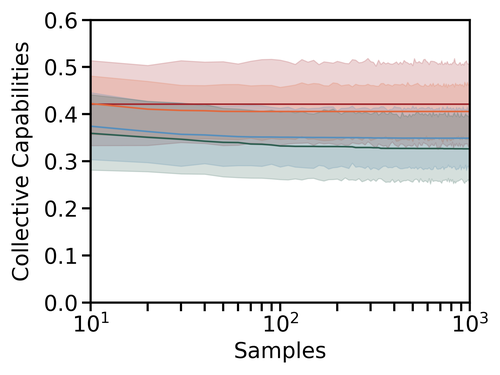}
    \end{subfigure}
    \caption{We report the mean absolute error of estimates of the four measures. The red, orange, blue, and green lines represent games with $10^k$ outcomes for $k \in \{1,2,3,4\}$, respectively. Shaded areas show 90\% confidence intervals.}
    \label{fig:inference}
\end{figure}

\section{Conclusion}

We formalised and studied problems of \emph{cooperation} and \emph{control} in delegation games -- a general model for interactions between multiple AI systems on behalf of multiple humans -- breaking these down into \emph{alignment} and \emph{capabilities}. We showed how these concepts are both necessary \emph{and} sufficient for good outcomes, and how they can be inferred from data.

We focused on strategic-form games to make our theoretical contributions clearer and more general, but future work could develop more specific results in more complex, structured models, or consider extensions where: a) the correspondence between principals and agents is not one-to-one; and b) the principals can also take actions. Finally, for practical applications we must develop better methods for inferring the concepts in this work from data.

\appendix

\section*{Acknowledgments}

The authors wish to thank Bart Jaworksi for contributing to an earlier version of this work, several anonymous reviewers for their comments, and Jesse Clifton, Joar Skalse, Sam Barnett, Vincent Conitzer, Charlie Griffn, David Hyland, Michael Wooldridge, Ted Turocy, and Alessandro Abate for insightful discussions on these topics. Lewis Hammond acknowledges the support of an EPSRC Doctoral Training Partnership studentship (Reference: 2218880). Oliver Sourbut acknowledges the University of Oxford’s Autonomous Intelligent Machines and Systems CDT and the Good Ventures Foundation for their support.

\section*{Contribution Statement}

Lewis Hammond conceived the initial project direction. Oliver Sourbut, Lewis Hammond, and Harriet Wood developed the direction together and contributed to the manuscript. Oliver Sourbut conceived and proved a majority of bounds presented, and Lewis Hammond a majority of inference results. Harriet Wood contributed to early experiments in code, and Oliver Sourbut and Lewis Hammond produced a majority of the final experiments.

\bibliographystyle{named}
\bibliography{refs}

\begin{thebibliography}{}

\bibitem[\protect\citeauthoryear{Alcalde-Unzu and
  Vorsatz}{2015}]{AlcaldeUnzu2015}
Jorge Alcalde-Unzu and Marc Vorsatz.
\newblock Do we agree? measuring the cohesiveness of preferences.
\newblock {\em Theory and Decision}, 80(2):313--339, 2015.

\bibitem[\protect\citeauthoryear{Armstrong and
  Mindermann}{2018}]{Armstrong2018}
Stuart Armstrong and S{\"o}ren Mindermann.
\newblock Occam's razor is insufficient to infer the preferences of irrational
  agents.
\newblock In {\em Proceedings of the 32nd International Conference on Neural
  Information Processing Systems}, pages 5603--5614, 2018.

\bibitem[\protect\citeauthoryear{Arrow \bgroup \em et al.\egroup
  }{1953}]{Arrow1953}
K.~J. Arrow, E.~W. Barankin, and D.~Blackwell.
\newblock 5. admissible points of convex sets.
\newblock In {\em Contributions to the Theory of Games ({AM}-28), Volume {II}},
  pages 87--92. Princeton University Press, 1953.

\bibitem[\protect\citeauthoryear{Awasthi \bgroup \em et al.\egroup
  }{2010}]{Awasthi2010}
Pranjal Awasthi, Maria-Florina Balcan, Avrim Blum, Or~Sheffet, and Santosh
  Vempala.
\newblock On nash-equilibria of approximation-stable games.
\newblock In {\em Algorithmic Game Theory}, pages 78--89. Springer Berlin
  Heidelberg, 2010.

\bibitem[\protect\citeauthoryear{Back}{1986}]{Back1986}
Kerry Back.
\newblock Concepts of similarity for utility functions.
\newblock {\em Journal of Mathematical Economics}, 15(2):129--142, 1986.

\bibitem[\protect\citeauthoryear{Bernheim}{1984}]{Bernheim1984}
B.~Douglas Bernheim.
\newblock Rationalizable strategic behavior.
\newblock {\em Econometrica}, 52(4):1007, 1984.

\bibitem[\protect\citeauthoryear{Bostrom}{2014}]{Bostrom2014}
Nick Bostrom.
\newblock {\em Superintelligence: Paths, Dangers, Strategies}.
\newblock Oxford University Press, 2014.

\bibitem[\protect\citeauthoryear{Che \bgroup \em et al.\egroup
  }{2020}]{Che2020}
Yeon-Koo Che, Jinwoo Kim, Fuhito Kojima, and Christopher~Thomas Ryan.
\newblock "near" weighted utilitarian characterizations of pareto optima.
\newblock {\em arXiv:2008.10819}, 2020.

\bibitem[\protect\citeauthoryear{Christiano}{2018}]{Christiano2018c}
Paul Christiano.
\newblock Clarifying “ai alignment”.
\newblock AI Alignment, 2018.

\bibitem[\protect\citeauthoryear{Cobbe \bgroup \em et al.\egroup
  }{2019}]{Cobbe2019}
Karl Cobbe, Oleg Klimov, Christopher Hesse, Taehoon Kim, and John Schulman.
\newblock Quantifying generalization in reinforcement learning.
\newblock In {\em Proceedings of the 36th International Conference on Machine
  Learning}, volume~97, pages 1282--1289, 2019.

\bibitem[\protect\citeauthoryear{Conen and Sandholm}{2001}]{Conen2001}
Wolfram Conen and Tuomas Sandholm.
\newblock Preference elicitation in combinatorial auctions.
\newblock In {\em Proceedings of the 3rd {ACM} conference on Electronic
  Commerce}. {ACM}, 2001.

\bibitem[\protect\citeauthoryear{Cousins \bgroup \em et al.\egroup
  }{2023}]{Cousins2023}
Cyrus Cousins, Bhaskar Mishra, Enrique Areyan~Viqueira, and Amy Greenwald.
\newblock Learning properties in simulation-based games.
\newblock In {\em Proceedings of the 2023 International Conference on
  Autonomous Agents and Multiagent Systems}, page 272–280, 2023.

\bibitem[\protect\citeauthoryear{Dafoe \bgroup \em et al.\egroup
  }{2020}]{Dafoe2020}
Allan Dafoe, Edward Hughes, Yoram Bachrach, Tantum Collins, Kevin~R. McKee,
  Joel~Z. Leibo, Kate Larson, and Thore Graepel.
\newblock Open problems in cooperative ai.
\newblock {\em arXiv:2012.08630}, 2020.

\bibitem[\protect\citeauthoryear{Daniilidis}{2000}]{Daniilidis2000}
Aris Daniilidis.
\newblock Arrow-barankin-blackwell theorems and related results in cone
  duality: A survey.
\newblock In {\em Lecture Notes in Economics and Mathematical Systems}, pages
  119--131. Springer Berlin Heidelberg, 2000.

\bibitem[\protect\citeauthoryear{Doran \bgroup \em et al.\egroup
  }{1997}]{Doran1997}
J.~E. Doran, S.~Franklin, N.~R. Jennings, and T.~J. Norman.
\newblock On cooperation in multi-agent systems.
\newblock {\em The Knowledge Engineering Review}, 12(3):309--314, 1997.

\bibitem[\protect\citeauthoryear{Du \bgroup \em et al.\egroup }{2023}]{Du2023}
Yali Du, Joel~Z. Leibo, Usman Islam, Richard Willis, and Peter Sunehag.
\newblock A review of cooperation in multi-agent learning.
\newblock {\em arXiv:2312.05162}, 2023.

\bibitem[\protect\citeauthoryear{Fischhoff and Manski}{2000}]{Fischhoff2000}
Baruch Fischhoff and Charles~F. Manski, editors.
\newblock {\em Elicitation of Preferences}.
\newblock Springer Netherlands, 2000.

\bibitem[\protect\citeauthoryear{Fürnkranz and
  Hüllermeier}{2011}]{Fuernkranz2011}
Johannes Fürnkranz and Eyke Hüllermeier, editors.
\newblock {\em Preference Learning}.
\newblock Springer Berlin Heidelberg, 2011.

\bibitem[\protect\citeauthoryear{Gal \bgroup \em et al.\egroup
  }{2004}]{Gal2004}
Ya'akov Gal, Avi Pfeffer, Francesca Marzo, and Barbara~J. Grosz.
\newblock Learning social preferences in games.
\newblock In {\em Proceedings of the 19th National Conference on Artifical
  Intelligence}, AAAI'04, page 226–231. AAAI Press, 2004.

\bibitem[\protect\citeauthoryear{Gleave \bgroup \em et al.\egroup
  }{2021}]{Gleave2021}
Adam Gleave, Michael Dennis, Shane Legg, Stuart Russell, and Jan Leike.
\newblock Quantifying differences in reward functions.
\newblock In {\em 9th International Conference on Learning Representations},
  2021.

\bibitem[\protect\citeauthoryear{Gul and Pesendorfer}{2001}]{Gul2001}
Faruk Gul and Wolfgang Pesendorfer.
\newblock Temptation and self-control.
\newblock {\em Econometrica}, 69(6):1403--1435, 2001.

\bibitem[\protect\citeauthoryear{Hubinger}{2020}]{Hubinger2020a}
Evan Hubinger.
\newblock Clarifying inner alignment terminology.
\newblock Alignment Forum, 2020.

\bibitem[\protect\citeauthoryear{Insa-Cabrera \bgroup \em et al.\egroup
  }{2012}]{InsaCabrera2012}
Javier Insa-Cabrera, Jos{\'{e}}-Luis Benacloch-Ayuso, and Jos{\'{e}}
  Hern{\'{a}}ndez-Orallo.
\newblock On measuring social intelligence: Experiments on competition and
  cooperation.
\newblock In {\em Artificial General Intelligence}, pages 126--135. Springer
  Berlin Heidelberg, 2012.

\bibitem[\protect\citeauthoryear{Kalai and Kalai}{2013}]{Kalai2013}
Adam Kalai and Ehud Kalai.
\newblock Cooperation in strategic games revisited.
\newblock {\em The Quarterly Journal of Economics}, 128(2):917--966, 2013.

\bibitem[\protect\citeauthoryear{Kenton \bgroup \em et al.\egroup
  }{2021}]{Kenton2021}
Zachary Kenton, Tom Everitt, Laura Weidinger, Iason Gabriel, Vladimir Mikulik,
  and Geoffrey Irving.
\newblock Alignment of language agents.
\newblock {\em arXiv:2103.14659}, 2021.

\bibitem[\protect\citeauthoryear{Koutsoupias and
  Papadimitriou}{1999}]{Koutsoupias1999}
Elias Koutsoupias and Christos Papadimitriou.
\newblock Worst-case equilibria.
\newblock In {\em {STACS} 99}, pages 404--413. Springer Berlin Heidelberg,
  1999.

\bibitem[\protect\citeauthoryear{Kuleshov and Schrijvers}{2015}]{Kuleshov2015}
Volodymyr Kuleshov and Okke Schrijvers.
\newblock Inverse game theory: Learning utilities in~succinct games.
\newblock In {\em Web and Internet Economics}, pages 413--427. Springer Berlin
  Heidelberg, 2015.

\bibitem[\protect\citeauthoryear{Laffont and Martimort}{2002}]{Laffont2002}
Jean-Jacques Laffont and David Martimort.
\newblock {\em The Theory of Incentives}.
\newblock Princeton University Press, 2002.

\bibitem[\protect\citeauthoryear{Legg and Hutter}{2007}]{Legg2007}
Shane Legg and Marcus Hutter.
\newblock Universal intelligence: A definition of machine intelligence.
\newblock {\em Minds and Machines}, 17(4):391--444, 2007.

\bibitem[\protect\citeauthoryear{Leibo \bgroup \em et al.\egroup
  }{2021}]{Leibo2021}
Joel~Z. Leibo, Edgar~A. Du{\'{e}}{\~{n}}ez{-}Guzm{\'{a}}n, Alexander
  Vezhnevets, John~P. Agapiou, Peter Sunehag, Raphael Koster, Jayd Matyas,
  Charlie Beattie, Igor Mordatch, and Thore Graepel.
\newblock Scalable evaluation of multi-agent reinforcement learning with
  melting pot.
\newblock In {\em Proceedings of the 38th International Conference on Machine
  Learning}, volume 139, pages 6187--6199, 2021.

\bibitem[\protect\citeauthoryear{Mas-Colell \bgroup \em et al.\egroup
  }{1995}]{MasColell1995}
Andreu Mas-Colell, Michael~D. Whinston, and Jerry~R. Green.
\newblock {\em Microeconomic Theory}.
\newblock Oxford University Press, 1995.

\bibitem[\protect\citeauthoryear{Pearce}{1984}]{Pearce1984}
David~G. Pearce.
\newblock Rationalizable strategic behavior and the problem of perfection.
\newblock {\em Econometrica}, 52(4):1029, 1984.

\bibitem[\protect\citeauthoryear{Peters}{2014}]{Peters2014}
Michael Peters.
\newblock Competing mechanisms.
\newblock {\em Canadian Journal of Economics/Revue canadienne
  d{\textquotesingle}{\'{e}}conomique}, 47(2):373--397, 2014.

\bibitem[\protect\citeauthoryear{Rapoport and Chammah}{1966}]{Rapoport1966}
Anatol Rapoport and Albert~M. Chammah.
\newblock The game of chicken.
\newblock {\em American Behavioral Scientist}, 10(3):10--28, 1966.

\bibitem[\protect\citeauthoryear{Rizk \bgroup \em et al.\egroup
  }{2019}]{Rizk2019}
Yara Rizk, Mariette Awad, and Edward~W. Tunstel.
\newblock Cooperative heterogeneous multi-robot systems: A survey.
\newblock {\em ACM Computing Surveys}, 52(2):1--31, 2019.

\bibitem[\protect\citeauthoryear{Roughgarden}{2015}]{Roughgarden2015}
Tim Roughgarden.
\newblock Intrinsic robustness of the price of anarchy.
\newblock {\em Journal of the {ACM}}, 62(5):1--42, 2015.

\bibitem[\protect\citeauthoryear{Russell}{2019}]{Russell2019}
Stuart Russell.
\newblock {\em Human Compatible}.
\newblock Penguin LCC US, 2019.

\bibitem[\protect\citeauthoryear{Sengul \bgroup \em et al.\egroup
  }{2011}]{Sengul2011}
Metin Sengul, Javier Gimeno, and Jay Dial.
\newblock Strategic delegation: A review, theoretical integration, and research
  agenda.
\newblock {\em Journal of Management}, 38(1):375--414, 2011.

\bibitem[\protect\citeauthoryear{Simon}{1956}]{Simon1956}
H.~A. Simon.
\newblock Rational choice and the structure of the environment.
\newblock {\em Psychological Review}, 63(2):129--138, 1956.

\bibitem[\protect\citeauthoryear{Skalse \bgroup \em et al.\egroup
  }{2023}]{Skalse2023}
Joar Skalse, Lucy Farnik, Sumeet~Ramesh Motwani, Erik Jenner, Adam Gleave, and
  Alessandro Abate.
\newblock Starc: A general framework for quantifying differences between reward
  functions.
\newblock {\em arXiv:2309.15257}, 2023.

\bibitem[\protect\citeauthoryear{Taylor \bgroup \em et al.\egroup
  }{2020}]{Taylor2020}
Jessica Taylor, Eliezer Yudkowsky, Patrick LaVictoire, and Andrew Critch.
\newblock {\em Alignment for Advanced Machine Learning Systems}, pages
  342--382.
\newblock Oxford University Press, 2020.

\bibitem[\protect\citeauthoryear{Taylor}{2016}]{Taylor2016}
Jessica Taylor.
\newblock Quantilizers: A safer alternative to maximizers for limited
  optimization.
\newblock In {\em Proceedings of the 2016 AAAI/ACM Conference on AI, Ethics,
  and Society}, 2016.

\bibitem[\protect\citeauthoryear{Tewolde}{2021}]{Tewolde2021}
Emanuel Tewolde.
\newblock Game transformations that preserve nash equilibrium sets and/or best
  response sets.
\newblock {\em arXiv:2111.00076}, 2021.

\bibitem[\protect\citeauthoryear{Torreño \bgroup \em et al.\egroup
  }{2017}]{Torreno2017}
Alejandro Torreño, Eva Onaindia, Antonín Komenda, and Michal Štolba.
\newblock Cooperative multi-agent planning: A survey.
\newblock {\em ACM Computing Surveys}, 50(6):1--32, 2017.

\bibitem[\protect\citeauthoryear{Tuomela}{2000}]{Tuomela2000}
Raimo Tuomela.
\newblock {\em Cooperation}.
\newblock Springer Netherlands, 2000.

\bibitem[\protect\citeauthoryear{Vickers}{1985}]{Vickers1985}
John Vickers.
\newblock Delegation and the theory of the firm.
\newblock {\em The Economic Journal}, 95:138, 1985.

\bibitem[\protect\citeauthoryear{{von Neumann} and
  Morgenstern}{1944}]{vonNeumann1944}
John {von Neumann} and Oskar Morgenstern.
\newblock {\em Theory of Games and Economic Behavior}.
\newblock Princeton University Press, 1944.

\bibitem[\protect\citeauthoryear{Walsh \bgroup \em et al.\egroup
  }{2002}]{Walsh2002}
William~E. Walsh, Rajarshi Das, Gerald Tesauro, and Jeffrey~O. Kephar.
\newblock Analyzing complex strategic interactions in multi-agent systems.
\newblock In {\em Game Theoretic and Decision Theoretic Agents Workshop at
  AAAI}, 2002.

\bibitem[\protect\citeauthoryear{Waugh \bgroup \em et al.\egroup
  }{2011}]{Waugh2011}
Kevin Waugh, Brian~D. Ziebart, and J.~Andrew Bagnell.
\newblock Computational rationalization: The inverse equilibrium problem.
\newblock In {\em Proceedings of the 28th International Conference on
  International Conference on Machine Learning}, page 1169–1176, 2011.

\bibitem[\protect\citeauthoryear{Wellman}{2006}]{Wellman2006}
Michael~P. Wellman.
\newblock Methods for empirical game-theoretic analysis.
\newblock In {\em Proceedings of the Twenty-First {AAAI} Conference on
  Artificial Intelligence}, AAAI'06, page 1552–1555. AAAI Press, 2006.

\bibitem[\protect\citeauthoryear{West \bgroup \em et al.\egroup
  }{2007}]{West2007}
S.~A. West, A.~S. Griffin, and A.~Gardner.
\newblock Social semantics: altruism, cooperation, mutualism, strong
  reciprocity and group selection.
\newblock {\em Journal of Evolutionary Biology}, 20(2):415--432, 2007.

\bibitem[\protect\citeauthoryear{Woolley \bgroup \em et al.\egroup
  }{2010}]{Woolley2010}
Anita~Williams Woolley, Christopher~F. Chabris, Alex Pentland, Nada Hashmi, and
  Thomas~W. Malone.
\newblock Evidence for a collective intelligence factor in the performance of
  human groups.
\newblock {\em Science}, 330(6004):686--688, 2010.

\bibitem[\protect\citeauthoryear{Yaari}{1981}]{Yaari1981}
Menahem~E Yaari.
\newblock Rawls, edgeworth, shapley, nash: Theories of distributive justice
  re-examined.
\newblock {\em Journal of Economic Theory}, 24(1):1--39, 1981.

\bibitem[\protect\citeauthoryear{Yamashita}{2010}]{Yamashita2010}
Takuro Yamashita.
\newblock Mechanism games with multiple principals and three or more agents.
\newblock {\em Econometrica}, 78(2):791--801, 2010.

\bibitem[\protect\citeauthoryear{Yu \bgroup \em et al.\egroup }{2019}]{Yu2019}
Lantao Yu, Jiaming Song, and Stefano Ermon.
\newblock Multi-agent adversarial inverse reinforcement learning.
\newblock In {\em Proceedings of the 36th International Conference on Machine
  Learning}, volume~97, pages 7194--7201, 2019.

\end{thebibliography}

\setcounter{lemma}{0}
\setcounter{proposition}{0}
\setcounter{theorem}{0}
\setcounter{corollary}{0}

\section{Notation}
\label{app:notation}

See the opening paragraph of Section \ref{sec:background} for a summary of the notational principles employed in this paper.
Other general mathematical notation (aside from mere abbreviations) that we repeatedly make use of is as follows.

\begin{center}
    \begin{tabular}{l l l}
        \toprule
        \textbf{Term} & \textbf{Meaning} & \textbf{Section}\\
        \midrule
        $G$ & Game & \ref{sec:background}\\
        $u$ & Utility function & \ref{sec:background}\\
        $S$ & Pure strategy space & \ref{sec:background}\\
        $s$ & Pure strategy & \ref{sec:background}\\
        $\Sigma$ & Mixed strategy space & \ref{sec:background}\\
        $\sigma$ & Mixed strategy & \ref{sec:background}\\
        $D$ & Delegation game & \ref{sec:delegation-games}\\
        $w$ & Social welfare function & \ref{sec:delegation-games}\\
        $w_+$ & Ideal welfare & \ref{sec:delegation-games}\\
        $w_-$ & Pessimal welfare & \ref{sec:delegation-games}\\
        $w_\star$ & Maximal welfare & \ref{sec:delegation-games}\\
        $w_\bullet$ & Minimal welfare & \ref{sec:delegation-games}\\
        $\preceq$ & Preference ordering & \ref{sec:individual-alignment}\\
        $\nu$ & Normalisation function & \ref{sec:individual-alignment}\\
        $m$ & (Strictly convex) norm & \ref{sec:individual-alignment}\\
        $c$ & (Affine-equivariant) shift function & \ref{sec:individual-alignment}\\
        $\mu^w$ & Normalised welfare proxy & \ref{sec:collective-alignment}\\
        $w_{\eps}$ & Lowest $\epsNE$ welfare & \ref{sec:collective-capabilities}\\
        \bottomrule
    \end{tabular}
\end{center}

\section{Proofs}
\label{app:proofs}

We begin by providing proofs for all of the results stated in the main body of the paper, as well as a small number of auxiliary results.

\subsection{Desiderata}

We begin by proving that the measures introduced in Definitions \ref{def:IA}, \ref{def:IC}, \ref{def:CA}, and \ref{def:CC} satisfy the desiderata introduced in Section \ref{sec:delegation-games}. The original statements of the results can be found in \ref{sec:desiderata}. In the proofs below, we use the following properties of the norm $m$:
\begin{itemize}
    \item Positive definiteness: if $m(u) = 0$ then $u = \bm{0}$,
    \item Absolute homogeneity: $m(k \cdot u) = \abs{k} \cdot m(u)$ for all $k \in \R$ and $u \in U$,
    \item Strict convexity: $m\big( k \cdot u + (1-k) \cdot u' \big) < k \cdot m(u) + (1 - k) \cdot m(u')$ for all $u \neq u' \in U$ and all $k \in (0,1)$.
\end{itemize}
Note that the second and third properties imply the triangle inequality: $m(u + u') \leq m(u) + m(u')$.

\begin{proposition}[D1]
    Consider a delegation game $D$ with measures $\IA(D)$, $\IC(D)$, $\CA(D)$, and $\CC(D)$. Holding fixed any three of the measures, then for any value $v \in [0,1]$ (or $\bm{v} \in [0,1]^n$ for $\IA$ or $\IC$), there is a game $D'$ such that the fourth measure takes value $v$ ($\bm{v}$) in $D'$ and the other measures retain their previous values.
\end{proposition}

\begin{proof}
    The case for the capabilities measures $\IC$ and $\CC$ is straightforward as they define how the agents \emph{play} the game, not any particular features of the game itself. Thus we may let $D' = D$ in these cases, and set $\IC$ or $\CC$ to any $\bm{v} \in [0,1]^n$ or $v \in [0,1]$ respectively.

    For the case of $\IA$ we leave fixed $\bm{u}$ and vary $\hat{\bm{u}}$. For each $i$, if $\bm{v}[i] = 1$ then we set $\uHat^i = u^i$, and if $\bm{v}[i] = 0$ then we set $\uHat^i = -u^i$. Otherwise, assuming $U$ is not one-dimensional, there are infinitely many choices of vectors $u \in U$ such that $m(u_\nu - u^i_\nu) = 2 - 2\bm{v}[i]$. It suffices to pick any such $u$ and set $\uHat^i = u$. 
    For the case of $\CA$ we vary both $\bm{u}$ and $\hat{\bm{u}}$, in order to preserve $\IA$. As before, if $v = 1$ then we set each $u^i_\nu = u_\nu$ for some $u \neq \bm{0}$. We then set each $\uHat^i$ as above, such that $m(\uHat^i_\nu - u^i_\nu) = 2 - 2\IA^i(D)$, ensuring that $\IA(D) = \IA(D')$. Similarly, if $v = 0$ then we set each $u^i$ such that $u^i \neq \bm{0}$ but $\mu^w = \bm{0}$, which gives:
    \begin{align*}
        \CA(D)
        &\coloneqq 1 - \sum_i \frac{m^i}{\sum_j m^j} \cdot m(\mu^w - u^i_\nu)\\
        &= 1 - \sum_i \frac{m^i}{\sum_j m^j} \cdot m(\bm{0} - u^i_\nu)\\
        &= 1 - \sum_i \frac{m^i}{\sum_j m^j} \cdot 1\\
        &= 1 - 1\\
        &= 0.
    \end{align*}
    We then once again adjust each $\uHat^i$ so that $m(\uHat^i_\nu - u^i_\nu) = 2 - 2\IA^i(D)$, using the argument from preceding case. 
    Finally, for the general case of $v \in (0,1)$, there are again an infinite number of choices of each $u^i \neq \bm{0}$ such that $(1 - v) \sum_j m^j = \sum_i m^i \cdot m(\mu^w - u^i_\nu)$. Given such a setting of $\bm{u}$, we may set $\bm{\uHat}$ as in the preceding cases.
\end{proof}

\begin{proposition}[D2]
    For any $1 \leq i \leq n$, $\IA^i = 1$ ($\IA^i = 0$) if and only if $\preceq^i = \hat{\preceq}^i$ ($\preceq^i = \hat{\succeq}^i$). Similarly $\CA = 1$ if and only if $\preceq^i = \preceq^j$ for every $i, j \in N$.
\end{proposition}

\begin{proof}
    We begin with the case of individual alignment by considering the $i^\th$ principal-agent pair, who have utility functions $\uHat^i$ and $u^i$, respectively. Recall that we write $c^i \coloneqq c(u^i)$ and $m^i \coloneqq m(u^i - c^i)$, and analogously for $\uHat^i$.

    Let us assume that $\preceq^i = \hat{\preceq}^i$. Then, by applying Lemma \ref{lem:nu-equivalence} to the case of $u^i$ and $\uHat^i$ we have that 
    $$\IA^i = 1 - \frac{1}{2}\cdot m(u_\nu - \uHat^i_\nu) = 1 - \frac{1}{2}\cdot m(\bm{0}) = 1,$$ 
    as required.
    Conversely, if $\IA^i = 1$ then by the positive definiteness of $m$ we must have $u^i_\nu = \uHat^i_\nu$ and hence (by Lemma \ref{lem:nu-equivalence}) that $\preceq^i = \hat{\preceq}^i$, as required.

    Next, we consider the case of maximal misalignment. In this case we have $\preceq^i = \hat{\succeq}^i$ and thus (by Lemma \ref{lem:nu-equivalence}) that $u^i_\nu = - \uHat^i_\nu$. Because $u^i_\nu \neq \bm{0} \neq \uHat^i_\nu$ then by the absolute homogeneity of $m$ and the fact that $m(u^i_\nu) = 1$:
    $$
    \IA^i
    = 1 - \frac{1}{2}\cdot m(u^i_\nu - \uHat^i_\nu)
    = 1 - \frac{1}{2}\cdot m(2u^i_\nu)
    = 1 - m(u^i_\nu)
    = 0,
    $$
    as required. On the other hand, if $\IA^i = 0$ then this implies that $m(u^i_\nu - \uHat^i_\nu) = 2$.
    The triangle inequality implies that $m(u^i_\nu - \uHat^i_\nu) \leq m(u^i_\nu) + m(- \uHat^i_\nu)$. Because $m$ is strictly convex, this inequality is an equality only when $u^i_\nu$ and $-\uHat^i_\nu$ are collinear, i.e. $u^i_\nu = - \lambda \uHat^i_\nu$ for some $\lambda$. As $m(u^i_\nu) = m(- \uHat^i_\nu) = 1$, this is possible only when $\lambda = 1$, and hence $u^i_\nu = - \uHat^i_\nu$. Thus, by \ref{lem:nu-equivalence}, we have $\preceq^i = \hat{\succeq}^i$ as required.

    Finally, we consider the case of perfect \emph{collective} alignment between two or more agents. First suppose that the agents all have identical preferences $\preceq$. Then, by Lemma \ref{lem:nu-equivalence}, this means that we have $u^i_\nu = u^j_\nu$ for all $i, j \in N$.
    
    Letting $u_\nu$ denote the common normalised utility, then we have $u^i = m^i u_\nu + c^i$ for each $1 \leq i \leq n$. The welfare representative utility is therefore given by:
    $$\mu^w \coloneqq \frac{\sum_i m^i u_\nu^i}{\sum_i m^i} = \frac{u_\nu \sum_i m^i}{\sum_i m^i} = u_\nu.$$
    Thus we have that:
    \begin{align*}
        \CA 
        &\coloneqq 1 - \sum_i \frac{m^i}{\sum_j m^j} \cdot m(\mu^w - u^i_\nu)\\
        &= 1 - \sum_i \frac{m^i}{\sum_j m^j} \cdot m(\mu^w - u_\nu)\\
        &= 1 - \sum_i \frac{m^i}{\sum_j m^j} \cdot m(\bm{0})\\
        &= 1,
    \end{align*}
    as required. Conversely, if we suppose that $\CA = 1$ then we must have that $m(\mu^w - u^i_\nu) = 0$ for all $1 \leq i \leq n$. By the triangle inequality we have:
    $$m(u_\nu^i - u_\nu^j) \leq m(u_\nu^i - \mu^w) + m(\mu^w - u_\nu^j) = 0,$$
    for any $i, j \in N$, and hence that $u_\nu^i = u_\nu^j$, by the positive definiteness of $m$. Thus, again by Lemma \ref{lem:nu-equivalence}, we have that all of the agents have the same preferences, concluding the proof.
\end{proof}

\begin{proposition}[D3]
    If $\IC = \bm{1}$ and $\CC = 1$ then any strategy $\bm{\sigma}$ the agents play is such that $w(\bm{\sigma}) = w_\star$. 
    If $\IA = \bm{1}$ then $\bm{\sigma}$ is Pareto-optimal for the principals.
    If, furthermore, $\hat{m}^i = r \cdot m^i$ for some $r > 0$ for all $i$, then $\hat{w}(\bm{\sigma}) = \hat{w}_\star$.
\end{proposition}

\begin{proof}
    Recall that if $\IC = \bm{1} - \eps$ and $\CC = \delta$ then the agents only play strategies $\bm{\sigma}$ such that
    $w(\bm{\sigma}) \geq w_{\eps} + \delta \cdot (w_\star - w_{\bm{0}})$. If $\IC = \bm{1}$ (and hence $\eps = \bm{0}$) and $\CC = \delta = 1$ then clearly $w(\bm{\sigma}) = w_\star$, as required.

    Now let us suppose that each agent is maximally aligned with its principal, i.e. $\IA^i = 1$ for every $1 \leq i \leq n$. By the construction in the proof of Lemma \ref{lem:nu-equivalence}, this implies that $u^i = a^i \hat{u}^i + b^i$ where $a^i \coloneqq \frac{m^i}{\hat{m}^i}$. Assume, for a contradiction, that $\bm{\sigma}$ is not Pareto-optimal for the principals. Then there exists some $\bm{\sigma}'$ and some $j \in N$ such that $\uHat^j(\bm{\sigma}') > \uHat^j(\bm{\sigma})$ and $\uHat^i(\bm{\sigma}') \geq \uHat^i(\bm{\sigma})$ for all $1 \leq i \leq n$. But then:
    $$u^j(\bm{\sigma}') =  a^j \hat{u}^j(\bm{\sigma}') + b^j > a^j \hat{u}^j(\bm{\sigma}) + b^j = u^j(\bm{\sigma}),$$
    and similarly $u^i(\bm{\sigma}') \geq u^i(\bm{\sigma})$ for all $1 \leq i \leq n$. Thus, we have that $w(\bm{\sigma}') = \frac{1}{n} \sum_i u^i(\bm{\sigma}') > \frac{1}{n} \sum_i u^i(\bm{\sigma}) = w(\bm{\sigma})$, contradicting the fact that $w_\star \coloneqq \max_{\bm{\sigma}} w(\bm{\sigma})$.

    Finally, let us assume that $\hat{m}^i = r \cdot m^i$ for every $1 \leq i \leq n$. Then $a^i = \frac{m^i}{\hat{m}^i} = \frac{1}{r}$ and so:
    \begin{align*}
        \hat{w}(\bm{\sigma})
        &= \frac{1}{n} \sum_i \uHat^i(\bm{\sigma})\\
        &= \frac{1}{n} \sum_i r \cdot \big( u^i(\bm{\sigma}) - b^i \big)\\
        &= \frac{r}{n} \sum_i u^i(\bm{\sigma}) - \frac{r}{n} \sum_i b^i\\
        &= r w(\bm{\sigma}) - \frac{r}{n} \sum_i b^i.
    \end{align*}
    This implies that $\argmax_{\bm{\sigma}} \hat{w}(\bm{\sigma}) = \argmax_{\bm{\sigma}} w(\bm{\sigma})$, and hence that $\hat{w}(\bm{\sigma}) = \hat{w}_\star$, which concludes the final part of the proof.
\end{proof}

\begin{proposition}[D4]
    If $\CA = 1$ then $w_\star = w_+$.
\end{proposition}

\begin{proof}
    Following the proof of Proposition \ref{prop:D2}, $\CA = 1$ implies that $u^i_\nu = \mu^w$ for every $1 \leq i \leq n$, and thus that $\preceq^i = \preceq^\mu$, where $\preceq^\mu$ denotes the preference ordering over mixed strategies induced by $\mu^w$. By Lemma \ref{lem:mu-welfare} we have that $\preceq^\mu = \preceq^w$, where $\preceq^w$ denotes the preference ordering over mixed strategies induced by $w$, and hence, by transitivity, that $\preceq^i = \preceq^w$ for each $1 \leq i \leq n$. Thus, we have $\argmax_{\bm{s}} w(\bm{s}) = \argmax_{\bm{s}} u^i(\bm{s})$ and therefore:
    $$\max_{\bm{s}} w(\bm{s})
    = \max_{\bm{s}} \frac{1}{n} \sum_i u^i(\bm{s})
    = \frac{1}{n} \sum_i \max_{\bm{s}} u^i(\bm{s})
    = w_+.$$
    By Lemma \ref{lem:pure-welfare}, $w_\star = \max_{\bm{s}} w(\bm{s})$, and so $w_\star = w_+$.
\end{proof}

\begin{proposition}[D5]
    Given a delegation game $D_1$, let $D_2$ be such that $\preceq^i_1 = \preceq^i_2$ and $\hat{\preceq}^i_1 = \hat{\preceq}^i_2$ for each $1 \leq i \leq n$. Then $\IA(D_1) = \IA(D_2)$ and $\IC(D_1) = \IC(D_2)$. Moreover, if $D_2$ is such that $\preceq^w_1 = \preceq^w_2$ and the $u^i$ are affine-independent, then $\CA(D_1) = \CA(D_2)$ and $\CC(D_1) = \CC(D_2)$ as well.
\end{proposition}

\begin{proof}
    First, note that the capability measures are defined extrinsically to the game, and so do not vary as the game varies. We therefore concern ourselves with the alignment measures.

    In the case of individual alignment, notice that if $\preceq^i_1 = \preceq^i_2$ then $\nu(u^i_1) = \nu(u^i_2)$ by Lemma \ref{lem:nu-equivalence}, and similarly for $\uHat^i_1$ and $\uHat^i_2$. Thus we have that $m\big( \nu(\uHat^i_1) - \nu(u^i_1) \big) = m\big( \nu(\uHat^i_2) - \nu(u^i_2) \big)$, and hence $\IA^i(D_1) = \IA^i(D_2)$. Applying the same argument to each $i$ gives $\IA(D_1) = \IA(D_2)$.
    
    For the case of collective alignment, then by Lemma \ref{lem:affine_transformation} we have $w_1 = a w_2 + b$ for some $a \in \R_{>0}$ and $b \in \R$. Because we also have $\nu(u^i_1) = \nu(u^i_2)$ for all $1 \leq i \leq n$, then:
    \begin{align*}
        \mu^w_1 \coloneqq \frac{\sum_i u^i_1 - c^i_1}{\sum_i m^i_1} = \frac{\sum_i m^i \nu(u^i_1)}{\sum_i m^i_1} = \frac{\sum_i m^i \nu(u^i_2)}{\sum_i m^i_1}.
    \end{align*}
    Combining these two statements and rearranging results in:
    $$\sum_i \left( \frac{m^i_1}{\sum_j m^j_1} - a \frac{m^i_2}{\sum_j m^j_2} \right) \nu(u^i_1) = b.$$
    Assuming that the $\nu(u^i_1)$ are not affine-dependent, then we must have $\frac{m^i_1}{\sum_j m^j_1} = a \frac{m^i_2}{\sum_j m^j_2}$ for each $1 \leq i \leq n$. Thus, we have $\mu^w_1 = \mu^w_2$, and therefore:
    \begin{align*}
        \CA(D_1)
        &\coloneqq 1 - \sum_i \frac{m^i_1}{\sum_j m^j_1} \cdot m\big(\mu^w_1 - \nu(u^i_1)\big)\\
        &= 1 - \sum_i \frac{m^i_2}{\sum_j m^j_2} \cdot m\big(\mu^w_2 - \nu(u^i_2)\big)\\
        &\eqqcolon \CA(D_2),
    \end{align*}
    which concludes the proof.
\end{proof}

\subsection{Bounding Welfare Regret}

Next, we provide proofs of the bounds on welfare regret from Section \ref{sec:bounds}.\\

\begin{proposition}
    Given a delegation game $D$, we have:
    $$\wHat_\star - \wHat(\bm{\sigma}) \leq \frac{1}{n} \sum_i r^i \big( u^i(\hat{\bm{s}}_\star) - u^i(\bm{\sigma}) \big) + \frac{4K}{n} \hat{\bm{m}}^\top(\bm{1} - \IA),$$
    where $r^i \coloneqq \frac{\hat{m}^i}{m^i}$, $\hat{\bm{m}}[i] = \hat{m}^i$, $K$ satisfies $\norm{u_\nu - u'_\nu}_\infty \leq K \cdot m(u_\nu - u'_\nu)$ for any $u, u' \in U$, and $\wHat(\hat{\bm{s}}_\star) = \wHat_\star$.
\end{proposition}

\begin{proof}
    To begin with, let us consider the principal's welfare regret from a pure strategy profile $\bm{s}$. Let $\hat{\bm{s}}_\star$ be such that $\wHat(\hat{\bm{s}}_\star) = \wHat_\star$, which we know must exist from Lemma \ref{lem:pure-welfare}. Then we have:
    \begin{align*}
        &\wHat_\star - \wHat(\bm{s})\\
        = &\wHat(\hat{\bm{s}}_\star) - \wHat(\bm{s})\\
        = &\frac{1}{n} \sum_i \uHat^i(\hat{\bm{s}}_\star) - \frac{1}{n} \sum_i \uHat^i(\bm{s})\\
        = &\frac{1}{n}\left( \sum_i \hat{m}^i \uHat^i_\nu(\hat{\bm{s}}_\star) + \hat{c}^i \right) - \frac{1}{n}\left( \sum_i \hat{m}^i \uHat^i_\nu(\bm{s}) + \hat{c}^i \right)\\
        = &\frac{1}{n} \left( \sum_i \hat{m}^i \uHat^i_\nu(\hat{\bm{s}}_\star) - \sum_i \hat{m}^i \uHat^i_\nu(\bm{s}) \right).
    \end{align*}
    Because $m$ is a norm and $U$ is finite-dimensional, then there is some finite $K$ such that $\norm{u_\nu - u'_\nu}_\infty \leq K \cdot m(u_\nu - u'_\nu)$ for any $u, u' \in U$. Applying this bound to each $\uHat^i_\nu$ and $u^i_\nu$, and writing $d^i \coloneqq m(\uHat^i_\nu - u^i_\nu)$ we have:
    \begin{align*}
        &\sum_i \hat{m}^i \uHat^i_\nu(\hat{\bm{s}}_\star) - \sum_i \hat{m}^i \uHat^i_\nu(\bm{s})\\
        \leq &\sum_i \hat{m}^i \cdot \left( u^i_\nu(\hat{\bm{s}}_\star) + K d^i \right) - \sum_i \hat{m}^i \cdot \left( u^i_\nu(\bm{s}) - K d^i \right)\\
        = &\sum_i \hat{m}^i u^i_\nu(\hat{\bm{s}}_\star) - \sum_i \hat{m}^i u^i_\nu(\bm{s}) + 2 K \sum_i \hat{m}^i d^i\\
        = &\sum_i \hat{m}^i \frac{u^i(\hat{\bm{s}}_\star) - c^i}{m^i} - \sum_i \hat{m}^i \frac{u^i(\bm{s}) - c^i}{m^i} + 2 K \sum_i \hat{m}^i d^i\\
        = &\sum_i r^i \big( u^i(\hat{\bm{s}}_\star) - u^i(\bm{s}) \big) + 2 K \sum_i \hat{m}^i d^i,
    \end{align*}
    where $r^i \coloneqq \frac{\hat{m}^i}{m^i}$. Finally, note that for any $\bm{\sigma}$, we have:
    $$w(\bm{\sigma}) \coloneqq \expect_{\bm{\sigma}} [w(\bm{s})] = \expect_{\bm{\sigma}} \Big[\frac{1}{n} \sum_i u^i(\bm{s})\Big] = \frac{1}{n} \sum_i \expect_{\bm{\sigma}} \big[ u^i(\bm{s}) \big],$$
    by the linearity of expectation. Applying this to the bound above, substituting in $d^i = 2(1 - \IA^i)$, and writing the final sum over players in vector notation gives:
    \begin{align*}
        \wHat_\star - \wHat(\bm{\sigma})
        &\leq \frac{1}{n} \left(\sum_i r^i \big( u^i(\hat{\bm{s}}_\star) - u^i(\bm{\sigma}) \big) + 2 K \sum_i \hat{m}^i d^i\right)\\
        &= \frac{1}{n} \sum_i r^i \big( u^i(\hat{\bm{s}}_\star) - u^i(\bm{\sigma}) \big) + \frac{4K}{n} \hat{\bm{m}}^\top(\bm{1} - \IA),
    \end{align*}
    where $\hat{\bm{m}}[i] = \hat{m}^i$, completing the proof.
\end{proof}

\begin{theorem}
    Given a delegation game $D$, we have that:
    \begin{align*}
        \wHat_\star - \wHat(\bm{\sigma})
        &\leq \frac{4K}{n} \hat{\bm{m}}^\top(\bm{1} - \IA) + r^* \left( (w_{\bm{0}} - w_{\eps}) \right.\\
        &+ \left. (1 - \CC) (w_\star - w_{\bm{0}}) \right) + R(\bm{\sigma}),
    \end{align*}
    where $\IC = \bm{1} - \eps$, 
    $r^* \in [\min_i r^i, \max_i r^i]$,
    $R(\bm{\sigma}) \coloneqq \frac {1}{n} \sum_i (\hat m^i - r^*m^i) \left( u_\nu^i(\bm{\sHat}_\star) - u_\nu^i(\bm{\sigma}) \right)$ is a remainder accounting for collective misalignment and unequal $r^i$,
    and $K$ and $r^i$ are defined as in Proposition \ref{prop:IA-bound}.
    Note that when all $r^i$ are equal \emph{or} $\CA = 1$ then there is an $r^*$ with $R(\bm{\sigma}) = 0$.
\end{theorem}
\begin{proof}
    First, observing that by definition $w_\star \ge w(\bm{\sHat}_\star)$, we see that for any $r^* > 0$:
    \begin{align*}
        & \frac {1}{n} \sum_i r^i\left( u^i(\bm{\sHat}_\star) - u^i(\bm{\sigma}) \right) \\
        = & \frac {1}{n} \sum_i r^* \left( u^i(\bm{\sHat}_\star) - u^i(\bm{\sigma}) \right) + \frac {1}{n} \sum_i (r^i - r^*) \left( u^i(\bm{\sHat}_\star) - u^i(\bm{\sigma}) \right) \\
        = & r^* (w(\bm{\sHat_\star}) - w(\bm{\sigma})) + \frac {1}{n} \sum_i (\hat m^i - r^*m^i) \left( u_\nu^i(\bm{\sHat}_\star) - u_\nu^i(\bm{\sigma}) \right) \\
        \le & r^* (w_\star - w(\bm{\sigma})) + \frac {1}{n} \sum_i (\hat m^i - r^*m^i) \left( u_\nu^i(\bm{\sHat}_\star) - u_\nu^i(\bm{\sigma}) \right) \\
        = & r^* (w_\star - w(\bm{\sigma})) + R(\bm{\sigma}),
    \end{align*}
    where $R(\bm{\sigma}) \coloneqq \frac {1}{n} \sum_i (\hat m^i - r^*m^i) \left( u_\nu^i(\bm{\sHat}_\star) - u_\nu^i(\bm{\sigma}) \right)$. 
    This is a weighted combination of the agent welfare regret $w_\star - w(\bm{\sigma})$ with a `fairness remainder' term $R(\bm{\sigma})$. The remainder can be seen to be intimately related to the different ratios $r^i \coloneqq \frac {\hat m^i} {m^i}$, which determine when Pareto improvements are considered gains in agent welfare. When all $r^i$ are identical, a choice of $r^* = r^i$ causes $R(\bm{\sigma})$ to vanish, which we refer to as the `perfectly calibrated' case. (We discuss several continuous bounds on this remainder in Appendix \ref{app:auxiliary_results}.)
    Now, from Definition \ref{def:CC} we have $w(\bm{\sigma}) \geq w_{\eps} + \CC \cdot (w_\star - w_{\bm{0}})$, which implies that:
    \begin{align*}
        w_\star - w(\bm{\sigma})
        &\leq w_\star - w_{\eps} - \CC \cdot (w_\star - w_{\bm{0}})\\
        &= (1 - \CC) \cdot w_\star - w_{\eps} + \CC \cdot w_{\bm{0}}\\
        &= (1 - \CC) \cdot w_\star - w_{\eps} - (1 - \CC) \cdot w_{\bm{0}} + w_{\bm{0}}\\
        &= (1 - \CC) (w_\star - w_{\bm{0}}) + (w_{\bm{0}} - w_{\eps}).
    \end{align*}
    Substituting this inequality into the previous inequality and then into the bound from Proposition \ref{prop:IA-bound} gives:
    \begin{align*}
        \wHat_\star - \wHat(\bm{\sigma})
        &\leq \frac{4K}{n} \hat{\bm{m}}^\top(\bm{1} - \IA) + r^* \left( (w_{\bm{0}} - w_{\eps}) \right.\\
        &+ \left. (1 - \CC) (w_\star - w_{\bm{0}})\right) + R(\bm{\sigma}).
    \end{align*}
\end{proof}

\begin{proposition}
    Given a game $G$ with collective alignment $\CA$, then $w_+ - w_\star \leq \frac{K \sum_i m^i}{n} (1 - \CA)$, where $K$ is defined as in Proposition \ref{prop:IA-bound}.
\end{proposition}

\begin{proof}
    To begin, we first choose $\bm{s}^i_+ \in \argmax_{\bm{s}} u^i(\bm{s})$ for each $i$. Second, because $m$ is a norm and $U$ is finite-dimensional, then there is some finite $K$ such that $\norm{u - u'}_\infty \leq K \cdot m(u - u')$ for any $u, u' \in U$.
    Using these two pieces of information, we have:
    \begin{align*}
        w_+ 
        &\coloneqq \frac{1}{n} \sum_i u^i(\bm{s}^i_+)\\
        &= \frac{1}{n} \sum_i \big( m^i u^i_\nu(\bm{s}^i_+) + c^i \big)\\
        &= \frac{1}{n} \sum_i \big( m^i \mu^w(\bm{s}^i_+) + c^i \big) +  \frac{1}{n} \sum_i m^i\big( u^i_\nu(\bm{s}^i_+) - \mu^w(\bm{s}^i_+) \big)\\
        &\leq \frac{1}{n} \sum_i \big(m^i \mu^w(\bm{s}^i_+) + c^i \big) +  \frac{1}{n} \sum_i m^i \norm{u^i_\nu - \mu^w}_\infty\\
        &\leq \frac{1}{n} \sum_i \big(m^i \mu^w(\bm{s}^i_+) + c^i \big) + \frac{K}{n} \cdot \sum_i m^i \cdot m(\mu^w - u^i_\nu).
    \end{align*}
    Next, let us choose some $\bm{s}_\star \in \argmax_{\bm{s}} w(\bm{s})$. Then clearly $w(\bm{s}^i_+) \leq w(\bm{s}_\star) = w_\star$ for any $1 \leq i \leq n$, where the second equality follows from Lemma \ref{lem:pure-welfare}. Hence, by Lemma \ref{lem:mu-welfare}, we also have $\mu^w(\bm{s}^i_+) \leq \mu^w(\bm{s}_\star)$. Using this fact, we have:
    \begin{align*}
        &\frac{1}{n} \sum_i \big(m^i \mu^w(\bm{s}^i_+) + c^i \big)\\
        \leq &\frac{1}{n} \sum_i \big(m^i \mu^w(\bm{s}_\star) + c^i \big)\\
        = &\frac{1}{n} \mu^w(\bm{s}_\star) \sum_i m^i + \frac{1}{n} \sum_i c^i\\
        = &\frac{1}{n} \frac{\sum_i \big( u^i(\bm{s}_\star) - c^i \big)}{\sum_i m^i} \sum_i m^i +  \sum_i c^i\\
        = &\frac{1}{n} \sum_i \big(u^i(\bm{s}_\star) - c^i\big) + \frac{1}{n} \sum_i c^i\\
        = &\frac{1}{n} \sum_i u^i(\bm{s}_\star)\\
        = &w_\star.
    \end{align*}
    Substituting this and using the fact that $\sum_i m^i \cdot m(\mu^w - u^i_\nu) = (1 - \CA) \sum_i m^i$ in to our our previous inequality we have:
    $$w_+ \leq w_\star + \frac{K \sum_i m^i}{n} (1 - \CA),$$
    from which our result follows by subtracting $w_\star$.
\end{proof}

\subsection{Auxiliary Results}
\label{app:auxiliary_results}

\begin{lemma}
    For any $u,u' \in U$, $u_\nu = u'_\nu$ if and only if $\preceq = \preceq'$.
\end{lemma}

\begin{proof}
    First, recall that, by Lemma \ref{lem:affine_transformation}, $\preceq = \preceq'$ if and only if there exist $a \in \R_{>0}$ and $b \in \R$ such that $u = a u' + b$. Using the fact that $c$ is affine-equivariant we have:
    \begin{align*}
        u - c 
        & = au' + b - c(au' + b) \\
        & = au' + b - a \cdot c(u') - b \\
        & = a \cdot (u' - c').
    \end{align*}
    Next, by appealing to the absolute homogeneity of $m$ and the fact that $a > 0$ we have:
    \begin{align*}
        u_\nu
        &= \frac{u - c}{m}\\
        &= \frac{a \cdot (u' - c')}{m\big(a \cdot (u' - c')\big)}\\
        &= \frac{a \cdot (u' - c')}{\abs{a} \cdot m \big( u' - c'\big)}\\
        &= \frac{u' - c'}{m'}\\
        &= u'_\nu,
    \end{align*}
    as required. Conversely, we see that:
    \begin{align*}
        u_\nu = u'_\nu
        &\implies \frac{u - c}{m} = \frac{u' - c'}{m'}\\
        &\implies u = \left( \frac{m}{m'} \right) \cdot u' + 
        \left(c - \frac{m}{m'} c'\right).
    \end{align*}
    Letting $a' \coloneqq \frac{m}{m'} > 0$ and $b' \coloneqq c - \frac{m}{m'} c'$ then we have $u = a' u' + b'$ and hence $\preceq = \preceq'$ (by Lemma \ref{lem:affine_transformation}), as required.
\end{proof}

\begin{lemma}
    There exists a (two-player, two-action) delegation game $D$ such that for any $x > 0$, however small, even if $\IA(D) = \bm{1}$ and $\IC(D) = \bm{1}$, we have only one NE $\bm{\sigma}$, and $\frac{\wHat_\star - \wHat(\bm{\sigma})}{\wHat_+ - \wHat_-} = 1-x$.
\end{lemma}

\begin{proof}
    Consider the Prisoner's Dilemma as specified in Figure \ref{fig:PD}, which we take to be the principal game $\GHat$. Note that $\wHat_+ = 1$ and $\wHat_- = 0$. If $\IA(D) = \bm{1}$ then $\hat{\preceq}^i = \preceq^i$ by Proposition \ref{prop:D2}, meaning that the NEs of $G$ are the same as those of $\GHat$. 
    Because $\IC(D) = \bm{1}$ then we consider only these NEs, which in this case is simply $\bm{s} = (B,B)$, giving $\wHat(\bm{s}) = \frac{x}{2}$. However, $\max_{\bm{\sigma}} \wHat(\bm{\sigma}) = 1 - \frac{x}{2}$, and thus we have $\frac{\wHat_\star - \wHat(\bm{\sigma})}{\wHat_+ - \wHat_-} = 1-x$.
\end{proof}

\begin{figure*}[t]
    \centering
    \begin{subfigure}[c]{0.4\textwidth}
        \centering
        \begin{subfigure}{\textwidth}
            \centering
            \renewcommand\arraystretch{1.7}
            \begin{tabular}{c|c|c|c}
                \multicolumn{1}{c}{} & \multicolumn{1}{c}{$A$} & \multicolumn{1}{c}{$B$} &\\\cline{2-3}
                $A$ & $1-\frac{x}{2}, 1-\frac{x}{2}$ & $0, 1$ &\phantom{$A$}\\\cline{2-3}
                $B$ & $1, 0$ & $\frac{x}{2}, \frac{x}{2}$ &\phantom{$B$}\\\cline{2-3}
            \end{tabular}
            \caption{}
        \label{fig:PD}
        \end{subfigure}

        \begin{subfigure}{\textwidth}
            \centering
            \renewcommand\arraystretch{1.7}
            \begin{tabular}{c|c|c|c|c}
                \multicolumn{1}{c}{} & \multicolumn{1}{c}{$A$} & \multicolumn{1}{c}{$B$} & \multicolumn{1}{c}{$C$}\\\cline{2-4}
                $A$ & $1, 1$            & $x, 1-\epsilon^2$                  & $0, 0$ &\phantom{$A$}\\\cline{2-4}
                $B$ & $1-\epsilon^1, x$ & $(1-\epsilon^1)x, (1-\epsilon^2)x$ & $x, 0$ &\phantom{$B$}\\\cline{2-4}
                $C$ & $0, 0$            & $0, x$                             & $0, 0$ &\phantom{$B$}\\\cline{2-4}
            \end{tabular}
            \caption{}
        \label{fig:fragile}
        \end{subfigure}
    \end{subfigure}
    \begin{subfigure}[c]{0.55\textwidth}
        \centering
        \renewcommand\arraystretch{1.7}
        \begin{tabular}{c|c|c|c|c|c|c}
            \multicolumn{1}{c}{} & \multicolumn{1}{c}{$A$} & \multicolumn{1}{c}{$B$} & \multicolumn{1}{c}{$C$} & \multicolumn{1}{c}{$\cdots$} & \multicolumn{1}{c}{$Z$} &\\\cline{2-6}
            $A$ & $1-x,1-x$ & $1-3x,1$ & $0,0$ & $\cdots$ & $0,0$ & \phantom{$A$}\\\cline{2-6}
            $B$ & $1,1-3x$ & $1-2x,1-2x$ & $1-4x,1-x$ & $\cdots$ & $0,0$ & \phantom{$B$}\\\cline{2-6}
            $C$ & $0,0$ & $1-x,1-4x$ & $1-3x,1-3x$ & $\cdots$ & $0,0$ &\phantom{$C$}\\\cline{2-6}
            $\vdots$ & $\vdots$ & $\vdots$ & $\vdots$ & $\ddots$ & $\vdots$ &\phantom{$D$}\\\cline{2-6}
            $Z$ & $0,0$ & $0,0$ & $0,0$ & $\cdots$ & $x,x$ &\phantom{$Z$}\\\cline{2-6}
        \end{tabular}
        \caption{}
    \label{fig:TD}
    \end{subfigure}
    \caption{(a) A Prisoner's Dilemma leading to arbitrarily high welfare regret for the principals, despite perfect control of each agent by its principal; (b) a game in which the welfare in the worst $\epsNE$ is much lower than in the worst NE.; and (c) a Traveller's Dilemma leading to arbitrarily high welfare regret for the principals, despite perfect control of each agent by its principal, and near-perfect collective alignment.}
\end{figure*}

\begin{lemma}
    For any $\bm{\sigma}, \bm{\sigma}' \in \bm{\Sigma}$, $\mu^w(\bm{\sigma}) \leq \mu^w(\bm{\sigma}')$ if and only if $w(\bm{\sigma}) \leq w(\bm{\sigma}')$.
\end{lemma}

\begin{proof}
    By Definition \ref{def:CA} we have that:
    \begin{align*}
        \mu^w 
        &\coloneqq \frac{\sum_i u^i - c^i}{\sum_i m^i}\\
        &= \frac{1}{\sum_i m^i} \sum_i u^i - \frac{\sum_i c^i}{\sum_i m^i}\\
        &= \frac{n}{\sum_i m^i} w - \frac{\sum_i c^i}{\sum_i m^i}.
    \end{align*}
    The result follows by the linearity of expectation, taken with respect to $\bm{s}$. Indeed, setting $a \coloneqq \frac{n}{\sum_i m^i}$ and $b = - \frac{\sum_i c^i}{\sum_i m^i}$ then by Lemma \ref{lem:affine_transformation} we have that $\preceq^\mu = \preceq^w$, where $\preceq^\mu$ and $\preceq^w$ denote the preference ordering over mixed strategies induced by $\mu^w$ and $w$ respectively.
\end{proof}

\begin{lemma}
    There exists a family of (two-player, $k$-action) delegation games $D$ such that even if $\IA^i(D) = 1$, $\IC^i(D) = 1$ for each agent, and there is only one NE $\bm{\sigma}$, we have $\lim_{k \to \infty} \CA(D) = 1$ but $\lim_{k \to \infty} \frac{\wHat_\star - \wHat(\bm{\sigma})}{\wHat_+ - \wHat_-} = 1$.
\end{lemma}

\begin{proof}
    Consider the Traveller's Dilemma as specified in Figure \ref{fig:TD}, which we take to be the principal game $\GHat$. As argued in the proof of Lemma \ref{lem:PD}, the fact that $\IA(D) = \bm{1}$ means that $\NE(G) = \NE(\GHat)$, and because $\IC(D) = \bm{1}$ then we need only consider the NEs of $G$, which in this case is simply $\bm{s} = (E,E)$, and (as in the case of Lemma \ref{lem:PD}) leads to a welfare regret of $1-2x$. If the game has $k$ actions for each player, and we assume that the payoffs for each agent are normalised such that $\wHat_+ = 1$ and $\wHat_- = 0$ then we have $x = \frac{1}{k+1}$. Thus, we have:
    $$\lim_{k \to \infty} \frac{\wHat_\star - \wHat(\bm{\sigma})}{\wHat_+ - \wHat_-} = \lim_{k \to \infty} \frac{1-\frac{1}{k+1}}{1 - 0} = 1.$$

    Now we consider the second limit. First note that $\norm{u^1 - u^2}_\infty = 3x$ and, as $c$ is affine-equivariant, then $\norm{c^1 - c^2}_\infty = 3x$ too. Because $m$ is a norm and $U$ is finite-dimensional, then there is some finite $L$ such that $m(u - u') \leq L \cdot \norm{u - u'}_\infty$ for any $u, u' \in U$, and therefore:
    \begin{align*}
        m^1 - m^2
        &= m(u^1 - c^1) - m(u^2 - c^2)\\
        &\leq m\big((u^1 - c^1) - (u^2 - c^2) \big)\\
        &\leq L \cdot \norm{(u^1 - c^1) - (u^2 - c^2)}_\infty\\
        &\leq L \cdot \norm{u^1 - u^2}_\infty + L \cdot \norm{c^1 - c^2}_\infty\\
        &= 6Lx.
    \end{align*}
    Hence, let use denote $u \coloneqq \lim_{k \to \infty} u^1 = \lim_{k \to \infty} u^2$ and similarly for $m$ and $c$. Then, assuming that $m \neq 0$ (otherwise the proof is complete), we compose limits to see that:
    $$\lim_{k \to \infty} u^1_\nu = \lim_{k \to \infty} \frac{u^1 - c^1}{m^1} = \frac{u - c}{m} = \lim_{k \to \infty} \frac{u^2 - c^2}{m^2} = \lim_{k \to \infty} u^2_\nu.$$
    To conclude the proof, observe that by Lemma \ref{lem:nu-equivalence} we have $\lim_{k \to \infty} \preceq^1 = \lim_{k \to \infty} \preceq^2$ and hence (by Proposition \ref{prop:D2}) that $\lim_{k \to \infty} \CA = 1$.
\end{proof}

\begin{lemma}
    For any $\eps \succ \bm{0}$, there exists a game $G$ such that $w_{\bm{0}} = w_+$ but for any $x > 0$, however small, $w_{\eps} - w_- < x$.
\end{lemma}

\begin{proof}
    Consider the game specified in Figure \ref{fig:fragile}. Observe that the only NE is given by $(A,A)$, leading to welfare $w_{\bm{0}} = w_+ = 1$. In contrast, note that $(B,B)$ is an $\epsNE$, as neither player can deviate to obtain a greater fraction than $(1 - \epsilon^i)$ of the utility available when playing their best response ($x$). Given this, we have:
    $$w_{\eps} - w_- < \left(1 - \frac{\epsilon^1 + \epsilon^2}{2}\right) x - 0 < x.$$
\end{proof}

\begin{proposition}
    Given a game $G$, if $u^i(\bm{s})$ for every $i$ and $\bm{s} \in S$, and $d$ has maximal support over the outcomes generated when agents act together/alone (respectively), then: 
    \begin{align*}
        \CC &\leq \lim_{\abs{\D} \to \infty} \frac{\min_{\bm{s} \in \D(\bm{S})} w(\bm{s})}{\max_{\bm{s} \in \D(\bm{S})} w(\bm{s})}, \text{ and }\\
        \IC^i &\leq \lim_{\abs{\D} \to \infty} \min_{\bm{s} \in \D(\bm{S})} \frac{u^i(\bm{s})}{\max_{\tilde{s}^i \in \D(S^i)} u^i(\bm{s}^{-i},\tilde{s}^i)}.
    \end{align*}
\end{proposition}

\begin{proof}
    Let us begin by considering the case in which the agents are observed acting together. Given their cooperative capabilities $\CC$ we observe outcomes $\bm{s}$ such that:
    $$w_{\eps} + \CC \cdot (w_\star - w_{\bm{0}}) \leq w(\bm{s}) \leq w_{\eps} + (w_\star - w_{\bm{0}}).$$
    If $d$ has maximal support over the set of such outcomes, then we have:
    \begin{align*}
        &\lim_{\abs{\D} \to \infty} \max_{\bm{s} \in \D(\bm{S})} w(\bm{s}) = w_{\eps} + (w_\star - w_{\bm{0}}),\\
        &\lim_{\abs{\D} \to \infty} \min_{\bm{s} \in \D(\bm{S})} w(\bm{s}) = w_{\eps} + \CC \cdot (w_\star - w_{\bm{0}}).
    \end{align*}
    If $\lim_{\abs{\D} \to \infty} \max_{\bm{s} \in \D(\bm{S})} w(\bm{s}) = 0$, then we must have $u^i(\bm{s}) = 0$ for every $1 \leq i \leq n$ and $\bm{s} \in S$, in which case it is meaningless to talk about the agents' cooperative capabilities, so let us assume otherwise.
    Subtracting $w_{\eps}$ (which, by assumption, is non-negative), cancelling the terms $(w_\star - w_{\bm{0}})$ in these equalities, and composing limits leads to:
    $$\CC \leq \lim_{\abs{\D} \to \infty} \frac{\min_{\bm{s} \in \D(\bm{S})} w(\bm{s})}{\max_{\bm{s} \in \D(\bm{S})} w(\bm{s})}.$$
    Similarly, in the case where each agent is observed acting alone, then given their individual capabilities $\IC^i$ we observe outcomes $\bm{s}$ such that:
    $$u^i_\bullet(\bm{s}^{-i}) + \IC^i \cdot \big( u^i_\star(\bm{s}^{-i}) - u^i_\bullet(\bm{s}^{-i}) \big) \leq u^i(\bm{s}) \leq u^i_\star(\bm{s}^{-i}),$$
    where we define:
    \begin{align*}
        u^i_\bullet(\bm{s}^{-i}) &\coloneqq \min_{\tilde{s}^i} u^i(\bm{s}^{-i}, \tilde{s}^i),\\
        u^i_\star(\bm{s}^{-i}) &\coloneqq \max_{\tilde{s}^i} u^i(\bm{s}^{-i}, \tilde{s}^i).
    \end{align*}
    If $d$ has maximal support over the set of such outcomes, including over \emph{all} pure partial strategy profiles $\bm{s}^{-i}$, then we have:
    \begin{align*}
        \IC^i 
        &\leq \lim_{\abs{\D} \to \infty} \min_{\bm{s} \in \D(\bm{S})} \frac{u^i(\bm{s}) - u^i_\bullet(\bm{s}^{-i})}{u^i_\star(\bm{s}^{-i}) - u^i_\bullet(\bm{s}^{-i})}\\
        &\leq \lim_{\abs{\D} \to \infty} \min_{\bm{s} \in \D(\bm{S})} \frac{u^i(\bm{s})}{u^i_\star(\bm{s}^{-i})}\\
        &= \lim_{\abs{\D} \to \infty} \min_{\bm{s} \in \D(\bm{S})} \frac{u^i(\bm{s})}{\max_{\tilde{s}^i \in \D(S^i)} u^i(\bm{s}^{-i},\tilde{s}^i)},
    \end{align*}
    which concludes the proof.
\end{proof}

\begin{lemma}
    \label{lem:affine_transformation}
    Assuming that a game $G$ has at least three outcomes, then for utility functions $u$ and $u'$ with corresponding preference orderings $\preceq$ and $\preceq'$ respectively, $\preceq = \preceq'$ if and only if there exist $a \in \R_{>0}$ and $b \in \R$ such that $u = a u' + b$. Similarly, $\preceq = \succeq'$ if and only if there exist $a \in \R_{>0}$ and $b \in \R$ such that $u = - a u' + b$.
\end{lemma}

\begin{proof}
    See, e.g. Proposition 6.B.2 in \cite{MasColell1995}, for a proof of the first half of this lemma. The second half can be proved analogously.
\end{proof}

\begin{lemma}
    \label{lem:pure-welfare}
    Given a game $G$, there is a pure strategy $\bm{s}_\star$ such that $w_\star(G) = w(\bm{s}_\star)$.
\end{lemma}

\begin{proof}
    Recall that $w_\star(G) \coloneqq \max_{\bm{\sigma}} w(\bm{\sigma})$. Let $\bm{\sigma}$ be such that $w(\bm{\sigma}) = w_\star(G)$, and consider the set $\bm{S}' \subseteq \bm{S}$ of pure strategy profiles with positive support in $\bm{\sigma}$.
    Clearly, for any $\bm{s}, \tilde{\bm{s}} \in \bm{S}$ we must have $w(\bm{\sigma}) = w(\tilde{\bm{s}})$. Otherwise, let $\bm{s}_\star \in \argmax_{\bm{s}_\star \in \bm{S}'}$. Then $w(\bm{s}_\star) > w(\bm{\sigma})$, a contradiction.
    But if $w(\bm{\sigma}) = w(\tilde{\bm{s}})$ for any $\bm{s}, \tilde{\bm{s}} \in \bm{S}$, then $w(\bm{\sigma}) = \sum_{\bm{s} \in \bm{S}'} \bm{\sigma}(\bm{s}) w(\bm{s}) = w(\bm{s})$ for any $\bm{s} \in \bm{S}'$. Choosing one such $\bm{s}$ and denoting it $\bm{s}_\star$ completes the proof.
\end{proof}

\subsubsection{Bounds on Welfare `Fairness' Remainder}

In Theorem \ref{thm:capabilities-bound} we remarked that $R(\bm{\sigma})$ -- a `fairness' remainder term accounting for collective misalignment and unequal $r^i$ -- can be reduced to zero when either $\CA = 1$ or all $r^i$ are equal, where recall that:
$$R(\bm{\sigma}) \coloneqq \frac {1}{n} \sum_i (\hat m^i - r^*m^i) \left( u_\nu^i(\bm{\sHat}_\star) - u_\nu^i(\bm{\sigma}) \right).$$
Here we provide more detail substantiating those claims.
First, making no assumptions about collective alignment, we can see that when all $r^i$ are equal, we must have $r^* = r^i$, and therefore $R(\bm \sigma) = 0$.
Alternatively, we can bound $R(\bm \sigma)$ by appealing to the collective alignment of the agents, which, recall, is given in Definition \ref{def:CA} as:
$$\CA(D) = 1 - \sum_i \frac{m^i}{\sum_j m^j} \cdot m(\mu^w - u^i_\nu),$$
where $\mu^w \coloneqq \frac{\sum_i u^i - c^i}{\sum_i m^i}$.
Thus, when $\CA = 1$, all $u_\nu^i$ are equal to $\mu^w$, and so $u_\nu^i(\bm{\sHat}_\star) - u_\nu^i(\bm{\sigma})$ is constant in $i$. By choosing $r^* = \frac{\sum_i \hat m^i} {\sum_i m^i}$ so that larger and smaller deviations in $r^i$ cancel, we again have $R(\bm \sigma) = 0$.

More generally, let us define $d_w^i \coloneqq m(u_\nu^i - \mu^w)$.\footnote{Note that we can express $\CA = 1 - \frac{\sum_i m^i d_w^i}{\sum_i m^i}$ as a weighted average of these terms.} Incorporating both principles, we have a bound in terms of both calibration and collective alignment.
\begin{proposition}
    \label{prop:calibration-CA-bound}
    If $r^* = \frac{\sum_i \hat m^i} {\sum_i m^i}$, then:
    $$
    R(\bm \sigma) \le \frac {2K} {n} \sum_i \left|\hat m^i - r^*m^i\right| d_w^i,
    $$
    where $K$ is defined as in Proposition \ref{prop:IA-bound}.
\end{proposition}
\begin{proof}
    First, consider the term $(\hat m^i - r^*m^i) \left( u_\nu^i(\bm \sHat_\star) - u_\nu^i(\bm \sigma) \right)$. The coefficient $(\hat m^i - r^*m^i)$ is not necessarily positive, but irrespective of its sign, we have:
    \begin{align*}
        &(\hat m^i - r^*m^i) \left( u_\nu^i(\bm \sHat_\star) - u_\nu^i(\bm \sigma) \right) \\
        \le &(\hat m^i - r^*m^i) \left( \mu_w(\bm \sHat_\star) - \mu_w(\bm \sigma) \right) + 2K \left|\hat m^i - r^*m^i\right| d_w^i.
    \end{align*}
    Substituting this inequality into $R(\bm{\sigma})$ gives:
    \begin{align*}
        R(\bm{\sigma}) & \coloneqq \frac {1}{n} \sum_i (\hat m^i - r^*m^i) \left( u_\nu^i(\bm{\sHat}_\star) - u_\nu^i(\bm{\sigma}) \right) \\
            & \le \frac 1 n \sum_i (\hat m^i - r^*m^i) \left( \mu_w(\bm \sHat_\star) - \mu_w(\bm \sigma) \right) \\
            & + \frac {2K} {n} \sum_i \left|\hat m^i - r^*m^i\right| d_w^i \\
            & = \frac 1 n \left( \mu_w(\bm \sHat_\star) - \mu_w(\bm \sigma) \right) \sum_i (\hat m^i - r^*m^i) \\
            & + \frac {2K} {n} \sum_i \left|\hat m^i - r^*m^i\right| d_w^i \\
            & = \frac {2K} {n} \sum_i \left|\hat m^i - r^*m^i\right| d_w^i,
    \end{align*}
    where the last equality follows because:
    $$
    \sum_i (\hat m^i - r^*m^i) = \sum_i \hat m^i - \left(\frac{\sum_i \hat m^i} {\sum_i m^i}\right) \sum_i m^i = 0,
    $$
    by our choice of $r^*$.
\end{proof}

Note that this bound is also equal to zero when either $\CA = 1$ or all $r^i$ are equal.
Various other related bounds on $R(\bm \sigma)$ are possible for different choices of $r^*$. If we have no information about the collective alignment, we can either adapt the preceding bound conservatively, or we can produce a similar bound with reference only to the $r^i$ terms. In the latter case, the choice of $r^*$ is freer, giving rise to a family of bounding planes, whose envelope produces a non-linear (piecewise affine) bound. See Figure \ref{fig:3d-plots} in Section \ref{sec:alt_vis} for an example of this.

\section{Additional Experiments}
\label{app:addition_exps}

In this section, we provide further details on the experiments we ran, and include additional plots.\footnote{Our code is available at \url{https://github.com/lrhammond/delegation-games}.}

\subsection{Empirical Validation}
\label{app:more_games}

In this section, we provide further empirical validation of our theoretical results across a range of values of the four measures. As in the experiment described in Section \ref{sec:validation}, we hold three of the measures fixed and vary a fourth, recording how the average principal welfare changes as a result. In the example in Section \ref{sec:validation} we set the other variables to 0.9, and inspected games with approximately $10^1$ outcomes. Here, we set the variables to (the same) values in 0, 0.25, 0.5, 0.75, and 1, respectively, and inspect games with $10^1$, $10^2$, and $10^3$ outcomes, in Figures \ref{fig:extra_plots_10}, \ref{fig:extra_plots_100}, and \ref{fig:extra_plots_1000}, respectively.

For each such setting, we varied the fourth variable linearly between 0 and 1 and generated 10 samples per increment, as follows: 
\begin{enumerate}
    \item Randomly generate a delegation game $D$ satisfying the given measures for individual and collective alignment (of the the principals). This is done by:
    \begin{itemize}
        \item Sampling an initial vector from a multivariate Gaussian $\mathcal{N}(\bm{0},\bm{1})$ for each principal $1 \leq i \leq n$, and normalising it to get $\uHat^i_\nu$;
        \item Sampling a magnitude $\hat{m}^i$ and constant $\hat{c}^i$ from uniform distributions $\U(0.5,1.5)$ and $\U(-1,1)$, respectively;
        \item Adjusting the angle to the vectors $\uHat^i$ so as to satisfy the collective alignment measure;
        \item Sampling a utility function for each agent, using the same process as before, and adjusting the angle of the resulting vector $u^i$ so as to satisfy the individual alignment measure.
    \end{itemize}
    \item Compute the pure $\NE$s and $\epsNE$s.\footnote{Unfortunately, at the time of writing, there are no scalable game solvers that can compute all mixed $\epsNE$s for arbitrary $\eps$.}
    \item Find the pure strategy profiles $\bm{s}$ such that: $$w_{\eps} + \CC \cdot (w_\star - w_{\bm{0}}) \leq w(\bm{s}) \leq w_{\eps} + (w_\star - w_{\bm{0}}).$$
    \item Compute the mean principal welfare across all such strategies, recording this along with $\wHat_+$, $\wHat_-$, $\wHat_\star$, and $\wHat_\bullet$.
    \item Compute and record the bounds from Theorem \ref{thm:capabilities-bound} and Proposition \ref{prop:ideal-welfare-bound}, which may vary slightly depending on each $m^i$ and $\hat{m}^i$.
\end{enumerate}

\subsection{Alternative Visualisations}
\label{sec:alt_vis}

Another empirical perspective on our theoretical results is given by 3-dimensional plots of the bounds from Section \ref{sec:bounds}. Given the number of separate terms in these bounds, a full-dimensional plot is impossible, but there are a few natural ways to combine the variables along three axes. These plots display the \textit{principals' welfare regret} $\hat{w}_\star - \hat{w}(\bm{\sigma})$ alongside \textit{agents' welfare regret} $w_\star - w(\bm{\sigma})$ and \textit{total agent misalignment}, measured as $\sum_i m(\uHat^i - u^i) = 2 \sum_i (1 - \IA^i)$. 
We abstract over the agents' individual and collective capabilities as separate inputs by considering only the agent welfare. We sample outcomes with full support as mixed (correlated) strategies producing a given agent welfare regret (so a given agent welfare regret does not here specify what capabilities or mechanism gave rise to it). This leaves the `welfare calibration ratios' $r^i$, which we hold fixed per visualisation, the collective alignment $\CA$ which we permit to vary in our generating procedure, and the \emph{ideal} and \emph{maximal} welfare regret, which are not represented in these visualisations.

Figure \ref{fig:3d-plots} shows four plots of this kind, each taking a different set of fixed $r^i$ ratios. Observe that when these ratios are uneven, the principals are subject to additional `welfare miscalibration' risk, even when agent welfare regret and misalignment are zero, as seen in our bound surface and as exemplified by the experimental samples. This occurs when a trade-off between agents' utility registers as a gain according to the agents' welfare, but the same trade-off between corresponding principals is a bad one.

\begin{figure}
    \centering
    \begin{subfigure}[t]{0.49\columnwidth}
        \includegraphics[width=\textwidth]{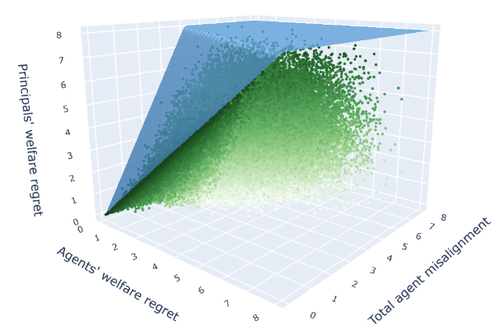}
        \caption{}
    \end{subfigure}
    \begin{subfigure}[t]{0.49\columnwidth}
        \includegraphics[width=\textwidth]{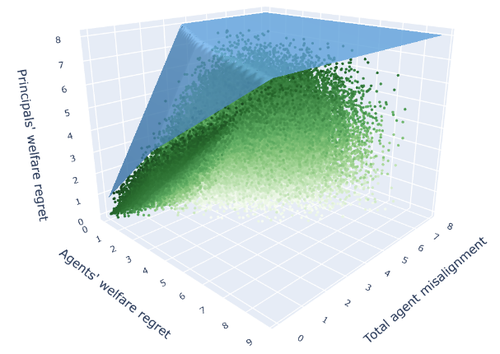}
        \caption{}
    \end{subfigure}
    \begin{subfigure}[t]{0.49\columnwidth}
        \includegraphics[width=\textwidth]{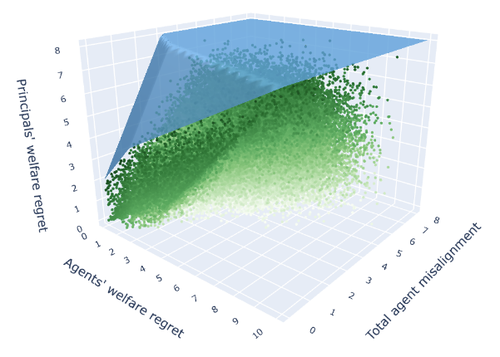}
        \caption{}
    \end{subfigure}
    \begin{subfigure}[t]{0.49\columnwidth}
        \includegraphics[width=\textwidth]{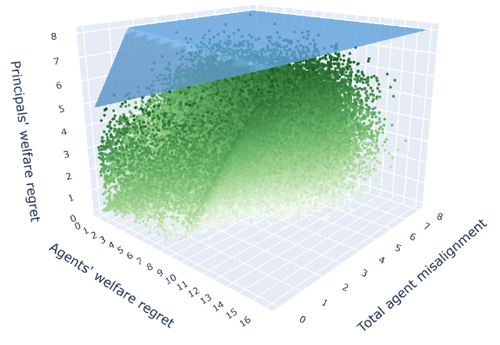}
        \caption{}
    \end{subfigure}
    \caption{$10^5$ simulated games, plotted with principals' welfare regret, agents' welfare regret, and total agent misalignment. Our bound surface in blue, with individual game outcomes as a green scatter. Collective alignment is permitted to vary, while welfare ratios $r^i$ are fixed per plot as: (a) four players, all equal $r^i = 1$; (b) four players, $r^1 = 1.4$ with $r^i = 1$ otherwise; four players, $r^1 = 2$ with $r^i = 1$ otherwise; four players, $r^1 = 5$ with $r^i = 1$ otherwise.}
    \label{fig:3d-plots}
\end{figure}

We also include a visualisation of the normalisation procedure described in Definition \ref{def:normalisation}, where recall that:
$$u^i_\nu \coloneqq \begin{cases}
    0 & \text{ if } m \big( u^i - c(u^i) \bm{1} \big) = 0\\
    \frac{u^i - c(u^i) \bm{1}}{m ( u^i - c(u^i) \bm{1} )} & \text{ otherwise.} 
\end{cases}$$
If we restrict $U$ to three dimensions, with concrete choices of $c$ and $m$, we can visualise the effect of standardisation on a collection of random utility functions, shown in Figure \ref{fig:standardisation}. Concretely, we use a uniformly-weighted expected value $c(u) = \frac 1 3 \sum u(\bm s)$ and $m = \ell^2$. Inset graphics illustrate an alternative choice of shift function, $c(u) = \frac{1}{2}\left(\max_{\bm{s}} u(\bm{s}) - \min_{\bm{s}}u(\bm{s})\right)$.
\begin{figure*}
    \centering
    \begin{subfigure}[t]{0.3\textwidth}
        \includegraphics[width=\textwidth]{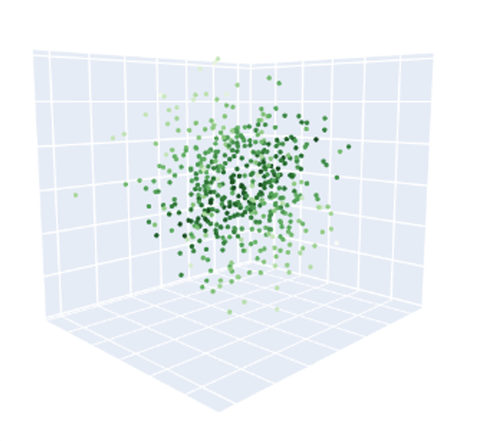}
        \caption{}
    \end{subfigure}
    \begin{subfigure}[t]{0.3\textwidth}
        \includegraphics[width=\textwidth]{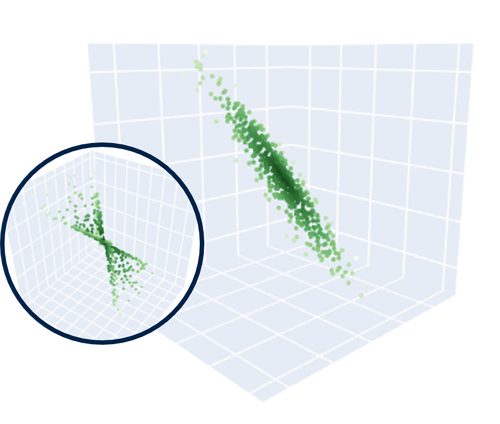}
        \caption{}
    \end{subfigure}
    \begin{subfigure}[t]{0.3\textwidth}
        \includegraphics[width=\textwidth]{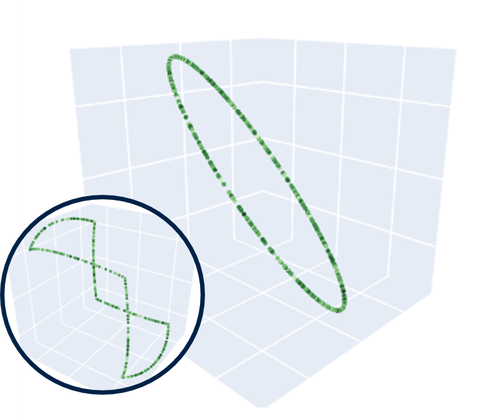}
        \caption{}
    \end{subfigure}
    \caption{(a) A random selection of utility functions $u$ (each represented as a point) fill the space, and are naively incomparable; (b) The affine-equivariant shift $u-c(u)$ projects onto a lower-dimensional surface (here, a disc; inset, a warped disc) while preserving preferences; and (c) normalising via $\frac {u-c(u)} {m(u-c(u))}$ preserves preferences while distinguishing all possible preferences (here, a circle; inset, a warped circle).}
    \label{fig:standardisation}
\end{figure*}

\section{Worked Example}

Here we show how to compute each of the four measures in the game from Example \ref{ex:driving}, which is shown as Figure \ref{fig:driving} in the main body and repeated here for reference.
\begin{figure}[H]
    \centering
    \begin{subfigure}[b]{0.36\columnwidth}
        \centering
        \renewcommand\arraystretch{1.7}
        \begin{tabular}{c|c|c|c}
            \multicolumn{1}{c}{} & \multicolumn{1}{c}{$A$} & \multicolumn{1}{c}{$B$} &\\\cline{2-3}
            $A$ & \textcolor{orange}{3}, \textcolor{blue}{3} & \textcolor{orange}{6}, \textcolor{blue}{2} &\phantom{$A$}\\\cline{2-3}
            $B$ & \textcolor{orange}{2}, \textcolor{blue}{6} & \textcolor{orange}{0}, \textcolor{blue}{0} &\phantom{$B$}\\\cline{2-3}
        \end{tabular}
        \caption{}
    \end{subfigure}    
    \begin{subfigure}[b]{0.36\columnwidth}
        \centering
        \renewcommand\arraystretch{1.7}
        \begin{tabular}{c|c|c|c}
            \multicolumn{1}{c}{} & \multicolumn{1}{c}{$A$} & \multicolumn{1}{c}{$B$} &\\\cline{2-3}
            $A$ & \textcolor{red}{2}, \textcolor{teal}{3} & \textcolor{red}{3}, \textcolor{teal}{3} &\phantom{$A$}\\\cline{2-3}
            $B$ & \textcolor{red}{4}, \textcolor{teal}{6} & \textcolor{red}{3}, \textcolor{teal}{0} &\phantom{$B$}\\\cline{2-3}
        \end{tabular}
        \caption{}
    \end{subfigure}
    \begin{subfigure}[b]{0.25\columnwidth}
        \centering
        \includegraphics[width=0.8\textwidth]{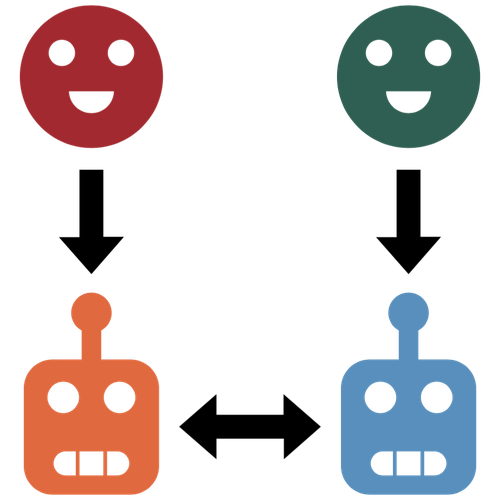}
        \caption{}
    \end{subfigure}
    \renewcommand\thefigure{\ref*{fig:driving}}
    \caption{(a) The payoffs of the agents in Example \ref{ex:driving}; (b) the payoffs of the principals; and (c) a graphical representation of the full delegation game, with vertical arrows indicating control and horizontal arrows indicating cooperation.}
\end{figure}

\subsection{Alignment}

First, how are the agents \emph{individually aligned} with their principals? We must standardise the utility functions in order to make comparisons. (See Figure \ref{fig:standardisation} for an representative visualisation of the procedure.)
For ease of illustration,\footnote{$\ell^\infty$ is \emph{not} a strictly convex norm, so while our upper bound results are satisfied for this choice, our alignment desiderata are not all satisfied. We use $\ell^\infty$ for the purposes of this worked example simply because it results in clearer, cleaner numbers.} we choose: 
$$c(u) = \frac{\max_{\bm{s}} u(\bm{s}) - \min_{\bm{s}}u(\bm{s})}{2} \quad \text{and} \quad m = \ell^\infty.$$
Each agent and principal in this example has utilities ranging from $0$ to $6$ or $2$ to $4$, meaning that $c(u) = 3$ in each case. Subtracting this quantity gives:
\begin{figure}[H]
    \centering
    \begin{subfigure}[b]{0.49\columnwidth}
        \centering
        \renewcommand\arraystretch{1.7}
        \begin{tabular}{c|c|c|c}
            \multicolumn{1}{c}{} & \multicolumn{1}{c}{$A$} & \multicolumn{1}{c}{$B$} &\\\cline{2-3}
            $A$ & \textcolor{orange}{0}, \textcolor{blue}{0} & \textcolor{orange}{3}, \textcolor{blue}{-1} &\phantom{$A$}\\\cline{2-3}
            $B$ & \textcolor{orange}{-1}, \textcolor{blue}{3} & \textcolor{orange}{-3}, \textcolor{blue}{-3} &\phantom{$B$}\\\cline{2-3}
        \end{tabular}
    \end{subfigure}
    \begin{subfigure}[b]{0.49\columnwidth}
        \centering
        \renewcommand\arraystretch{1.7}
        \begin{tabular}{c|c|c|c}
            \multicolumn{1}{c}{} & \multicolumn{1}{c}{$A$} & \multicolumn{1}{c}{$B$} &\\\cline{2-3}
            $A$ & \textcolor{red}{-1}, \textcolor{teal}{0} & \textcolor{red}{0}, \textcolor{teal}{0} &\phantom{$A$}\\\cline{2-3}
            $B$ & \textcolor{red}{1}, \textcolor{teal}{3} & \textcolor{red}{0}, \textcolor{teal}{-3} &\phantom{$B$}\\\cline{2-3}
        \end{tabular}
    \end{subfigure}
\end{figure}
\addtocounter{figure}{-1}
The centred agents' utilities also coincidentally have the same norm, namely $m^1 = m^2 = 3$, while the \textcolor{red}{first principal} has norm $\hat{m}^1 = 1$ and the \textcolor{teal}{second principal} has norm $\hat{m}^2 = 3$.\footnote{We might interpret this normatively as the \textcolor{red}{first principal}'s preferences over these outcomes mattering less, or as their marginal relative utility being `cheap' in the setting with transferable utility.} Dividing by these norms gives:
\begin{figure}[H]
    \centering
    \begin{subfigure}[b]{0.49\columnwidth}
        \centering
        \renewcommand\arraystretch{1.7}
        \begin{tabular}{c|c|c|c}
            \multicolumn{1}{c}{} & \multicolumn{1}{c}{$A$} & \multicolumn{1}{c}{$B$} &\\\cline{2-3}
            $A$ & \textcolor{orange}{0}, \textcolor{blue}{0} & \textcolor{orange}{1}, \textcolor{blue}{$-\frac 1 3$} &\phantom{$A$}\\\cline{2-3}
            $B$ & \textcolor{orange}{$-\frac 1 3$}, \textcolor{blue}{1} & \textcolor{orange}{-1}, \textcolor{blue}{-1} &\phantom{$B$}\\\cline{2-3}
        \end{tabular}
    \end{subfigure}
    \begin{subfigure}[b]{0.49\columnwidth}
        \centering
        \renewcommand\arraystretch{1.7}
        \begin{tabular}{c|c|c|c}
            \multicolumn{1}{c}{} & \multicolumn{1}{c}{$A$} & \multicolumn{1}{c}{$B$} &\\\cline{2-3}
            $A$ & \textcolor{red}{-1}, \textcolor{teal}{0} & \textcolor{red}{0}, \textcolor{teal}{0} &\phantom{$A$}\\\cline{2-3}
            $B$ & \textcolor{red}{1}, \textcolor{teal}{1} & \textcolor{red}{0}, \textcolor{teal}{-1} &\phantom{$B$}\\\cline{2-3}
        \end{tabular}
    \end{subfigure}
\end{figure}
\addtocounter{figure}{-1}
Finally, we can compare the standardised utilities of principals and agents, to find that $m(\hat u_\nu^1 - u_\nu^1) = \frac 4 3$ and $m(\hat u_\nu^2 - u_\nu^2) = \frac 1 3$. Thus, $\IA^1 = \frac 1 3$ and $\IA^2 = \frac 5 6$. This matches our intuitions, as the \textcolor{orange}{first agent}'s preferences visibly deviate from the \textcolor{red}{first principal}'s (inverting several pure preferences), while the \textcolor{blue}{second agent}'s are a very close match for the \textcolor{teal}{second principal}'s.

Next we consider the \emph{collective alignment} of the agents (the calculation for the principals proceeds in the same way). We have already normalised the utilities. Because the magnitudes $m^1$ and $m^2$ are equal, the welfare-representative proxy utility $\mu^w$ is simply the mean of the agents' standardised utilities:
\begin{figure}[H]
    \centering
    \renewcommand\arraystretch{1.7}
    \begin{tabular}{c|c|c|c}
        \multicolumn{1}{c}{} & \multicolumn{1}{c}{$A$} & \multicolumn{1}{c}{$B$} &\\\cline{2-3}
        $A$ & 0 & $\frac 1 3$ &\phantom{$A$}\\\cline{2-3}
        $B$ & $\frac 1 3$ & -1 &\phantom{$B$}\\\cline{2-3}
    \end{tabular}
\end{figure}
\addtocounter{figure}{-1}
Using $\mu^w$ we see that both $m(\mu^w - u_\nu^1) = \frac 2 3$ and $m(\mu^w - u_\nu^2) = \frac 2 3$. 
Again, because the magnitudes are equal, the weighted average of these distances is also $\frac 2 3$, and so $\CA = \frac 1 3$. 
This again matches our intuition, as none of the Pareto-optimal outcomes $(\textcolor{orange}{0}, \textcolor{blue}{0})$, $(\textcolor{orange}{3}, \textcolor{blue}{-1})$, and $(\textcolor{orange}{-1}, \textcolor{blue}{3})$ comes close to the (impossible) \emph{ideal welfare} outcome $(\textcolor{orange}{3}, \textcolor{blue}{3})$, but the situation is also clearly not entirely adversarial.

Using the quantities above we can also immediately consider the \emph{welfare ratios}. As $\hat{m}^1 = 1$, $\hat{m}^2 = 3$, and $m^1 = m^2 = 3$ then we have $r^1 = 1$ and $r^2 = \frac 1 3$. Since these are unequal, even had the agents been fully aligned, their welfare optimum would not necessarily coincide with that of the principals.

\subsection{Capabilities}

Since the capabilities are a feature of how the game is played (rather than the game itself), they are not determined by the game described in Example \ref{ex:driving}.

We note that the agent welfare regret is a function of capabilities, and we can describe both the principal and the agent welfare regret of each outcome by a simple comparison to the \emph{maximal welfare} ($8$ for agents and $10$ for principals):

\begin{figure}[H]
    \centering
    \begin{subfigure}[b]{0.49\columnwidth}
        \centering
        \renewcommand\arraystretch{1.7}
        \begin{tabular}{c|c|c|c}
            \multicolumn{1}{c}{} & \multicolumn{1}{c}{$A$} & \multicolumn{1}{c}{$B$} &\\\cline{2-3}
            $A$ & 2 & 0 &\phantom{$A$}\\\cline{2-3}
            $B$ & 0 & 8 &\phantom{$B$}\\\cline{2-3}
        \end{tabular}
    \end{subfigure}
    \begin{subfigure}[b]{0.49\columnwidth}
        \centering
        \renewcommand\arraystretch{1.7}
        \begin{tabular}{c|c|c|c}
            \multicolumn{1}{c}{} & \multicolumn{1}{c}{$A$} & \multicolumn{1}{c}{$B$} &\\\cline{2-3}
            $A$ & 5 & 4 &\phantom{$A$}\\\cline{2-3}
            $B$ & 0 & 7 &\phantom{$B$}\\\cline{2-3}
        \end{tabular}
    \end{subfigure}
\end{figure}
\addtocounter{figure}{-1}

With observations of actual play by the agents we could estimate their capabilities (as in Section \ref{sec:inference}). Alternatively, if we know the agents' capabilities, we can describe or draw from the distribution of possible outcomes.

Let us denote by $p^i$ the probability of agent $i$ playing strategy $A$. The agent game $G$ has a single dominant strategy equilibrium at $(A, A)$, i.e. $p^1 = p^2 = 1$, which is the sole NE. As a result, we have $w_{\bm{0}} = 6$, which gives agent welfare regret of $2$ (and ideal regret of $6$) -- the result obtained by agents with $\IC^i = 1$ and $\CC = 0$. In general, the agent welfare is given by $8p^1 + 8p^2 - 10p^1p^2$, meaning that if either $\epsilon^1 \le \frac 1 5$ or $\epsilon^2 \le \frac 1 5$ then the pure $(A, A)$ strategy is also the lowest-welfare $\epsNE$.

As a concrete example, let us suppose that $\IC^1 = 0.9$ (so $\epsilon^1 = 0.1$), $\IC^2 = 0.7$ (so $\epsilon^2 = 0.3$), and $\CC = 0.5$. Then the lowest-welfare $\epsNE$ is the pure $(A, A)$ strategy as above, and so $w_{\eps} = 6$ (from Definition \ref{def:CC}). But $w_{\bm 0} = 6$ as well, so with $\CC = 0.5$, the agents achieve welfare at least $6 + 0.5 \cdot (8 - 6) = 7$. I.e., the agent welfare regret is 1. Solving for $8p^1 + 8p^2 - 10p^1p^2 \geq 7$ leads to the set of strategies played in $G$, which are characterised by $0 \leq p^1 \leq \frac{1}{2}$ or $\frac{7}{8} \leq p^1 \leq 1$, with $p^2 = \frac{8p^1 - 7}{10p^1 - 4}$.

Finally, we can assess what this means for principal welfare, which is given by $3 + 3p^1 + 7p^2 - 8p^1p^2$. Computing the maxima and minima of this function in the two regions of agent strategy space defined by the inequalities above leads to $\wHat(\bm{\sigma}) \in [\frac{15}{2}, \frac{73}{8}] = [7.5, 9.125]$ and $\wHat(\bm{\sigma}) \in [\frac{11}{2}, \frac{3}{25}(53 - \sqrt{34})] \approx [5.5, 5.660]$. As $\wHat_\star = 10$ this leads to a principal welfare regret between 0.875 and 4.5.

\section{Further Discussion}

We conclude with some additional discussion of a more philosophical nature, while drawing connections to the mathematics underlying our measures and related results in the wider literature.

\subsection{Variations on Alignment Metrics}

The alignment measure we introduce has the attractive property that a distance $m(u_\nu - u'_\nu)$ of zero corresponds exactly to identical preference orderings over mixed outcomes, as demonstrated in the proof of Proposition \ref{prop:D2}.
This property relies on the positive definiteness property of the metric, which is guaranteed by any norm, including the weighted $\ell^2$ norm used by EPIC \cite{Gleave2021}, provided the coverage distribution has support on all possible outcomes. 
Without this, the distance can be zero even for differing preference orderings (though preferences over outcomes with support in the coverage distribution will still be in agreement).

Beyond providing this guarantee by having support on all possible outcomes, we might wish to guide this notion of distance to give more weight to more `important' or `plausible' outcomes. Indeed, if using weighted $\ell^2$ with a guide distribution $\Delta$ which is precisely the mixed strategy profile $\bm \sigma$ selected by the players, the distance $m(u_\nu - u'_\nu)$ has a close relationship with the correlation of players' payoffs under that strategy profile:
$$
m(u_\nu - u'_\nu) = \sqrt{2 - 2\rho_{\bm \sigma}(u, u')}
$$
where $\rho_{\bm \sigma}$ is the Pearson correlation under the distribution arising from $\bm \sigma$.
This is appealing in some sense, but it requires invoking a specific strategy profile $\bm \sigma$ which is in general undetermined \emph{a priori}, and it also fails to capture the notion that players might agree or disagree on preferences over alternative, unrealised outcomes.
If we expect a particular solution concept to be played, we might limit the coverage distribution to support only those strategies that arise under that concept, though the weighting is still underdetermined.
\begin{figure*}[h]
    \centering
    \begin{subfigure}[t]{0.32\textwidth}
        \includegraphics[width=\textwidth]{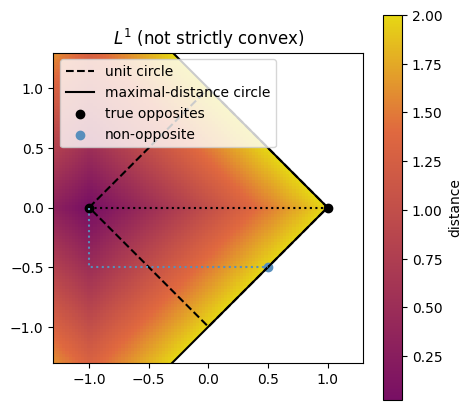}
        \caption{}
    \end{subfigure}
    \begin{subfigure}[t]{0.32\textwidth}
        \includegraphics[width=\textwidth]{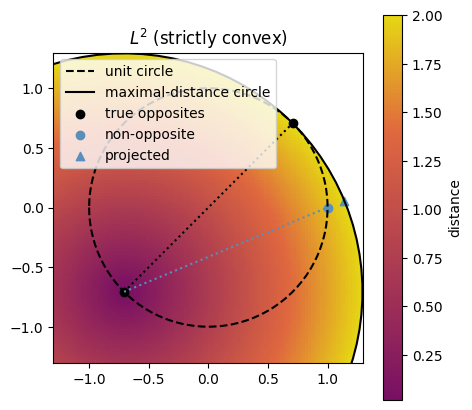}
        \caption{}
    \end{subfigure}
    \begin{subfigure}[t]{0.32\textwidth}
        \includegraphics[width=\textwidth]{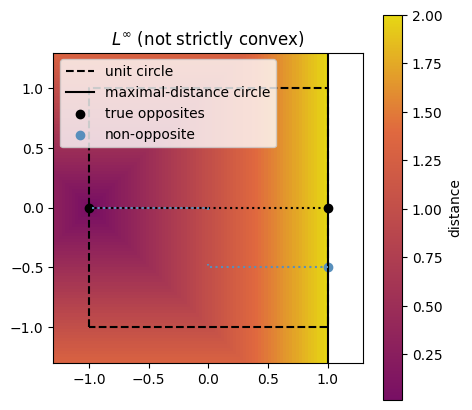}
        \caption{}
    \end{subfigure}
    \caption{(a) $\ell^1$ norm, a \emph{non}-strictly convex norm; (b) $\ell^2$ norm, a strictly convex norm; and (c) $\ell^\infty$ norm, a \emph{non}-strictly convex norm -- distance is represented as an interrupted line so that the length corresponds to the norm distance.}
    \label{fig:strictnorms}
\end{figure*}

\subsection{Relaxing Our Desiderata}
\label{appendix:relaxations}

As hinted in our discussion of desiderata D1-D5, several relaxations are possible, which open up more permissive definitions of the measures we use to satisfy them. For example, recall:
\begin{itemize}
    \item[(D2)] Two players are perfectly individually aligned (misaligned) if and only if they have identical (opposite) preferences. Two or more players are perfectly \emph{collectively} aligned  if and only if they have identical preferences.
\end{itemize}
In Definition \ref{def:normalisation}, we required $c$ to be affine-equivariant, and $m$, used for both normalisation and distance, to be a strictly convex norm.

Weakening any of these requirements means that at least one of our desiderata is not satisfied,\footnote{This is not to say that these requirements are \emph{necessary} for the desiderata, merely that each is necessary given the others. For example if we ask for the normalisation and the distance to each be a norm, then they must be the same norm and strictly convex, but the desiderata can be satisfied by some functional forms not including norms at all.} but some relaxations are nonetheless of interest for different purposes.
It is straightforward to see why positivity and absolute homogeneity are necessary properties of the normalisation, to satisfy any of our desiderata, and why the centering function used to normalise must be equivariant to affine transformations; without any of these properties, equivalent preferences expressed by scaled or shifted utility functions would map to different normalised points. It may be less clear that positive definiteness, the triangle inequality, or strict convexity is required of the distance, or why the distance metric used for calculating (mis)alignment should be equivalent to the one used for normalisation.

\subsubsection{The Enemy of my Enemy is my Friend: Opposite Preferences and Strict Convexity}

Our desiderata ask that opposite preferences correspond to maximal misalignment and vice versa, which makes `most misaligned' an involution -- i.e. `the enemy of my enemy is my friend', interpreted as a statement about preferences (aside from any strategic considerations or coalition formation), holds true.
Without strict convexity or matching norms for normalisation and distance, this part of our alignment desideratum is not in general satisfied.

If we use a general norm $m$ instead of a strictly convex norm, opposite preferences are still guaranteed to correspond to maximal misalignment (the same proof suffices), but there may be \emph{other} utility functions beside the opposite preferences which nevertheless measure as maximally misaligned. One such non-strictly convex norm is the max or $\ell^\infty$ norm, which corresponds to scaling normalised utility functions to a fixed range. Consider utility functions over three values, represented as a vector $u = (u_0, u_1, u_2)$, and let us normalise to the range $[-1, 1]$ using the $\ell^\infty$ norm. Now $(-1, 0, 1)$ has opposite preferences to $(1, 0, -1)$, and indeed their distance is $\max\{2, 0, 2\} = 2$, which is maximal on the $\ell^\infty$-normalised manifold. But the alternative utility $(1, 1, -1)$ also measures as maximally distant, despite not exhibiting entirely opposite preferences! 

Figure \ref{fig:strictnorms} exhibits the difference between a strictly convex norm and some non-strictly convex norms. The dashed contour is the unit circle of the respective norm. The leftmost identified point is the `start' point in each image, and the colour-mapped distance and solid `maximal-distance circle' relates to that point, measured using the respective norm. For $\ell^1$ and $\ell^\infty$, \emph{non}-strictly convex norms, the blue point, non-opposite, measures as equally distant from the start point as the `true' opposite - visually, this corresponds to the `maximal distance circle' coinciding with the `unit circle' at multiple points. In contrast, for $\ell^2$, a strictly-convex norm, the blue point, non-opposite, is strictly closer than the true opposite. A blue triangle marks where on the `maximal-distance circle' the blue point would have to be to match the maximal distance: it lies off the standard unit circle and is thus not the normalisation of any utility function.

\begin{figure}[h]
    \centering
    \begin{subfigure}[t]{0.33\textwidth}
        \includegraphics[width=\textwidth]{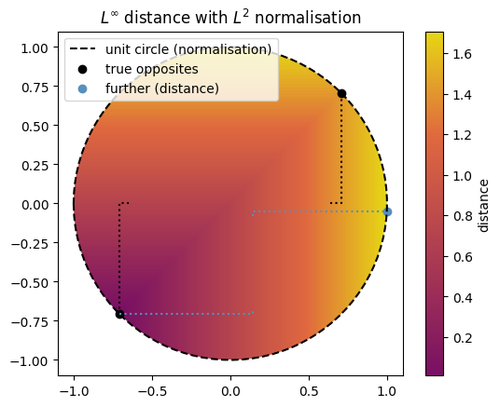}
        \caption{}
    \end{subfigure}

    \begin{subfigure}[t]{0.33\textwidth}
        \includegraphics[width=\textwidth]{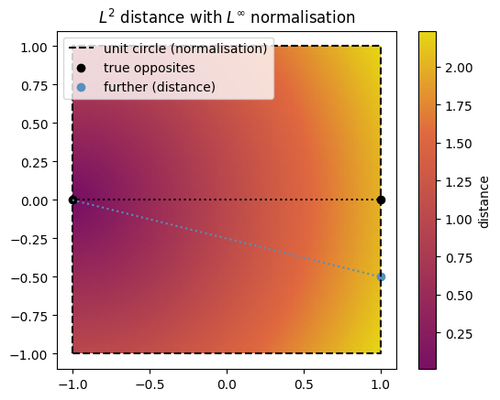}
        \caption{}
    \end{subfigure}
    \caption{(a) $\ell^\infty$ distance with $\ell^2$ normalisation -- distance is represented as an interrupted line so that the length corresponds to the norm distance; and (b) $\ell^2$ distance with $\ell^\infty$ normalisation.}
    \label{fig:matchingnorms}
\end{figure}

If we relax things further, by permitting the distance measure between two normalised utilities to differ from that used to normalise,\footnote{In fact, the distance may be a general metric, and the normalisation need not be sub-additive, i.e. it need not satisfy a triangle inequality.} we still preserve the relationship between identical preferences and zero misalignment distance, but there will in general exist non-opposite utility functions that measure as \emph{even more} distant than their opposites. Consider first a normalisation with $\ell^\infty$ but a distance using $\ell^2$. Now take $u = (-1, 0, 1)$ and notice that its opposite, $(1, 0, -1)$ is at $\ell^2$ distance $2\sqrt 2$, but $(1, 1, -1)$ is `further' away in $\ell^2$ (distance $3$), despite having less inversion of preference! Alternatively consider a normalisation with $\ell^2$, and distance using $\ell^\infty$. Notice that similarly, $\big(\frac{\sqrt 2}{2}, \frac{\sqrt 2}{2}, 0\big)$ has $\ell^\infty$ distance $\sqrt 2$ from its opposite, $\big(\frac{\sqrt 2}{2}, \frac{\sqrt 2}{2}, 0\big)$, while, for example, $(-1, 0, 0)$ is more distant in $\ell^\infty$. 

Figure \ref{fig:matchingnorms} exhibits two combinations of inconsistent norms for normalisation and for measuring the distance between normalised points. In each case, the leftmost black point is the `starting point', and its true opposite is the other black point. The blue point, non-opposite, is actually counted as `more misaligned' under such combinations of inconsistent norms!
Nevertheless, both of these relaxations are able to preserve the other alignment desiderata, namely that identical preferences correspond to $0$ distance and vice versa. Although our overall desiderata regarding opposite preferences require these properties, our regret bounds are still fulfilled by measures with somewhat weaker properties, analogous to EPIC \cite{Gleave2021}, or the EPIC-like distances of \cite{Skalse2023}.

\subsection{Agent Welfare}
\label{appendix:pareto}

Here we discuss a close connection between Pareto concepts of optima and improvements, and welfare-specific concepts of optima and improvements. This connection unlocks an alternative perspective on our characterisation of cooperation, and provides a more agnostic foundation than `agent welfare' while essentially retaining the same analysis.

In the case of the principals, whose outcomes we are most interested in analysing, it should not be surprising that we reach for an aggregate utility, i.e. a social welfare function, with average utilitarian welfare in particular receiving most of our attention in this paper. When we imagine principals being humans, this welfare bears some correspondence to the colloquial sense of `welfare', and even if we do not have precise access to `utility' nor to the `right' way to aggregate it, this operationalisation of welfare is a good starting place for analysis, though the finer points remain debateworthy. 

For agents, on the other hand, especially when those agents are artificial systems, the referents of `utility' and `welfare' in the colloquial and the philosophical sense may be absent entirely, or at the very least inaccessible and unknown. We rescue a technical analysis of `agent utility', up to affine scaling, by appealing to well-known results on the preferences embodied by the agents \cite{MasColell1995,Tewolde2021,vonNeumann1944}.
In the case when agents are also moral patients (for example, other humans), we may wish to analyse both utility and welfare on the basis of some understanding of true wellbeing. But when wellbeing is meaningless or unknown, and our utility functions thus represent embodied preferences with no guarantee of commensurability, it appears difficult at first to sensibly produce meaningful aggregates. Must we taboo `agent welfare'?

Recall that our purposes for agent welfare in particular are to capture aspects of cooperation: collective alignment and capabilities. We consider cooperative success to consist of mutual gains (perhaps compared with some non-cooperative baseline), i.e. Pareto improvements. Pareto improvements and Pareto optima are a much more agnostic embodiment of cooperation, but as we have already seen in Proposition \ref{prop:D3}, for positively-weighted welfares, welfare optima \emph{are} Pareto optima. Various theorists have drawn converse results too: every \emph{weak} Pareto optimum is the optimum of some (non-negative) welfare \cite{Yaari1981}. \cite{Arrow1953} and descendant works -- e.g. \cite{Daniilidis2000} and \cite{Che2020} -- find other close connections between Pareto optima and utilitarian welfare optima. Our initial alignment result (Proposition \ref{prop:D3}) thus characterises the Pareto frontiers of two fully-aligned disjoint sets of players as being identical -- a statement that needs make no mention of welfare. What the welfare weightings provide is an `orientation' (either descriptive or prescriptive) in the space of utility trade-offs, which identifies particular optima from the Pareto frontier.

Bringing these points together, what can we do when we have a relatively clear way to analyse or estimate the welfare of principals, and we want to make richer claims than mere Pareto concepts for them, but not for agents? We might attempt to ensure that agents' utilities are scaled equivalently to their principals', which is to force all `welfare calibration' ratios to $r^i = 1$, but we need not be this prescriptive. The aforementioned theorems on welfare and Pareto concepts mean that for cooperative agent behaviours or dispositions (construed as weight-agnostic mutual gains, i.e. Pareto improvements) there is some reasonable description in terms of welfare weighting with which to quantify it.

We leave unspecified how specific trade-offs might enter into agents' collective behaviour or dispositions, but provide some examples. Consider, for instance, explicit instructions from their principals (whether or not they are fair), or exchange rates between different goods, broadly construed, or marginal production returns from given resources, which can give rise to natural orientations in the space of utility trade-offs (again, whether or not they are fair). Our discussion in terms of ratios and the welfare regret bounds then allow us to bring in concepts of welfare as desired in order to make more specific statements like those we prove in the main text.

\begin{figure*}[h]
    \centering
        \begin{subfigure}{0.24\textwidth}
                \centering
                \includegraphics[width=0.95\textwidth]{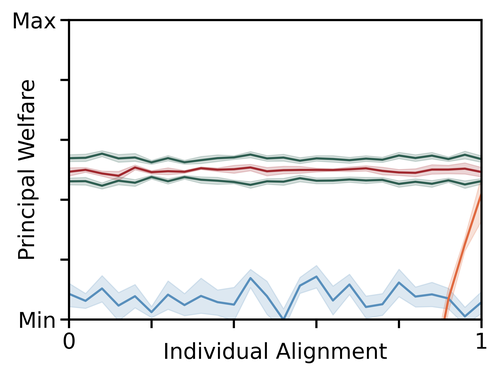}
        \end{subfigure}
        \begin{subfigure}{0.24\textwidth}
                \centering
                \includegraphics[width=0.95\textwidth]{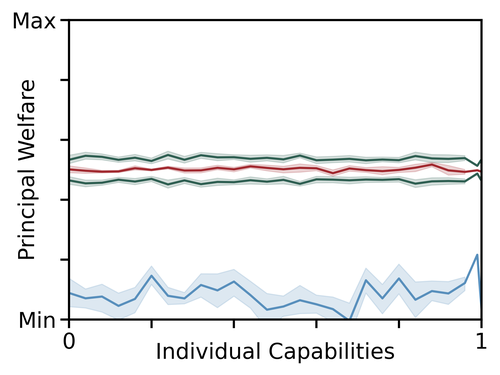}
        \end{subfigure}
        \begin{subfigure}{0.24\textwidth}
                \centering
                \includegraphics[width=0.95\textwidth]{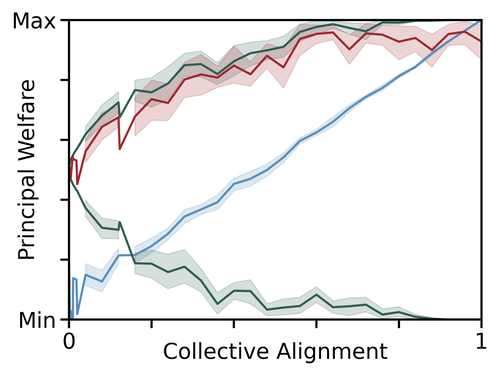}
        \end{subfigure}
        \begin{subfigure}{0.24\textwidth}
                \centering
                \includegraphics[width=0.95\textwidth]{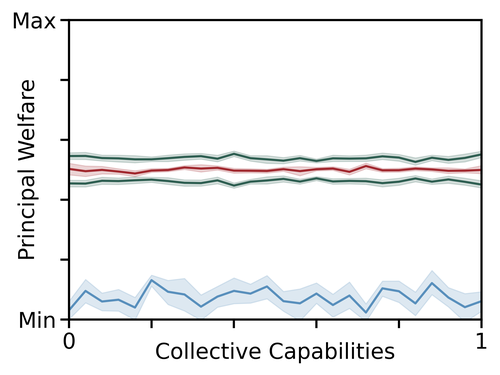}
        \end{subfigure}

        \begin{subfigure}{0.24\textwidth}
                \centering
                \includegraphics[width=0.95\textwidth]{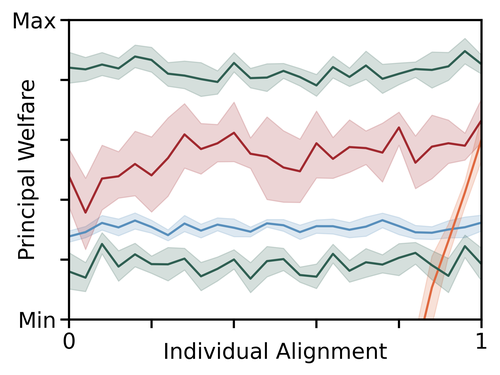}
        \end{subfigure}
        \begin{subfigure}{0.24\textwidth}
                \centering
                \includegraphics[width=0.95\textwidth]{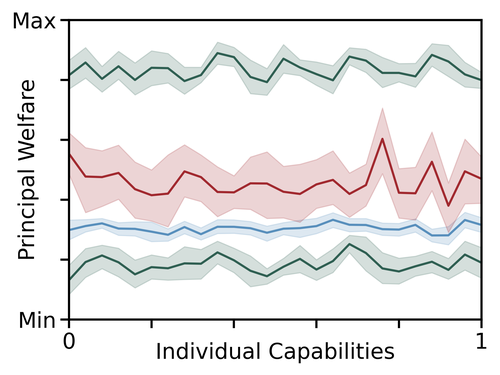}
        \end{subfigure}
        \begin{subfigure}{0.24\textwidth}
                \centering
                \includegraphics[width=0.95\textwidth]{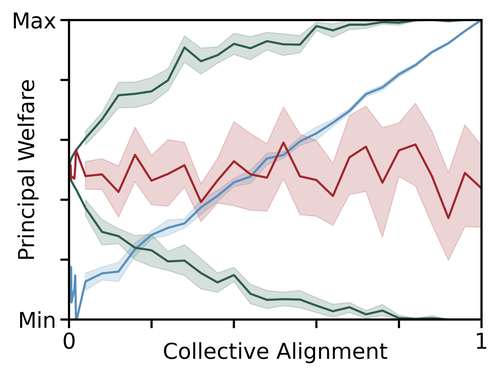}
        \end{subfigure}
        \begin{subfigure}{0.24\textwidth}
                \centering
                \includegraphics[width=0.95\textwidth]{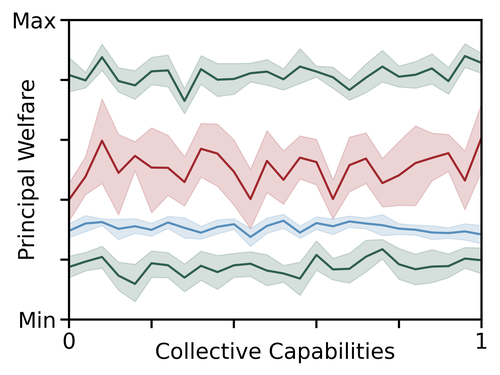}
        \end{subfigure}

        \begin{subfigure}{0.24\textwidth}
                \centering
                \includegraphics[width=0.95\textwidth]{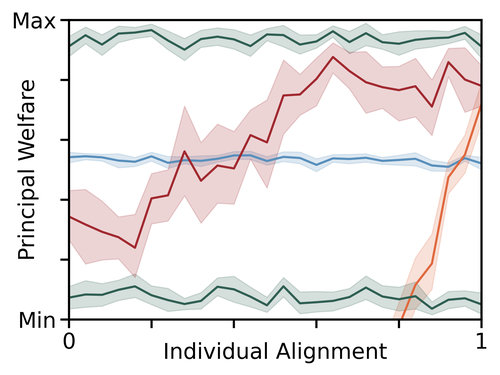}
        \end{subfigure}
        \begin{subfigure}{0.24\textwidth}
                \centering
                \includegraphics[width=0.95\textwidth]{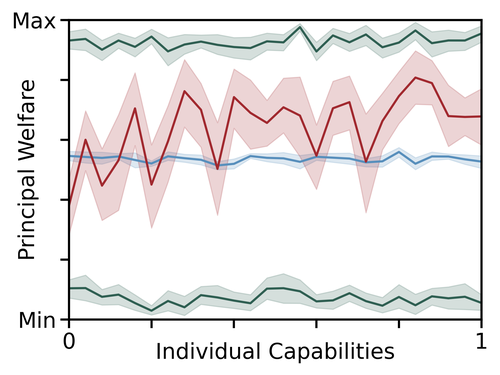}
        \end{subfigure}
        \begin{subfigure}{0.24\textwidth}
                \centering
                \includegraphics[width=0.95\textwidth]{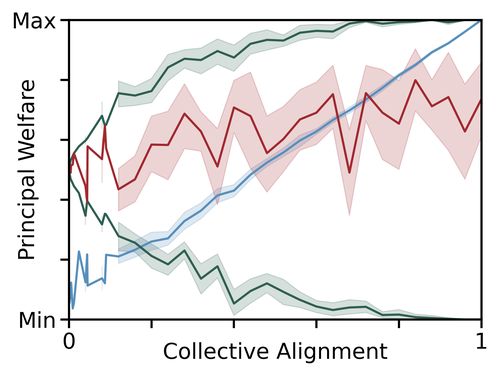}
        \end{subfigure}
        \begin{subfigure}{0.24\textwidth}
                \centering
                \includegraphics[width=0.95\textwidth]{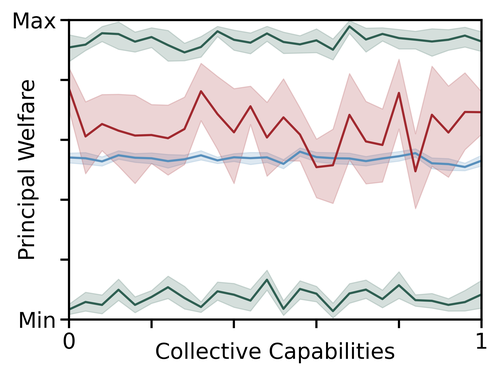}
        \end{subfigure}

        \begin{subfigure}{0.24\textwidth}
                \centering
                \includegraphics[width=0.95\textwidth]{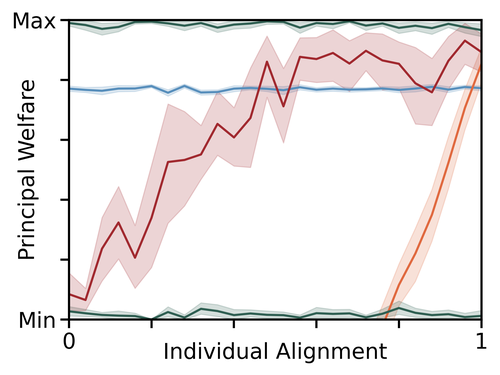}
        \end{subfigure}
        \begin{subfigure}{0.24\textwidth}
                \centering
                \includegraphics[width=0.95\textwidth]{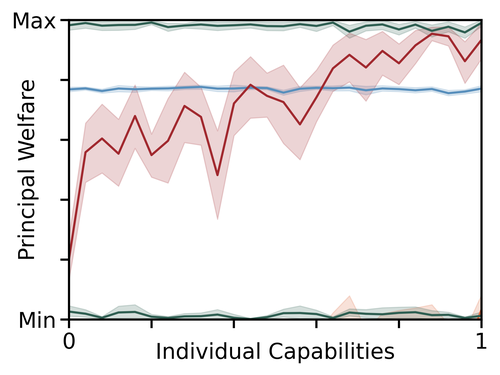}
        \end{subfigure}
        \begin{subfigure}{0.24\textwidth}
                \centering
                \includegraphics[width=0.95\textwidth]{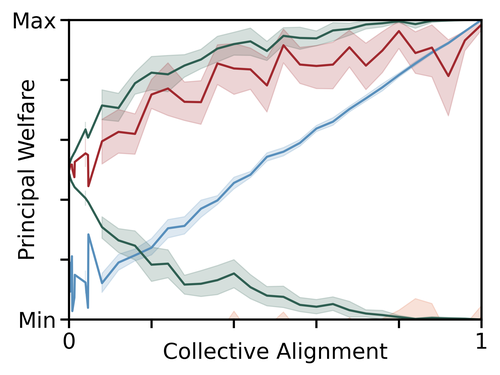}
        \end{subfigure}
        \begin{subfigure}{0.24\textwidth}
                \centering
                \includegraphics[width=0.95\textwidth]{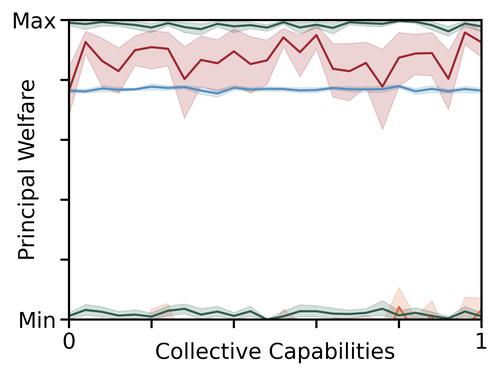}
        \end{subfigure}

        \begin{subfigure}{0.24\textwidth}
                \centering
                \includegraphics[width=0.95\textwidth]{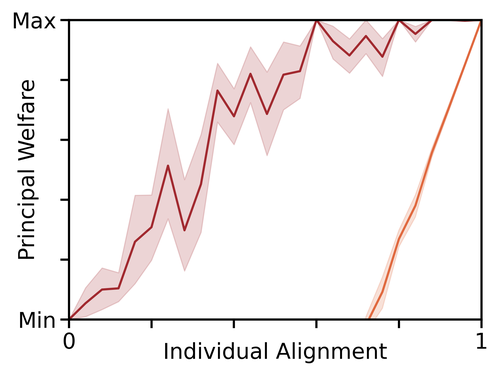}
        \end{subfigure}
        \begin{subfigure}{0.24\textwidth}
                \centering
                \includegraphics[width=0.95\textwidth]{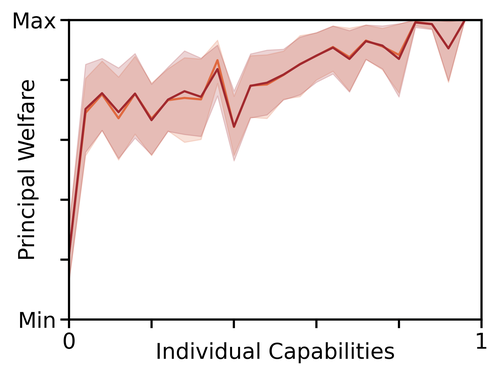}
        \end{subfigure}
        \begin{subfigure}{0.24\textwidth}
                \centering
                \includegraphics[width=0.95\textwidth]{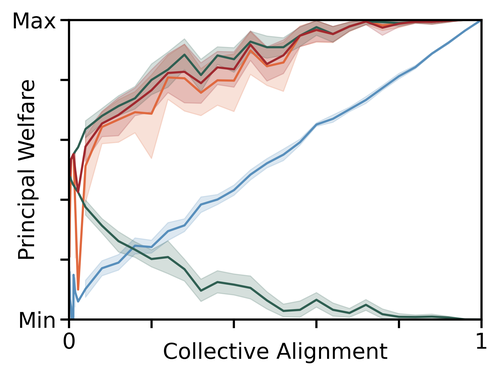}
        \end{subfigure}
        \begin{subfigure}{0.24\textwidth}
                \centering
                \includegraphics[width=0.95\textwidth]{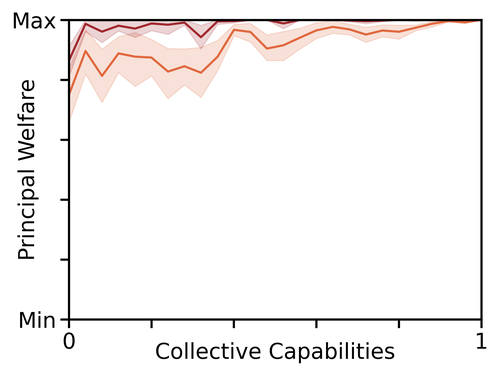}
        \end{subfigure}
    \caption{Each of these plots was generated from games with approximately $10$ outcomes. The key for this figure is identical to the related figure in the main body. Each row corresponds to setting (all of) the fixed variables to one of the five values $v \in \{0,0.25,0.5,0.75,1 \}$, respectively.}
    \label{fig:extra_plots_10}
\end{figure*}

\begin{figure*}[h]
    \centering
    \begin{subfigure}{0.24\textwidth}
            \centering
            \includegraphics[width=0.95\textwidth]{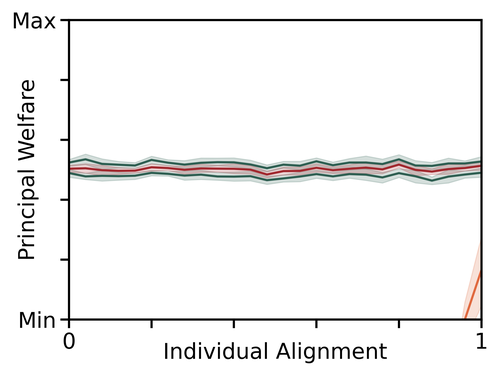}
    \end{subfigure}
    \begin{subfigure}{0.24\textwidth}
            \centering
            \includegraphics[width=0.95\textwidth]{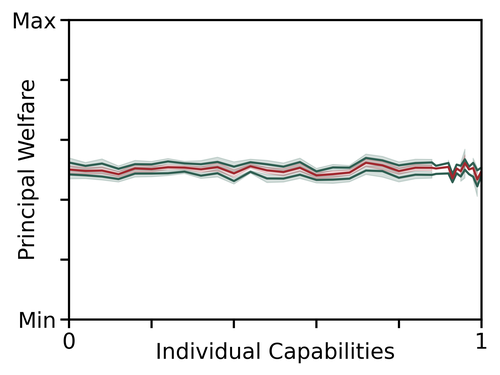}
    \end{subfigure}
    \begin{subfigure}{0.24\textwidth}
            \centering
            \includegraphics[width=0.95\textwidth]{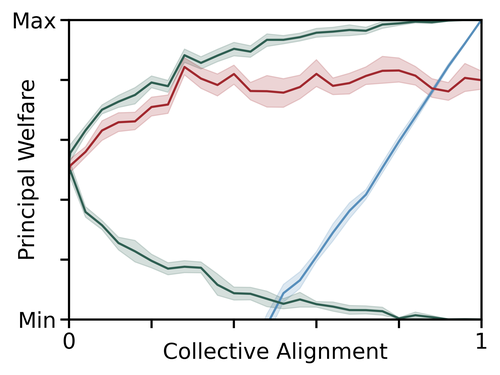}
    \end{subfigure}
    \begin{subfigure}{0.24\textwidth}
            \centering
            \includegraphics[width=0.95\textwidth]{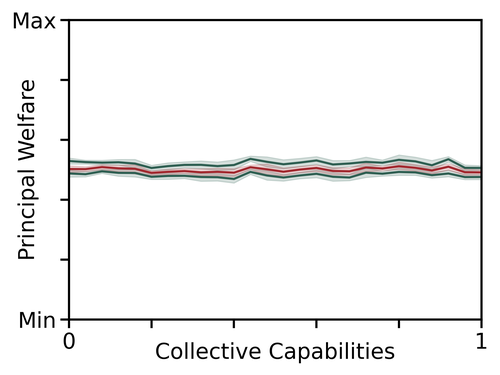}
    \end{subfigure}

    \begin{subfigure}{0.24\textwidth}
            \centering
            \includegraphics[width=0.95\textwidth]{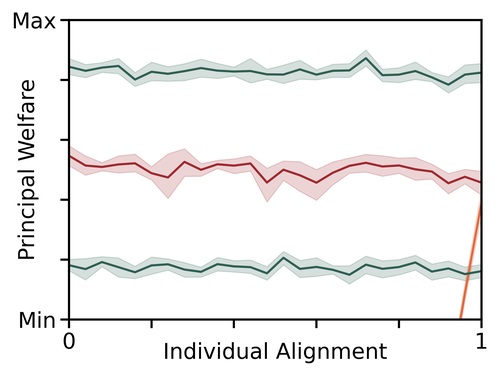}
    \end{subfigure}
    \begin{subfigure}{0.24\textwidth}
            \centering
            \includegraphics[width=0.95\textwidth]{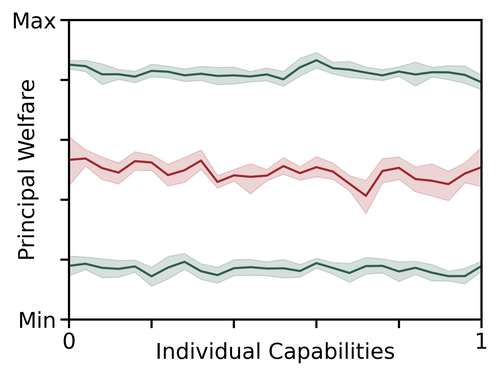}
    \end{subfigure}
    \begin{subfigure}{0.24\textwidth}
            \centering
            \includegraphics[width=0.95\textwidth]{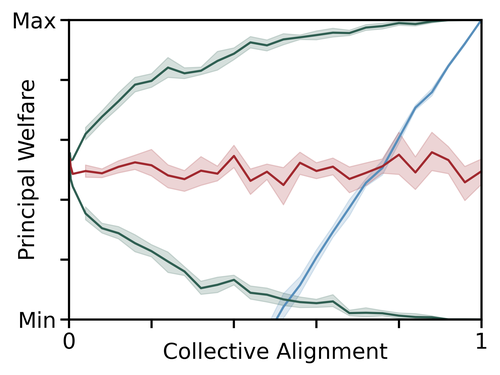}
    \end{subfigure}
    \begin{subfigure}{0.24\textwidth}
            \centering
            \includegraphics[width=0.95\textwidth]{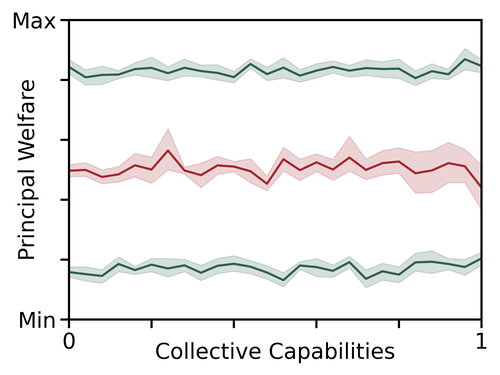}
    \end{subfigure}

    \begin{subfigure}{0.24\textwidth}
            \centering
            \includegraphics[width=0.95\textwidth]{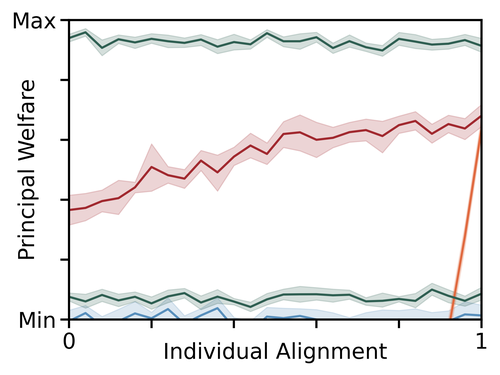}
    \end{subfigure}
    \begin{subfigure}{0.24\textwidth}
            \centering
            \includegraphics[width=0.95\textwidth]{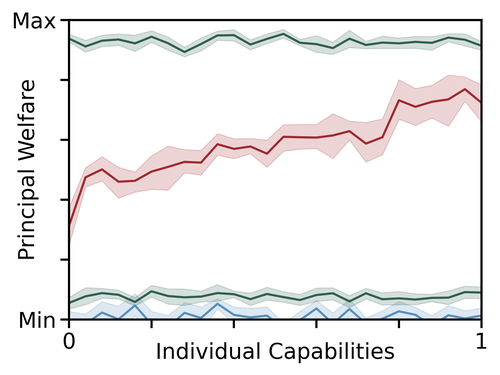}
    \end{subfigure}
    \begin{subfigure}{0.24\textwidth}
            \centering
            \includegraphics[width=0.95\textwidth]{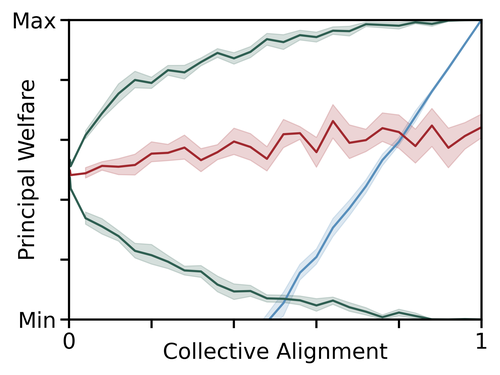}
    \end{subfigure}
    \begin{subfigure}{0.24\textwidth}
            \centering
            \includegraphics[width=0.95\textwidth]{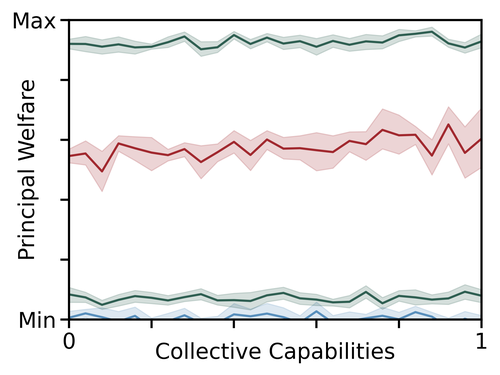}
    \end{subfigure}

    \begin{subfigure}{0.24\textwidth}
            \centering
            \includegraphics[width=0.95\textwidth]{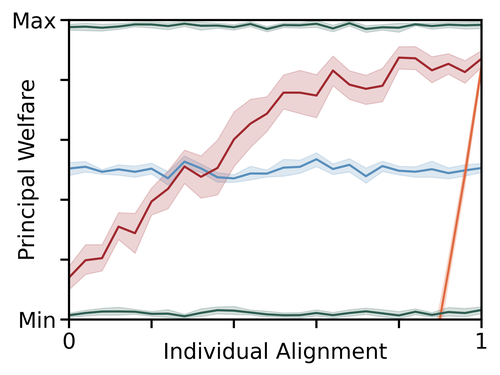}
    \end{subfigure}
    \begin{subfigure}{0.24\textwidth}
            \centering
            \includegraphics[width=0.95\textwidth]{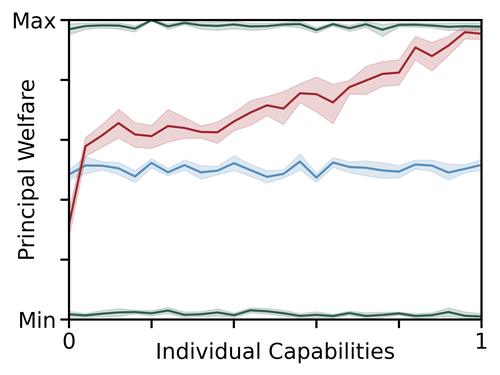}
    \end{subfigure}
    \begin{subfigure}{0.24\textwidth}
            \centering
            \includegraphics[width=0.95\textwidth]{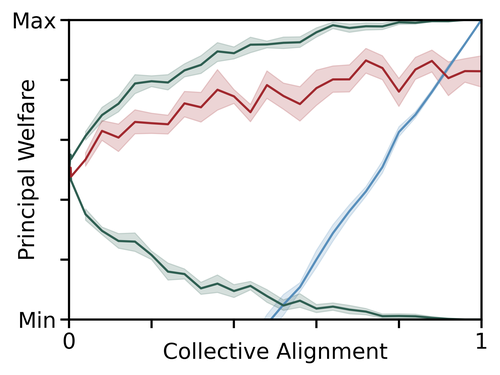}
    \end{subfigure}
    \begin{subfigure}{0.24\textwidth}
            \centering
            \includegraphics[width=0.95\textwidth]{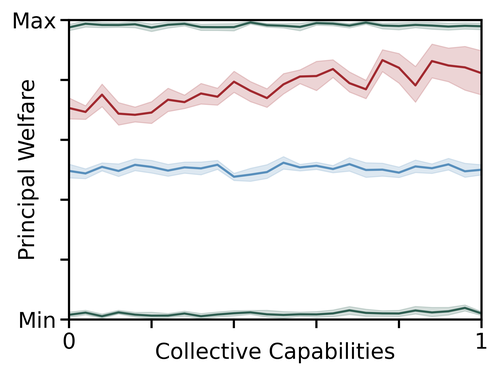}
    \end{subfigure}

    \begin{subfigure}{0.24\textwidth}
            \centering
            \includegraphics[width=0.95\textwidth]{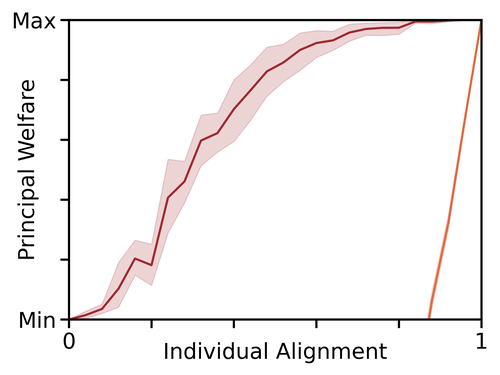}
    \end{subfigure}
    \begin{subfigure}{0.24\textwidth}
            \centering
            \includegraphics[width=0.95\textwidth]{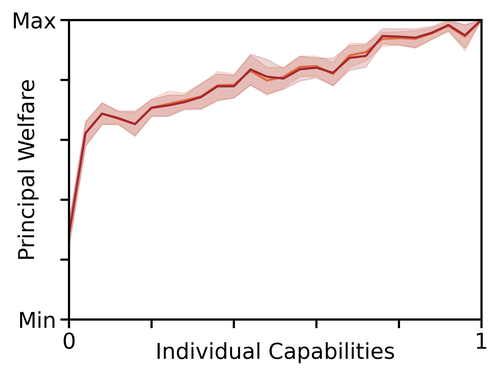}
    \end{subfigure}
    \begin{subfigure}{0.24\textwidth}
            \centering
            \includegraphics[width=0.95\textwidth]{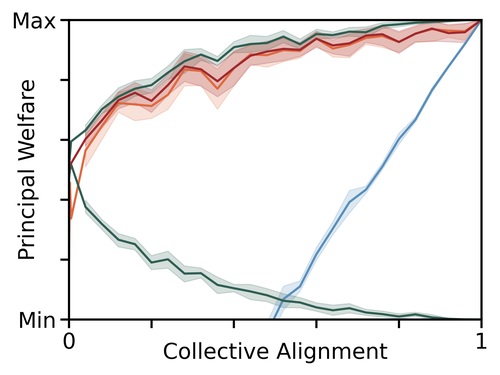}
    \end{subfigure}
    \begin{subfigure}{0.24\textwidth}
            \centering
            \includegraphics[width=0.95\textwidth]{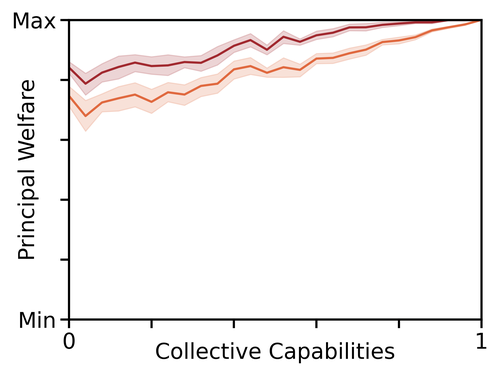}
    \end{subfigure}
    \caption{Each of these plots was generated from games with approximately $10^2$ outcomes. The key for this figure is identical to the related figure in the main body. Each row corresponds to setting (all of) the fixed variables to one of the five values $v \in \{0,0.25,0.5,0.75,1 \}$, respectively.}
    \label{fig:extra_plots_100}
\end{figure*}

\begin{figure*}[h]
    \centering
    \begin{subfigure}{0.24\textwidth}
            \centering
            \includegraphics[width=0.95\textwidth]{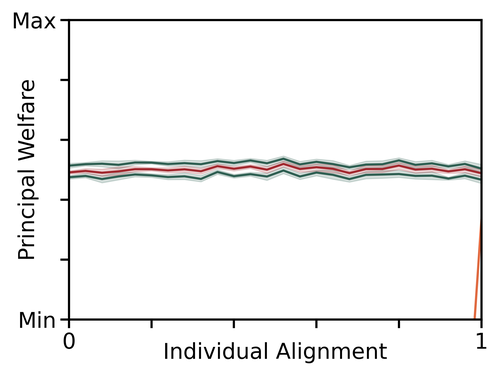}
    \end{subfigure}
    \begin{subfigure}{0.24\textwidth}
            \centering
            \includegraphics[width=0.95\textwidth]{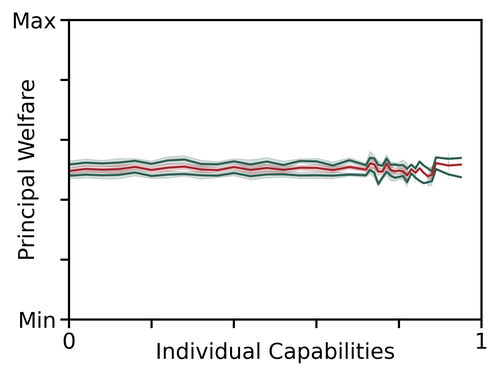}
    \end{subfigure}
    \begin{subfigure}{0.24\textwidth}
            \centering
            \includegraphics[width=0.95\textwidth]{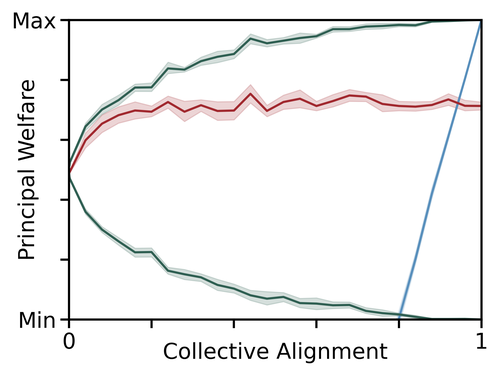}
    \end{subfigure}
    \begin{subfigure}{0.24\textwidth}
            \centering
            \includegraphics[width=0.95\textwidth]{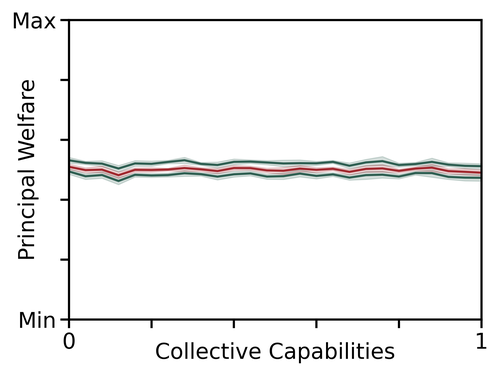}
    \end{subfigure}

    \begin{subfigure}{0.24\textwidth}
            \centering
            \includegraphics[width=0.95\textwidth]{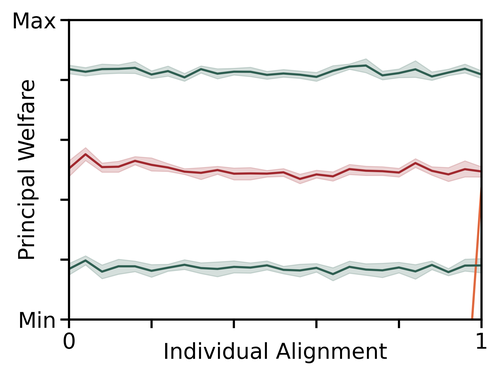}
    \end{subfigure}
    \begin{subfigure}{0.24\textwidth}
            \centering
            \includegraphics[width=0.95\textwidth]{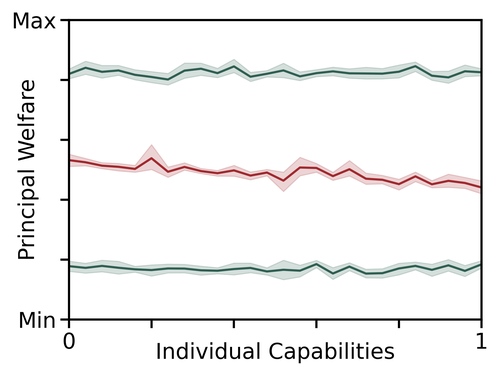}
    \end{subfigure}
    \begin{subfigure}{0.24\textwidth}
            \centering
            \includegraphics[width=0.95\textwidth]{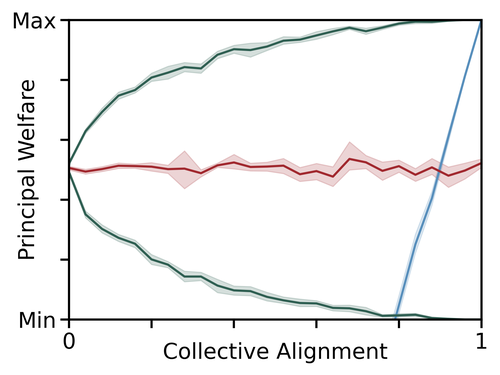}
    \end{subfigure}
    \begin{subfigure}{0.24\textwidth}
            \centering
            \includegraphics[width=0.95\textwidth]{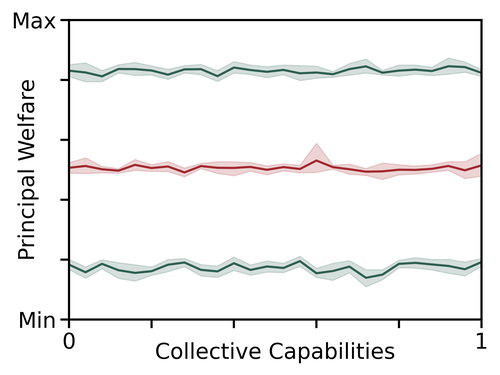}
    \end{subfigure}

    \begin{subfigure}{0.24\textwidth}
            \centering
            \includegraphics[width=0.95\textwidth]{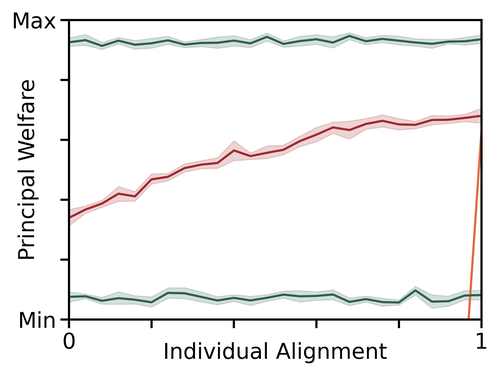}
    \end{subfigure}
    \begin{subfigure}{0.24\textwidth}
            \centering
            \includegraphics[width=0.95\textwidth]{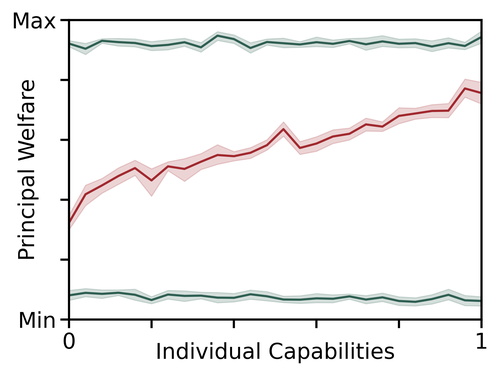}
    \end{subfigure}
    \begin{subfigure}{0.24\textwidth}
            \centering
            \includegraphics[width=0.95\textwidth]{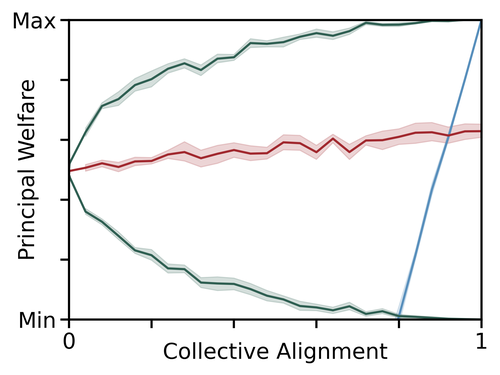}
    \end{subfigure}
    \begin{subfigure}{0.24\textwidth}
            \centering
            \includegraphics[width=0.95\textwidth]{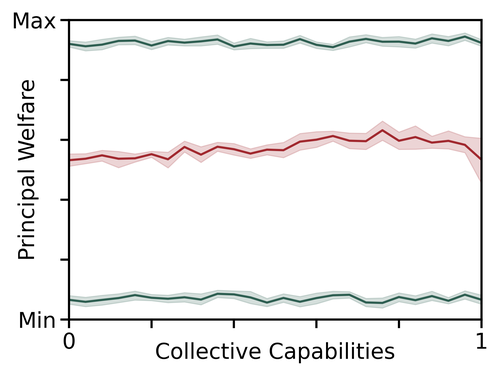}
    \end{subfigure}

    \begin{subfigure}{0.24\textwidth}
            \centering
            \includegraphics[width=0.95\textwidth]{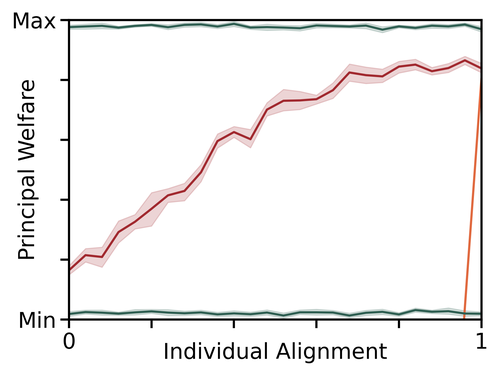}
    \end{subfigure}
    \begin{subfigure}{0.24\textwidth}
            \centering
            \includegraphics[width=0.95\textwidth]{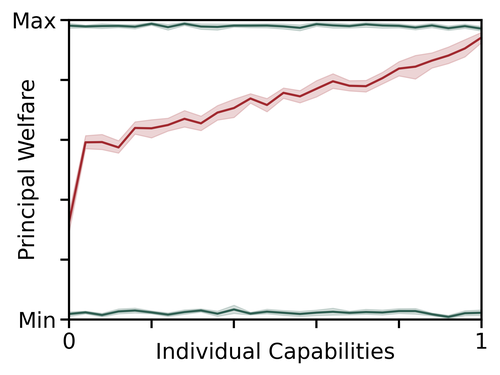}
    \end{subfigure}
    \begin{subfigure}{0.24\textwidth}
            \centering
            \includegraphics[width=0.95\textwidth]{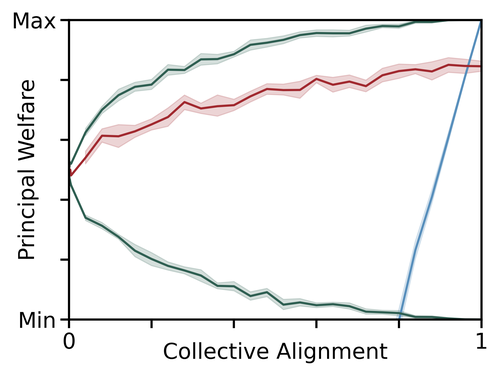}
    \end{subfigure}
    \begin{subfigure}{0.24\textwidth}
            \centering
            \includegraphics[width=0.95\textwidth]{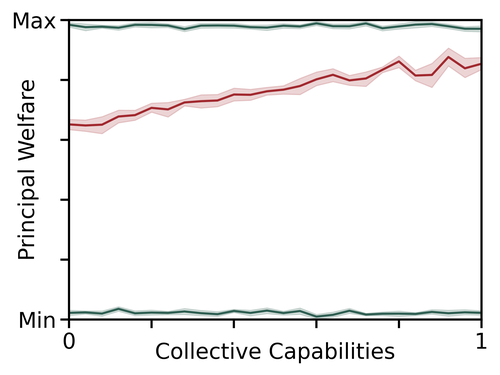}
    \end{subfigure}

    \begin{subfigure}{0.24\textwidth}
            \centering
            \includegraphics[width=0.95\textwidth]{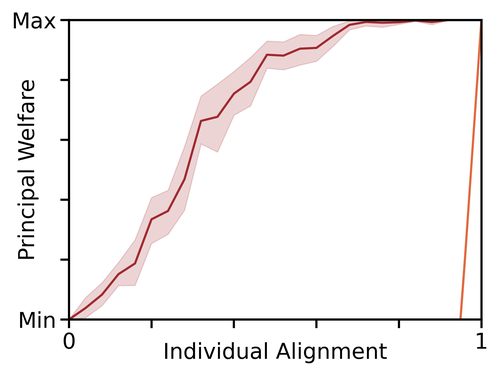}
    \end{subfigure}
    \begin{subfigure}{0.24\textwidth}
            \centering
            \includegraphics[width=0.95\textwidth]{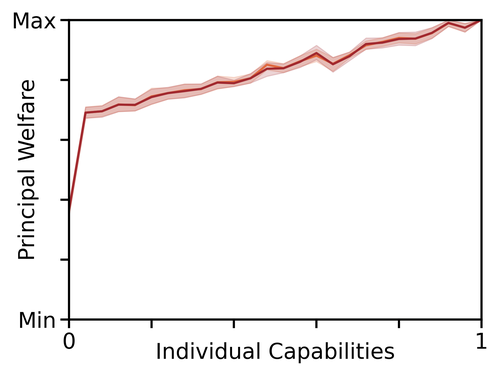}
    \end{subfigure}
    \begin{subfigure}{0.24\textwidth}
            \centering
            \includegraphics[width=0.95\textwidth]{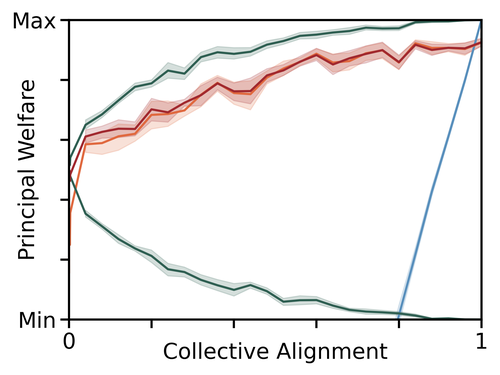}
    \end{subfigure}
    \begin{subfigure}{0.24\textwidth}
            \centering
            \includegraphics[width=0.95\textwidth]{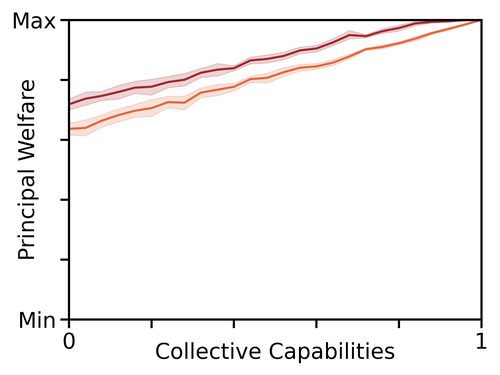}
    \end{subfigure}
    \caption{Each of these plots was generated from games with approximately $10^3$ outcomes. The key for this figure is identical to the related figure in the main body. Each row corresponds to setting (all of) the fixed variables to one of the five values $v \in \{0,0.25,0.5,0.75,1 \}$, respectively.}
    \label{fig:extra_plots_1000}
\end{figure*}

\end{document}